\documentclass[a4paper,fleqn]{cas-sc}

\usepackage{subfigure} 
\usepackage{graphicx}
\usepackage{booktabs}
\usepackage{enumitem}
\usepackage{float}
\usepackage[section]{placeins}
\usepackage{tikz}
\usetikzlibrary{arrows.meta, positioning, shapes.geometric, fit}

\usepackage{amsthm}
\usepackage[numbers]{natbib}

\usepackage{algpseudocode}
\usepackage{algorithm}

\usepackage[english]{babel}

\newtheorem{theorem}{Theorem}[section]
\newtheorem{definition}{Definition}[section]

\newtheorem{corollary}{Corollary}[theorem]

\def\tsc#1{\csdef{#1}{\textsc{\lowercase{#1}}\xspace}}
\tsc{WGM}
\tsc{QE}
\tsc{EP}
\tsc{PMS}
\tsc{BEC}
\tsc{DE}
\ExplSyntaxOn
\cs_gset:Npn \__first_footerline:
{
  \group_begin:
  \small
  \sffamily
  \ifnum\theblind>0\relax
  \else
	\__short_authors:
  \fi
  \group_end:
}
\ExplSyntaxOff

\begin{document}
	\let\WriteBookmarks\relax
	\def\floatpagepagefraction{1}
	\def\textpagefraction{.001}
	\shorttitle{SCV-PINN: A Novel Framework for Forward and Inverse Problems}
	\shortauthors{Barman, Ray, Chatterjee}

	\title [mode = title]{Split Complex-Valued Physics-Informed Neural Networks for Forward and Inverse Nonlinear Partial Differential Equations}                      
	
	\author[1]{Biswanath Barman}
    \ead{biswa.nit.rourkela@gmail.com}

        \author[1]{Rajendra K. Ray$^{*}$} 
        \cortext[cor1]{Corresponding author}
        \ead{rajendra@iitmandi.ac.in}
       
        \author[2]{Debdeep Chatterjee}  
        
        \address[1]{School of Mathematical and Statistical Sciences, Indian Institute of Technology Mandi\\
        Mandi, Himachal Pradesh, 175005, India}
        \address[2]{Department of Computer Science \& Engineering, Sikkim Manipal Institute of Technology, \\
        Rangpo, Sikkim, 737136, India}

\begin{abstract}
Physics-informed neural networks (PINNs) have emerged as an effective framework for solving forward and inverse partial differential equations (PDEs). However, conventional real-valued PINNs can exhibit spectral bias and limited capability in representing high-frequency, oscillatory, and phase-dependent features arising in complex nonlinear regimes. In this work, we propose a generalized split complex-valued physics-informed neural network (SCV-PINN) framework that incorporates complex-valued network parameters and latent representations. Split complex-valued activation functions operate separately on the real and imaginary components, allowing the network to represent coupled amplitude--phase features while maintaining computational efficiency. The framework is systematically examined through ablation studies involving different split activation functions and collocation-point sampling strategies. Extensive numerical experiments are conducted for forward and inverse problems involving Burgers, Allen--Cahn, Korteweg--de Vries, nonlinear Schrödinger, Helmholtz, and Poisson equations on regular and irregular domains, together with Kovasznay flow at $\mathrm{Re}=20$, lid-driven cavity flow at $\mathrm{Re}=100$, the Lorenz system, inverse Burgers, and inverse Navier--Stokes problems at $\mathrm{Re}=100$. Its applicability to high-dimensional nonlinear systems is further demonstrated using a three-dimensional Navier--Stokes Beltrami flow, achieving a relative $L_2$ error of $4.07\times10^{-5}$. Across the considered problems, SCV-PINNs provide high-fidelity solution predictions and effective parameter identification, with improved performance over conventional real-valued PINNs and selected existing variants. These results demonstrate the potential of complex-valued representations for mitigating spectral-bias-related limitations in nonlinear, oscillatory, and high-dimensional physical systems.
\end{abstract}


	\date{}
	\maketitle

\section{Introduction}

Many physical phenomena are governed by nonlinear interactions that generate oscillatory structures, sharp spatial variations, and complex spatiotemporal dynamics. Wave propagation, fluid motion, and nonlinear transport provide representative examples in which resolving the governing physical behavior requires a computational model capable of capturing multiple interacting solution features. Neural networks offer a flexible approach to approximating such systems, but their effectiveness depends strongly on how the underlying physical information is represented and learned. Conventional real-valued neural networks (RVNNs), including multilayer perceptrons, recurrent networks, and convolutional architectures, have become established computational tools owing to their mathematical foundations, efficient implementations, and well-developed optimization techniques \cite{goodfellow2016deep,bishop2006pattern,bishop2023deep,aggarwal2018neural,strang2019linear,zhang2016understanding}. Nevertheless, physical problems involving oscillations, wave interference, and phase-dependent interactions motivate the exploration of alternative representations that can explicitly accommodate coupled amplitude and phase information. Complex-valued neural networks (CVNNs), which extend trainable parameters and intermediate representations from the real field $\mathbb{R}$ to the complex field $\mathbb{C}$, provide one such possibility \cite{krantz1999handbook,ponnusamy2006complex,zill2009first,suresh2013supervised,guberman2016complex,barrachina2023theory,abdalla2023complex}.

Complex-valued representations have been investigated in signal processing, imaging, and learning problems where phase relationships contain information beyond signal magnitude alone \cite{lee2022complex,bassey2021survey,fuchs2021complex,scarnati2021complex,choi2018phase,fink2014predictin}. Several studies have reported favorable representation, generalization, and optimization properties for complex-valued architectures in particular learning settings \cite{barrachina2023theory,trabelsi2017deep}. Their motivation is also related to descriptions of neuronal activity involving oscillatory firing and temporal synchronization, where amplitude and phase provide complementary interpretations of signal behavior \cite{lee2022complex,bassey2021survey,hirose2006complex,reichert2013neuronal}. Mathematically, complex-valued multiplication introduces coupling between real and imaginary components, allowing both magnitude and phase relationships to influence feature propagation. Early investigations of complex-valued neurons further demonstrated their ability to represent nonlinear classification structures, including the XOR problem, through decision boundaries distinct from those of individual real-valued neurons \cite{nitta2003solving,nitta2003inherent,nitta2000analysis}. These characteristics suggest that complex-valued parameterizations may provide useful representational flexibility for scientific problems involving oscillatory and coupled physical fields.

Despite these advantages, complex-valued learning introduces additional theoretical and computational considerations. Liouville's theorem establishes that every bounded entire holomorphic function is constant, restricting the construction of globally bounded holomorphic activation functions \cite{ponnusamy2006complex,tanaka2013complex}. The development of Wirtinger calculus provides a mathematical basis for differentiating non-holomorphic complex-valued mappings and extending gradient-based optimization to complex parameters \cite{fischer2002appendix,wirtinger1927formalen,amin2013learning,amin2011wirtinger}. Nevertheless, the choice of nonlinear activation functions and the implementation of complex-valued differentiation remain important practical considerations. Split complex-valued neural networks provide an alternative formulation in which complex-valued weights and biases are retained, while real-valued activation functions are applied separately to the real and imaginary components \cite{hammad2024comprehensive}. This construction avoids imposing holomorphicity on the activation functions and allows familiar nonlinear transformations to be incorporated into complex-valued architectures. Its compatibility with standard real-valued automatic differentiation and optimization tools makes it particularly relevant to scientific machine learning.

In parallel, physics-informed neural networks (PINNs) have established a computational framework for approximating solutions of differential equations by incorporating governing physical laws into neural-network training \cite{raissi2019physics,karniadakis2021physics,cuomo2022scientific}. Unlike purely data-driven models, PINNs combine available observations with differential equation residuals and initial or boundary conditions, enabling the treatment of both forward solution problems and inverse parameter identification. Their applications span fluid mechanics, biomedical systems, uncertainty quantification, high-dimensional differential equations, stochastic processes, and fractional models \cite{raissi2020hidden,sun2020surrogate,raissi2019deep,jin2021nsfnets,barman2026efficient,kissas2020machine,liu2019multi,yang2019adversarial,zhu2019physics,yang2021b,han2018solving,pang2019fpinns}. However, conventional PINNs based on real-valued fully connected networks can encounter difficulties when approximating strongly nonlinear, multiscale, or highly oscillatory solutions \cite{fuks2020limitations,krishnapriyan2021characterizing,wang2021understanding,wang2022and}. These difficulties are associated with competing loss components, stiffness in optimization dynamics, and the tendency of neural networks to prioritize low-frequency solution components during training \cite{wang2021understanding,rahaman2019spectral,cao2019towards,tancik2020fourier,basri2020frequency,jacot2018neural}. Consequently, improving both the representation and optimization of physically constrained solutions remains an active area of computational research.

Considerable progress has been made toward improving PINN performance through architectural and optimization strategies. Adaptive loss-weighting approaches seek to balance the contributions of individual training objectives \cite{mcclenny2023self,song2024loss}, while sequential learning methods improve the treatment of time-dependent physical systems \cite{mattey2022novel,sundar2025sequential}. Transformer-based physics-informed architectures have introduced alternative feature-interaction mechanisms for learning complex solution structures \cite{zhao2023pinnsformer,zhu2026physicssolver,barman2026physicsformer}. These methods provide valuable advances and address different aspects of physics-informed optimization and representation. Nevertheless, a complementary question remains: whether the complex domain itself can provide additional representational flexibility for physics-informed learning, particularly when physical solutions exhibit oscillatory behavior, nonlinear coupling, and phase-dependent interactions. Investigating this question requires a formulation that retains the advantages of complex-valued representations while remaining computationally practical across different classes of differential equations.

Recent studies have explored complex-valued PINNs, complex-analysis-inspired activation functions, and specialized complex-valued learning formulations for physical systems \cite{si2026complex,mohuț2026towards,zhang2025complex}. These developments demonstrate the relevance of complex-domain learning to scientific computation. However, the performance of split complex-valued architectures across a broad collection of real- and complex-valued governing equations, forward and inverse problems, and regular and irregular computational geometries warrants further systematic investigation. In particular, the effects of split activation functions, collocation-point distributions, and optimization procedures need to be examined within a consistent computational setting. Such an assessment is essential for distinguishing the benefits of complex-valued representations from those arising from training configurations or problem-specific choices. This motivates the development and evaluation of a generalized split complex-valued physics-informed framework.

In this work, we develop a generalized \emph{split complex-valued physics-informed neural network} (SCV-PINN) framework for solving forward and inverse nonlinear differential equations. The proposed architecture combines complex-valued trainable weights, biases, and intermediate representations with split activation functions that independently transform the real and imaginary components while preserving their coupling through complex-valued operations. For real-valued governing equations, the imaginary component acts as an auxiliary latent representation, whereas complex-valued physical systems are treated through their coupled real and imaginary fields. This construction provides a unified approach for investigating amplitude--phase representations without requiring holomorphic activation functions. Training is performed using a hybrid optimization procedure consisting of Adam \cite{kingma2014adam} followed by limited-memory Broyden--Fletcher--Goldfarb--Shanno (L-BFGS) optimization \cite{liu1989limited}. Since the distribution of collocation points influences the enforcement of governing equations, six non-adaptive sampling strategies---equispaced grid, uniform random, Latin hypercube, Halton, Hammersley, and Sobol---are examined alongside residual-based adaptive distribution (RAD) sampling \cite{wu2023comprehensive}. Systematic ablation studies are conducted to investigate the effects of split activation functions and sampling strategies on predictive performance and training behavior.

The proposed framework is evaluated through an extensive collection of nonlinear ordinary and partial differential equations, including oscillatory wave problems, reaction--diffusion systems, elliptic equations, incompressible flows, and nonlinear dynamical systems. The numerical investigations include regular and irregular computational domains, sparse-observation inverse problems, and a three-dimensional Navier--Stokes Beltrami flow benchmark. Comparisons with conventional real-valued PINNs and selected existing variants examine solution errors, parameter identification, and convergence characteristics. The results demonstrate the potential of split complex-valued representations to improve physics-informed approximation across the investigated configurations, while retaining the governing physical constraints. Rather than replacing existing PINN methodologies, SCV-PINN provides a complementary computational formulation that extends the available representation choices for nonlinear, oscillatory, and higher-dimensional physical systems.

\section*{Main Contributions}

The main contributions of this work are summarized as follows:

\begin{itemize}

\item A generalized SCV-PINN framework is developed for forward and inverse nonlinear differential equations by integrating complex-valued trainable parameters and latent representations with split activation functions in a unified physics-informed formulation.

\item A practical computational implementation is established using split complex-valued operations and hybrid Adam--L-BFGS optimization, enabling the treatment of real- and complex-valued physical systems within a consistent training framework.

\item Systematic ablation studies examine the influence of split activation functions and seven collocation-point sampling strategies, including six non-adaptive approaches and residual-based adaptive distribution (RAD) sampling.

\item The framework is comprehensively evaluated across nonlinear wave equations, reaction--diffusion systems, elliptic problems, irregular geometries, incompressible flows, and forward and inverse parameter-identification problems, including three-dimensional Navier--Stokes Beltrami flow.

\item Numerical comparisons demonstrate lower prediction errors, improved parameter identification, and favorable convergence behavior relative to conventional real-valued PINNs and selected existing variants across the considered benchmark configurations, highlighting the computational potential of split complex-valued representations.

\end{itemize}

The remainder of this paper is organized as follows. Section~\ref{sec:sec3} introduces the fundamentals of complex-valued neural networks (CVNNs), and Section~\ref{sec:sec4} describes the proposed SCV-PINN formulation and computational methodology. Section~\ref{sec:sec5} presents the numerical experiments and comparative analyses, while Section~\ref{sec:sec6} concludes the study. Additional ablation studies on split activation functions and collocation-point sampling strategies are provided in Appendix~\ref{appen:appenA} and Appendix~\ref{appen:appenB}, respectively.

\section{Preliminary}
\label{sec:sec3}
\subsection{Complex-valued Neural Networks (CVNNs)}
The early development of complex-valued neural networks (CVNNs) was hindered by Liouville's theorem, since a bounded function that is complex differentiable over the entire complex plane must be constant, whereas the loss functions required for learning are generally non-holomorphic. Around the 1990s, it was argued that this could prevent the proper analysis and training of CVNNs~\cite{tanaka2013complex}. The difficulty was later resolved by Wirtinger calculus~\cite{amin2013learning, fischer2002appendix}, which is now widely used for training complex-valued models. The relevant results are briefly recalled below.

\subsubsection{Liouville's Theorem}

\begin{definition}[Analytic Function]
\label{def:analytic_function}
A complex function $w = f(z)$ is said to be \emph{analytic at a point} $z_0$ if it is differentiable at $z_0$ and at every point in some neighborhood of $z_0$.
\end{definition}

\begin{definition}[Entire Function]
\label{def:entire_function_one}
A function $f:\mathbb{C} \rightarrow \mathbb{C}$ is called \emph{entire} if it is analytic (holomorphic) at every point in the complex plane $\mathbb{C}$.
\end{definition}

\begin{definition}
\label{def:bounded_function}
A function $f$ is bounded if there exists a constant $M \in \mathbb{R}^{+}$ such that $|f(z)| < M$ for all $z \in \mathbb{C}$.
\end{definition}

\begin{theorem}[Cauchy Integral Theorem]
\label{thm:cauchy_integral_theorem}
If $f$ is analytic in a region $\mathcal{D}$, then for any closed contour $C$ contained in $\mathcal{D}$,
\begin{equation}
\oint_C f(z)\, dz = 0.
\label{eq:cauchy_theorem}
\end{equation}
\end{theorem}

\begin{theorem}[Cauchy Integral Formula]
\label{thm:cauchy_integral_formula}
If $f$ is analytic on and inside a closed contour $C$, and $z_0$ lies inside $C$, then
\begin{equation}
f(z_0)=\frac{1}{2\pi i}\oint_C \frac{f(z)}{z-z_0}\,dz.
\label{eq:cauchy_formula}
\end{equation}
\end{theorem}

\begin{corollary}
\label{cor:cauchy_derivative_formula}
The higher-order derivatives of $f$ are given by
\begin{equation}
f^{(n)}(z_0)=\frac{n!}{2\pi i}\oint_C \frac{f(z)}{(z-z_0)^{n+1}}\,dz.
\label{eq:cauchy_derivative}
\end{equation}
\end{corollary}

\begin{theorem}[Liouville's Theorem \cite{ponnusamy2006complex}]
\label{thm:liouville}
If a function $f:\mathbb{C} \to \mathbb{C}$ is entire and bounded, then $f$ must be constant. That is,
\begin{equation}
\left( f \text{ is entire and } \exists M>0 \text{ such that } |f(z)|<M,\; \forall z \in \mathbb{C} \right) \Rightarrow f(z)=c,\; c \in \mathbb{C}.
\label{eq:liouville}
\end{equation}
\end{theorem}

\begin{proof}
Since $f$ is entire, it admits the Taylor expansion $f(z)=\sum_{k=0}^{\infty} a_k z^k$. Using the Cauchy integral formula~\ref{thm:cauchy_integral_formula} on a circle $C_r$ of radius $r>0$ centered at the origin, and $|f(z)|\leq M$,
\begin{equation}
|a_k| = \left|\frac{1}{2\pi i} \oint_{C_r} \frac{f(z)}{z^{k+1}}\, dz\right|
\leq \frac{1}{2\pi}\,\frac{M}{r^{k+1}}\,(2\pi r) = \frac{M}{r^k}.
\label{eq:cauchy_ineq}
\end{equation}
Since $r$ is arbitrary, letting $r \to \infty$ yields $a_k = 0$ for all $k \geq 1$. Hence, $f(z)=a_0$ is constant.
\end{proof}

\begin{theorem}[Generalized Liouville's Theorem \cite{ponnusamy2006complex}]
\label{thm:generalized_liouville}
A non-constant entire function comes arbitrarily close to every complex number.
\end{theorem}

\begin{proof}
Suppose that $f$ is entire and that $|f(z) - z_0| \geq \epsilon$ for some $z_0\in\mathbb{C}$, $\epsilon > 0$, and all $z \in \mathbb{C}$. Then $g(z) = 1/(f(z) - z_0)$ is entire and $|g(z)| \leq 1/\epsilon$. By Liouville's theorem~\ref{thm:liouville}, $g$ is constant, and hence $f(z) = 1/g(z) + z_0$ is also constant.
\end{proof}

\begin{theorem}
\label{thm:harmonic_bounded_constant}
A function harmonic and bounded in $\mathbb{C}$ must be a constant.
\end{theorem}

\begin{theorem}
\label{thm:real_imaginary_bounded}
If the real or imaginary part of an entire function is bounded above or below by a real number $M$, then the function is a constant.
\end{theorem}

\begin{theorem}[Mean Value Property]
\label{thm:mean_value_property}
Suppose $u(z)$ is harmonic in a domain containing the disk $|z-z_0|\leq R$. Then
\begin{equation}
u(z_0)=\frac{1}{2\pi}\int_{0}^{2\pi} u\left(z_0+Re^{i\theta}\right)\, d\theta.
\end{equation}
\end{theorem}

\begin{theorem}[Maximum Principle for Harmonic Functions \cite{ponnusamy2006complex}]
\label{thm:maximum_principle_harmonic}
Let $u(z)$ be harmonic in a domain $D$. If $u(z)$ is nonconstant, then it cannot attain either a maximum or a minimum value at any interior point of $D$.
\end{theorem}

\begin{proof}
It suffices to prove the maximum case, since the minimum case follows by considering $-u(z)$. Assume that $u$ attains its maximum at an interior point $z_0 \in D$, and let $0<r\leq \mathrm{dist}(z_0,\partial D)$. By the Mean Value Property~\ref{thm:mean_value_property},
\begin{equation}
\frac{1}{2\pi}\int_{0}^{2\pi}
\left[u(z_0)-u(z_0+re^{i\theta})\right]d\theta =0.
\end{equation}
Because the integrand is continuous and nonnegative, it vanishes identically, i.e., $u(z_0+re^{i\theta})=u(z_0)$ for $0\leq\theta\leq2\pi$. Thus, $u$ is constant near $z_0$ and, by the connectedness of $D$ and the uniqueness property of harmonic functions, constant throughout $D$, which contradicts the assumption.
\end{proof}

\subsubsection{Wirtinger Calculus for Complex-Valued Optimization}

Wirtinger calculus extends differentiation from holomorphic functions to general complex-valued mappings and plays an important role in the optimization of complex-valued neural networks. For $z=x+iy$ and a complex-valued function $f(z)$, the Wirtinger derivatives with respect to $z$ and its conjugate $\bar{z}$ are defined as \cite{barrachina2023theory}
\begin{equation}
\frac{\partial f}{\partial z}
=
\frac{1}{2}
\left(
\frac{\partial f}{\partial x}
-
i\frac{\partial f}{\partial y}
\right),
\qquad
\frac{\partial f}{\partial \bar{z}}
=
\frac{1}{2}
\left(
\frac{\partial f}{\partial x}
+
i\frac{\partial f}{\partial y}
\right),
\label{eq:wirtinger_derivative}
\end{equation}
and, conversely,
\begin{equation}
\frac{\partial f}{\partial x}
=
\frac{\partial f}{\partial z}
+
\frac{\partial f}{\partial \bar{z}},
\qquad
\frac{\partial f}{\partial y}
=
i
\left(
\frac{\partial f}{\partial z}
-
\frac{\partial f}{\partial \bar{z}}
\right).
\label{eq:real_imaginary_derivatives}
\end{equation}
For holomorphic functions satisfying the Cauchy--Riemann conditions, $\partial f/\partial \bar{z} = 0$. In the SCV-PINN framework, this formulation provides a consistent treatment of split complex-valued representations that is compatible with standard automatic differentiation and optimization procedures.

\subsubsection{Classification of Complex-Valued Neural Networks}

CVNNs can be categorized into fully complex-valued and split complex-valued neural networks~\cite{lee2022complex, yang2007sensitivity}. In fully complex-valued networks, both the parameters and the activation functions are defined in the complex domain, which preserves amplitude and phase information but leads to high computational complexity and difficulties in constructing holomorphic activation functions~\cite{georgiou1992complex, hirose1992proposal}. Split complex-valued networks instead decompose complex quantities either into real and imaginary parts, $\{z_1, z_2\} = \{(x_1, y_1), (x_2, y_2)\}$, or into magnitude and phase, $\{z_1, z_2\} = \{(r_1, \phi_1), (r_2, \phi_2)\}$. They are further divided into (i) networks with real-valued weights, where complex inputs are converted into real-valued vectors of larger dimension and the amplitude--phase coupling may be weakened, and (ii) networks with complex-valued weights and real-valued activation functions, where the split activation
\begin{equation}
\sigma_{\mathbb{C}}(z)=\sigma(\Re[z]) + i\,\sigma(\Im[z])
\label{eq:split_activation}
\end{equation}
is applied with a conventional real-valued activation $\sigma(\cdot)$, such as ReLU, sigmoid, or hyperbolic tangent. In the latter case, the complex product $z_1z_2=(x_1x_2-y_1y_2)+i(x_1y_2+x_2y_1)=r_1r_2e^{i(\phi_1+\phi_2)}$ couples the real and imaginary parts, i.e., amplitudes are multiplied and phases are added~\cite{hirose2006complex}. This formulation retains these complex-valued characteristics while maintaining numerical stability and computational efficiency, and it is therefore adopted in the proposed SCV-PINN framework.

\subsection{Various Types of Complex Activation Functions}

The design of activation functions in complex-valued neural networks (CVNNs) is considerably more challenging than in real-valued networks due to restrictions imposed by complex analysis. In particular, bounded holomorphic functions are severely constrained by Liouville's theorem, motivating the use of non-holomorphic and split activation formulations in practical applications.

\subsubsection{Split Activation Functions}

For a complex-valued input $z=x+iy=|z|e^{i\arg(z)}$, split activation functions independently process the real and imaginary components. Two common formulations are
\begin{align}
f_A(z) &= f_{R}(x) + i f_{I}(y), \label{eq:typeA} \\
f_B(z) &= f_r(|z|) + i f_{\phi}(\arg(z)), \label{eq:typeB}
\end{align}
where $f_R, f_I, f_r,$ and $f_{\phi}$ denote standard real-valued activation functions. Type-A activations operate in Cartesian coordinates, whereas Type-B activations separately treat magnitude and phase information, making them suitable for oscillatory problems.

\subsubsection{Fully Complex Activation Functions}

Fully complex activation functions treat the complex variable as a single entity and are generally required to satisfy holomorphicity conditions. A representative example is the multi-valued neuron (MVN),
\begin{equation}
f(z)=\exp(i\arg(z))=\frac{z}{|z|},
\label{eq:mvn}
\end{equation}
while elementary functions such as $\sin(z)$ and $\tanh(z)$ may also be employed. However, these formulations often suffer from increased computational complexity and optimization difficulties.

\subsubsection{Complex ReLU Variants}

Several ReLU-based complex activation functions have been proposed for stable optimization. The $z$-ReLU activation is defined as
\begin{equation}
z\text{-ReLU}(z)=
\begin{cases}
z, & \text{if } 0<\arg(z)<\frac{\pi}{2}, \\
0, & \text{otherwise},
\end{cases}
\label{eq:zrelu}
\end{equation}
whereas the modReLU activation introduces magnitude thresholding:
\begin{equation}
\text{modReLU}(z)=
\begin{cases}
\text{ReLU}(|z|+b)\dfrac{z}{|z|}, & |z|\geq b, \\
0, & \text{otherwise},
\end{cases}
\label{eq:modrelu}
\end{equation}
with $b$ denoting a learnable bias parameter. Another commonly used variant is the cardioid activation function,
\begin{equation}
f(z)=\frac{1+\cos(\arg(z))}{2}z.
\label{eq:cardioid}
\end{equation}

\subsubsection{Output Layer Activation}

Although the internal representations of the SCV-PINN are complex-valued, the governing equations considered in this work admit real-valued physical solutions. Therefore, for a network output $z=x+iy$, the physical prediction is taken as $u_{\mathrm{pred}}=\Re(z)=x$, while the imaginary component acts as an auxiliary latent representation that enriches the approximation space.

\section{Methodology}
\label{sec:sec4}
\subsection{Real-Valued Multi-Layer Perceptrons (RV-MLP)}

In conventional real-valued PINNs (RV-PINNs), the solution is approximated by a real-valued multi-layer perceptron (RV-MLP), which maps the input vector $\mathbf{x} \in \mathbb{R}^{d_0}$ of spatial and/or temporal coordinates to the physical quantities through a parametric mapping $\mathcal{F}_{\theta}(\mathbf{x})$, where $\theta$ denotes all trainable parameters. An $L$-layer RV-MLP is recursively defined as
\begin{align}
\mathbf{g}^{(0)}(\mathbf{x}) &= \mathbf{x}, \\
\mathbf{h}^{(l)}(\mathbf{x}) &= \mathbf{W}^{(l)} \mathbf{g}^{(l-1)}(\mathbf{x}) + \mathbf{b}^{(l)}, \\
\mathbf{g}^{(l)}(\mathbf{x}) &= \sigma\big(\mathbf{h}^{(l)}(\mathbf{x})\big), \qquad l = 1,2,\ldots,L,
\end{align}
where $\mathbf{W}^{(l)} \in \mathbb{R}^{d_l \times d_{l-1}}$ and $\mathbf{b}^{(l)} \in \mathbb{R}^{d_l}$ are the weight matrix and bias vector of the $l$-th layer, and $\sigma(\cdot)$ is an element-wise nonlinear activation function. The network output is
\begin{equation}
\mathcal{F}_{\theta}(\mathbf{x}) = \mathbf{W}^{(L+1)} \mathbf{g}^{(L)}(\mathbf{x}) + \mathbf{b}^{(L+1)},
\label{eq:rv_mlp_output}
\end{equation}
with $\theta = \{ \mathbf{W}^{(1)}, \mathbf{b}^{(1)}, \ldots, \mathbf{W}^{(L+1)}, \mathbf{b}^{(L+1)} \} \subset \mathbb{R}$. Hence, RV-MLPs perform magnitude-based transformations only and do not explicitly encode phase information, which may be insufficient for oscillatory, multi-scale, or phase-dependent phenomena.

\subsection{Real-Valued Physics-Informed Neural Networks}

Consider a general nonlinear partial differential equation (PDE) defined on the space--time domain $\Omega \times [0,T]$:
\begin{align}
\mathcal{N}_{x,t}[u(x,t)] &= s(x,t),
\qquad x \in \Omega,\; t\in[0,T],
\label{eq:governing} \\
u(x,0) &= h(x),
\qquad x\in\Omega,
\label{eq:initial} \\
\mathcal{B}_{x,t}[u(x,t)] &= g(x,t),
\qquad x\in\partial\Omega,\; t\in[0,T],
\label{eq:boundary}
\end{align}
where $\mathcal{N}_{x,t}$ is the nonlinear differential operator, $\mathcal{B}_{x,t}$ is the boundary operator, and $s$, $h$, and $g$ are the source, initial, and boundary data, respectively. In the PINN framework, $u(x,t)$ is approximated by a neural network $\hat{u}(x,t;\theta)$, and the governing equation and the initial and boundary conditions are enforced through the loss function. Using $\mathcal{N}_f$ collocation points, $\mathcal{N}_0$ initial points, $\mathcal{N}_b$ boundary points, and $\mathcal{N}_d$ observations, the residual, initial, boundary, and data losses are
\begin{align}
\mathcal{L}_f&=
\frac{1}{\mathcal{N}_f}
\sum_{i=1}^{\mathcal{N}_f}
\left|
\mathcal{N}_{x,t}
\big[
\hat{u}(x_f^i,t_f^i)
\big]
-
s(x_f^i,t_f^i)
\right|^2,
&
\mathcal{L}_0&=
\frac{1}{\mathcal{N}_0}
\sum_{i=1}^{\mathcal{N}_0}
\left|
\hat{u}(x_0^i,0)-h(x_0^i)
\right|^2,
\label{eq:residual_loss}
\\
\mathcal{L}_b&=
\frac{1}{\mathcal{N}_b}
\sum_{i=1}^{\mathcal{N}_b}
\left|
\mathcal{B}_{x,t}
\big[
\hat{u}(x_b^i,t_b^i)
\big]
-
g(x_b^i,t_b^i)
\right|^2,
&
\mathcal{L}_d&=
\frac{1}{\mathcal{N}_d}
\sum_{i=1}^{\mathcal{N}_d}
\left|
\hat{u}(x_d^i,t_d^i)-u_d^i
\right|^2,
\label{eq:boundary_loss}
\end{align}
and the overall loss is
\begin{equation}
\mathcal{L}
=
\lambda_f \mathcal{L}_f
+
\lambda_0 \mathcal{L}_0
+
\lambda_b \mathcal{L}_b
+
\lambda_d \mathcal{L}_d,
\label{eq:total_loss_rvpinn}
\end{equation}
where $\lambda_f$, $\lambda_0$, $\lambda_b$, and $\lambda_d$ are weighting coefficients. All derivatives are computed using automatic differentiation, so that the governing physics is enforced without mesh-based discretization.

\subsection{Complex-Valued Multi-Layer Perceptrons (CV-MLP)}

In the proposed SCV-PINN framework, the solution is approximated by a complex-valued multi-layer perceptron (CV-MLP), which operates in the complex domain and simultaneously encodes magnitude and phase information. For input coordinates $\mathbf{x}\in\mathbb{R}^{d_0}$, an $L$-layer CV-MLP $\mathcal{F}_{\theta}(\mathbf{x}) \in \mathbb{C}^{m}$ is written as
\begin{align}
\mathbf{g}^{(0)}(\mathbf{x}) &= \mathbf{x}, \\
\mathbf{h}^{(l)}(\mathbf{x})
&=
\mathbf{W}^{(l)}
\mathbf{g}^{(l-1)}(\mathbf{x})
+
\mathbf{b}^{(l)}, \\
\mathbf{g}^{(l)}(\mathbf{x})
&=
\sigma_{\mathbb{C}}\big(
\mathbf{h}^{(l)}(\mathbf{x})
\big),
\qquad l=1,\dots,L,
\end{align}
with the output layer
\begin{equation}
\mathcal{F}_{\theta}(\mathbf{x})
=
\mathbf{W}^{(L+1)}
\mathbf{g}^{(L)}
+
\mathbf{b}^{(L+1)},
\label{eq:cv_mlp_output}
\end{equation}
where $\mathbf{W}^{(l)}\in\mathbb{C}^{d_l\times d_{l-1}}$ and $\mathbf{b}^{(l)}\in\mathbb{C}^{d_l}$ are complex-valued weights and biases, so that $\theta=\{\mathbf{W}^{(l)},\mathbf{b}^{(l)}\}\subset \mathbb{C}$. The split activation $\sigma_{\mathbb{C}}(z)=\sigma(\Re(z))+i\,\sigma(\Im(z))$ is used, where $\sigma(\cdot)$ is a real-valued activation function such as Tanh, GELU, or ReLU. Compared with RV-MLPs, CV-MLPs provide richer feature representations for oscillatory and phase-dependent phenomena. The architectural differences between the two frameworks are illustrated in Fig.~\ref{fig:fig2}, where $\mathcal{N}_{\theta}$ denotes the network mapping $\mathcal{F}_{\theta}$ and $f(z)$ denotes the split activation $\sigma_{\mathbb{C}}(z)$.

\refstepcounter{figure}

\begin{center}
\includegraphics[width=\textwidth]{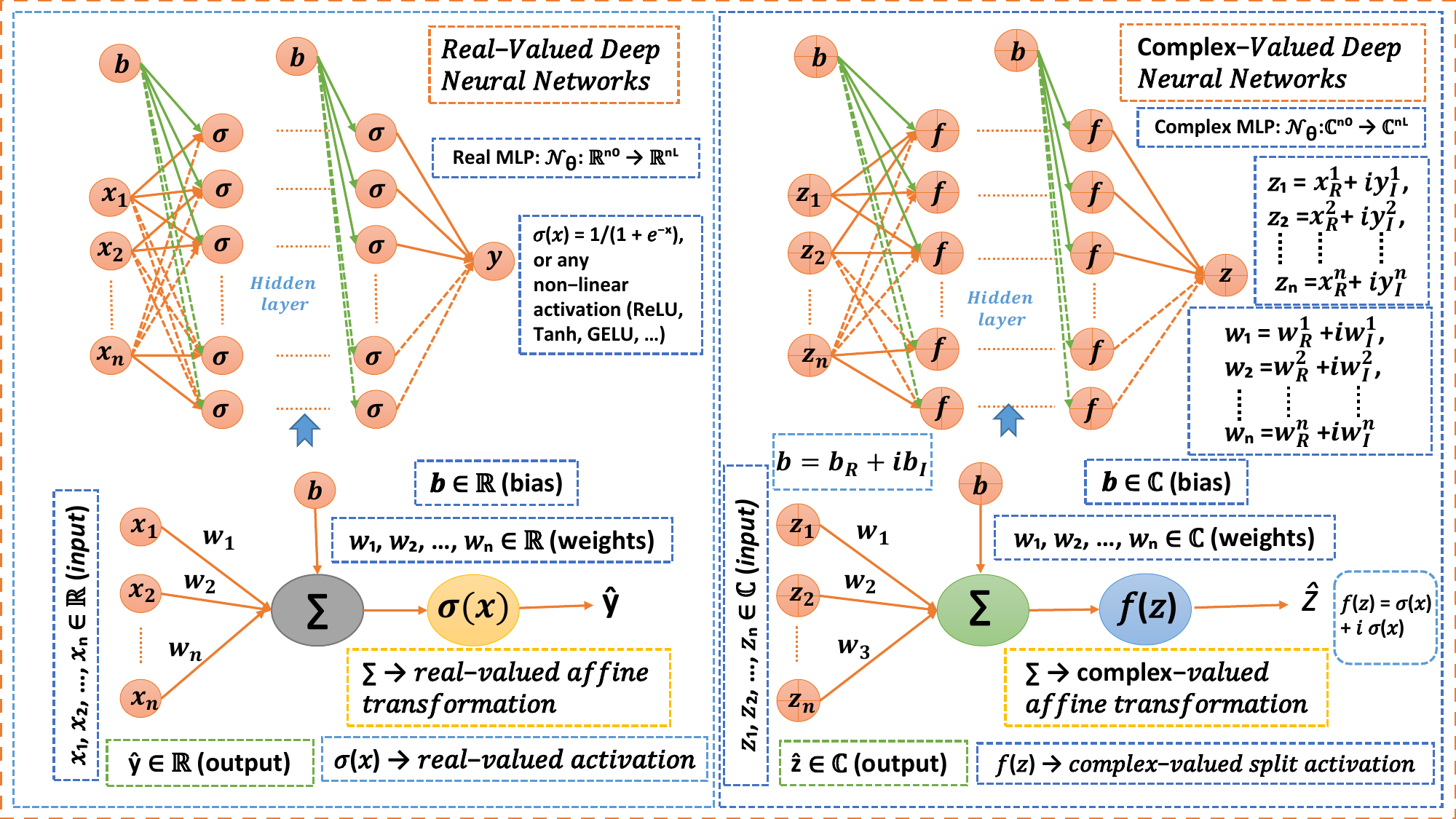}
\end{center}

\vspace{0.5em}

\noindent\textbf{Fig. \thefigure.}
Architectural comparison between conventional real-valued deep neural networks (RV-DNNs) and complex-valued deep neural networks (CV-DNNs), highlighting the enhanced representation capability of complex-valued feature learning.

\label{fig:fig2}

\label{fig:rv_vs_cvmlp}

\subsection{Split Complex-Valued Physics-Informed Neural Networks (SCV-PINN)}

Let $\Omega\subset\mathbb{R}^{d}$ denote the spatial domain and $\mathcal{T}=[0,T]$ the temporal interval. We consider a general nonlinear PDE of the form
\begin{align}
\mathcal{N}_{x,t}[u(x,t);\lambda] &= s(x,t),
\qquad (x,t)\in\Omega\times\mathcal{T},
\label{eq:scvpinn_governing}
\\
u(x,0) &= u_0(x),
\qquad x\in\Omega,
\label{eq:scvpinn_initial}
\\
\mathcal{B}_{x,t}[u(x,t)] &= g(x,t),
\qquad (x,t)\in\partial\Omega\times\mathcal{T},
\label{eq:scvpinn_boundary}
\end{align}
where $u(x,t)\in\mathbb{R}^{m}$ is the unknown physical field, such as velocity, pressure, temperature, or concentration, and $\lambda$ is the set of physical or PDE parameters. Here, $\mathcal{N}_{x,t}$ is the nonlinear differential operator, $\mathcal{B}_{x,t}$ is the boundary operator associated with Dirichlet, Neumann, or mixed conditions, $s(x,t)$ is the source term, $u_0(x)$ is the initial condition, and $g(x,t)$ is the boundary data on $\partial\Omega$. The workflow of the proposed SCV-PINN framework is illustrated in Fig.~\ref{fig:fig3}, and the training procedures are summarized in Algorithms~\ref{alg:scvpinn_unified} and \ref{alg:scvrad}.

\subsubsection{Complex-Valued Neural Representation}

The solution is approximated by a CV-MLP $\mathcal{F}_{\theta}$ acting on the complex input $\mathbf{z}$, and the physical prediction is obtained from its real component:
\begin{equation}
\hat{y}(x,t)
=
\mathcal{F}_{\theta}(\mathbf{z})
=
y_{\mathrm{Re}}
+
i\,y_{\mathrm{Im}}
\in\mathbb{C}^{m},
\qquad
\hat{u}(x,t)
=
\Re
\big[
\hat{y}(x,t)
\big],
\end{equation}
where $\theta\subset\mathbb{C}$ contains all trainable parameters. The imaginary component acts as an auxiliary latent representation that enriches the approximation space during training.

\subsubsection{Input Normalization and Complex Embedding}

Each coordinate is normalized to $[-1,1]$ and embedded into the complex domain as
\begin{equation}
\xi_k
=
\frac{
2(x_k-L_k)
}{
R_k-L_k
}
-1,
\qquad
z_k=\xi_k+i\,0,
\end{equation}
where $L_k$ and $R_k$ are the lower and upper bounds of the $k$-th coordinate, which gives the complex input vector $\mathbf{z}\in\mathbb{C}^{d+1}$ (e.g., $z_1=\xi_x+i\,0$ and $z_2=\xi_t+i\,0$ for $(x,t)$ in Fig.~\ref{fig:fig3}).

\subsubsection{Complex Linear Transformation}

For the $l$-th layer, the complex-valued affine mapping $\mathbf{h}^{(l)}=\mathbf{W}^{(l)}\mathbf{g}^{(l-1)}+\mathbf{b}^{(l)}$, with $\mathbf{W}^{(l)}=\mathbf{W}_R^{(l)}+i\,\mathbf{W}_I^{(l)}$ and $\mathbf{b}^{(l)}=\mathbf{b}_R^{(l)}+i\,\mathbf{b}_I^{(l)}$, is evaluated through its real and imaginary components as
\begin{align}
\Re(\mathbf{h}^{(l)})
&=
\mathbf{W}_R^{(l)}
\Re(\mathbf{g}^{(l-1)})
-
\mathbf{W}_I^{(l)}
\Im(\mathbf{g}^{(l-1)})
+
\mathbf{b}_R^{(l)}, \\
\Im(\mathbf{h}^{(l)})
&=
\mathbf{W}_R^{(l)}
\Im(\mathbf{g}^{(l-1)})
+
\mathbf{W}_I^{(l)}
\Re(\mathbf{g}^{(l-1)})
+
\mathbf{b}_I^{(l)}.
\end{align}

\subsubsection{Split Complex Activation}

A split complex-valued activation is employed:
\begin{equation}
\sigma_{\mathbb{C}}(z)
=
\sigma(\Re[z])
+
i\,\sigma(\Im[z]),
\end{equation}
where the real-valued nonlinearity $\sigma(\cdot)$ is applied independently to the real and imaginary components. The split activations CTanh, CGELU, CSwish, CSigmoid, CSoftplus, CReLU, and $z$-ReLU~\cite{hammad2024comprehensive} are examined in Appendix~\ref{appen:appenA}.

\subsubsection{Network Mapping}

The SCV-PINN mapping follows the CV-MLP recursion with the complex input $\mathbf{z}$:
\begin{equation}
\mathbf{g}^{(0)} = \mathbf{z},
\qquad
\mathbf{g}^{(l)}
=
\sigma_{\mathbb{C}}
\big(
\mathbf{W}^{(l)}
\mathbf{g}^{(l-1)}
+
\mathbf{b}^{(l)}
\big),
\quad l=1,\dots,L,
\qquad
\hat{y}(x,t)
=
\mathbf{W}^{(L+1)}
\mathbf{g}^{(L)}
+
\mathbf{b}^{(L+1)}.
\end{equation}

\subsubsection{Complex PDE Residual}

For a PDE written as $\partial u/\partial t+\mathcal{N}[u]-\mathcal{K}[u]=0$, the SCV-PINN residual is constructed as
\begin{equation}
f(x,t)
=
\frac{\partial \hat{y}}{\partial t}
+
\mathcal{N}
\big[
\Re(\hat{y})
\big]
-
\mathcal{K}
\big[
\hat{y}
\big].
\end{equation}
Here, $\mathcal{N}$ and $\mathcal{K}$ denote the nonlinear and linear parts of the differential operator. Linear operators are applied directly to the complex-valued representation, whereas nonlinear terms are evaluated using the real component to preserve physical consistency. All derivatives are computed by automatic differentiation, i.e., $\partial\hat{y}/\partial x_k=\partial y_{\mathrm{Re}}/\partial x_k+i\,\partial y_{\mathrm{Im}}/\partial x_k$ (split automatic differentiation in Fig.~\ref{fig:fig3}).

\subsubsection{Loss Function}

The training objective combines physics, initial--boundary, and data constraints. Let $\mathcal{T}_f=\{(x_f^i,t_f^i)\}_{i=1}^{\mathcal{N}_f}$ be the collocation points, $\mathcal{T}_u$ the $\mathcal{N}_u$ initial and boundary samples, and $\mathcal{T}_d$ the $\mathcal{N}_d$ observations (for inverse or supervised problems). The physics, initial--boundary, and data losses are
\begin{align}
\mathcal{L}_f
&=
\frac{1}{\mathcal{N}_f}
\sum_{i=1}^{\mathcal{N}_f}
|f(x_f^i,t_f^i)|^2
=
\frac{1}{\mathcal{N}_f}
\sum_{i=1}^{\mathcal{N}_f}
\left[
(\Re f_i)^2
+
(\Im f_i)^2
\right],
\\
\mathcal{L}_u
&=
\frac{1}{\mathcal{N}_u}
\sum_{i=1}^{\mathcal{N}_u}
\left|
u^{(i)}
-
\Re
\big[
\hat{y}(x^{(i)},t^{(i)})
\big]
\right|^2,
\qquad
\mathcal{L}_d
=
\frac{1}{\mathcal{N}_d}
\sum_{i=1}^{\mathcal{N}_d}
\left|
u_d^{(i)}
-
\Re
\big[
\hat{y}(x_d^{(i)},t_d^{(i)})
\big]
\right|^2,
\end{align}
and the total loss is
\begin{equation}
\mathcal{L}
=
\lambda_f\mathcal{L}_f
+
\lambda_u\mathcal{L}_u
+
\lambda_d\mathcal{L}_d,
\end{equation}
where $\lambda_f$, $\lambda_u$, and $\lambda_d$ are weighting coefficients.

\subsubsection{Sampling Strategy}

Both non-adaptive and adaptive sampling strategies are considered following~\cite{wu2023comprehensive}, and they are compared in Appendix~\ref{appen:appenB}. In non-adaptive sampling, the collocation points remain fixed during training and are generated by uniform random sampling, structured grid sampling, Latin hypercube sampling (LHS), or the Halton, Hammersley, and Sobol sequences. In the adaptive residual-based distribution (RAD), a candidate set $\mathcal{C}=\{(x_i,t_i)\}_{i=1}^{\mathcal{N}_{\mathrm{cand}}}$ is uniformly generated at each resampling stage, the complex residual $|f_i|^2=(\Re f_i)^2+(\Im f_i)^2$ is evaluated, and new collocation points are drawn from $\mathcal{C}$ with probability
\begin{equation}
p_i
=
\frac{w_i}{\sum_{j} w_j},
\qquad
w_i
=
\frac{|f_i|^{\kappa}}
{\frac{1}{\mathcal{N}_{\mathrm{cand}}}\sum_{j} |f_j|^{\kappa}}
+c,
\end{equation}
where $\kappa$ is the residual exponent and $c$ is a stabilization constant (Algorithm~\ref{alg:scvrad}). Thus, more collocation points are allocated to regions with large residuals, which improves the resolution of sharp solution structures.

\refstepcounter{figure}

\begin{center}
\includegraphics[width=\textwidth]{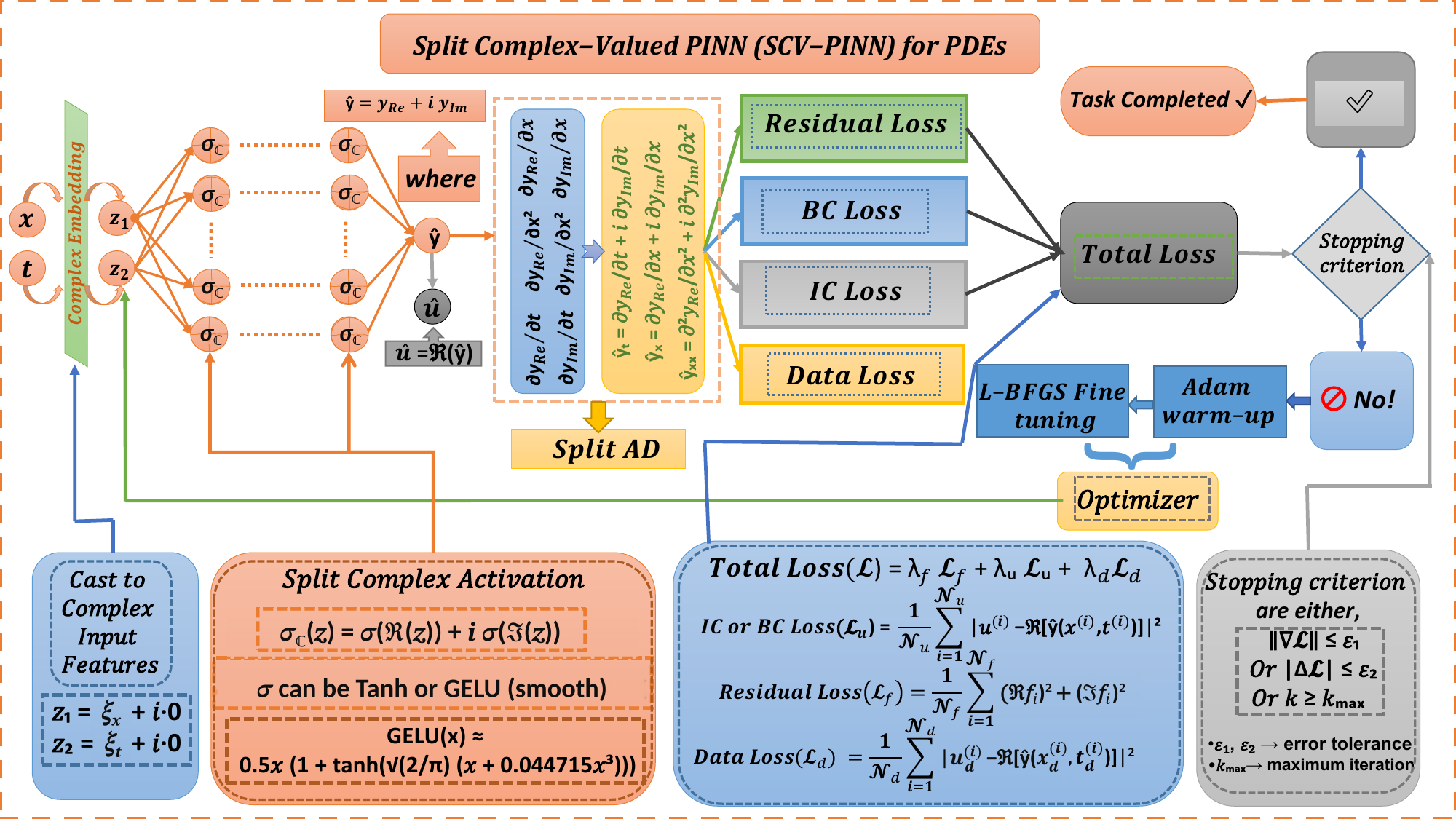}
\end{center}

\vspace{0.5em}

\noindent\textbf{Fig. \thefigure.}
General workflow of the proposed split complex-valued physics-informed neural network (SCV-PINN) framework for solving general nonlinear partial differential equations, including both forward solution prediction and inverse parameter identification problems.

\label{fig:fig3}

\subsubsection{Optimization Strategy}

Training is performed using a hybrid procedure. The Adam optimizer, $\theta\leftarrow\theta-\eta\nabla_{\theta}\mathcal{L}$ with learning rate $\eta$, is first employed, and the parameters are then refined using the L-BFGS quasi-Newton optimizer to achieve high-accuracy convergence.

\subsubsection{Stopping Criteria}

Let $\mathcal{L}_k$ denote the total loss at the $k$-th iteration. Training is terminated when $\|\nabla_{\theta}\mathcal{L}_k\| \leq \epsilon_1$, $|\mathcal{L}_{k+1}-\mathcal{L}_k| \leq \epsilon_2$, or $k \geq k_{\max}$, where $\epsilon_1$ and $\epsilon_2$ are problem-dependent tolerances for the gradient norm and loss variation, and $k_{\max}$ is the maximum number of iterations.
\subsubsection{Final Approximation}

After optimization, the solution is approximated as $u(x,t)\approx\hat{u}(x,t)=\Re\big[\mathcal{F}_{\theta^\ast}(\mathbf{z})\big]$, where $\theta^\ast$ denotes the trained parameters.

\begin{algorithm}[H]
\caption{SCV-PINNs for forward and inverse nonlinear PDE problems}
\label{alg:scvpinn_unified}
\begin{algorithmic}[1]

\Require PDE operator $\mathcal{N}_{x,t}[\cdot;\lambda]$, boundary operator $\mathcal{B}_{x,t}[\cdot]$, initial condition $u_0$, training sets $\mathcal{T}_u$, $\mathcal{T}_d$, and $\mathcal{T}_f$ of sizes $\mathcal{N}_u$, $\mathcal{N}_d$, and $\mathcal{N}_f$, network parameters $\theta$, unknown PDE parameters $\lambda$ (for inverse problems), loss weights $\lambda_f$, $\lambda_u$, $\lambda_d$, activation $\sigma_{\mathbb{C}}$, and optimization budgets $N_A$, $N_L$

\Ensure Predicted solution $\hat{u}(x,t)$ and identified parameters $\lambda^\ast$ (inverse case)

\State Initialize the SCV-PINN $\mathcal{F}_{\theta}$ and trainable PDE parameters $\lambda$ (if unknown)

\State Generate collocation points $\mathcal{T}_f\subset\Omega\times\mathcal{T}$ using the prescribed sampling strategy

\For{each point $(x_i,t_i)\in\mathcal{T}_u\cup\mathcal{T}_d\cup\mathcal{T}_f$}

\State Normalize $(x_i,t_i)$ to $\boldsymbol{\xi}_i\in[-1,1]^{d+1}$ and embed
\[
\mathbf{z}_i=\boldsymbol{\xi}_i+i\,0
\]

\State Compute the complex network output
\[
\hat{y}_i=\mathcal{F}_{\theta}(\mathbf{z}_i)\in\mathbb{C}^{m}
\]

\State Apply split activation
\[
\sigma_{\mathbb{C}}(z)
=
\sigma(\Re[z])
+
i\,\sigma(\Im[z])
\]

\State Extract the physical prediction
\[
\hat{u}_i=\Re(\hat{y}_i)
\]

\EndFor

\State Compute the initial--boundary and data losses
\[
\mathcal{L}_{u}
=
\frac{1}{\mathcal{N}_u}
\sum_{i=1}^{\mathcal{N}_u}
\left|
u^{(i)}-\hat{u}_i
\right|^2,
\qquad
\mathcal{L}_{d}
=
\frac{1}{\mathcal{N}_d}
\sum_{i=1}^{\mathcal{N}_d}
\left|
u_d^{(i)}-\hat{u}_i
\right|^2
\]

\State Evaluate the PDE residuals $f_i$ using automatic differentiation

\State Compute the physics loss
\[
\mathcal{L}_{f}
=
\frac{1}{\mathcal{N}_f}
\sum_{i=1}^{\mathcal{N}_f}
\left[
(\Re f_i)^2+(\Im f_i)^2
\right]
\]

\State Form the total loss
\[
\mathcal{L}
=
\lambda_f\mathcal{L}_f
+
\lambda_u\mathcal{L}_u
+
\lambda_d\mathcal{L}_d
\]

\State Train $(\theta,\lambda)$ using the Adam optimizer for $N_A$ iterations

\State Refine using L-BFGS optimization for $N_L$ iterations

\State Obtain the final prediction
\[
\hat{u}(x,t)
=
\Re\!\left(
\mathcal{F}_{\theta^\ast}(\mathbf{z})
\right)
\]

\If{inverse problem}
\State Simultaneously identify the optimal PDE parameters $\lambda^\ast$
\EndIf

\end{algorithmic}
\end{algorithm}

\begin{algorithm}[H]
\caption{SCV-RAD: Residual-Adaptive Collocation for SCV-PINNs}
\label{alg:scvrad}
\begin{algorithmic}[1]

\Require Initial SCV-PINN parameters $\theta$, number of collocation points $\mathcal{N}_f$, candidate pool size $\mathcal{N}_{\mathrm{cand}} \gg \mathcal{N}_f$, residual exponent $\kappa$, stabilization constant $c$, resampling interval $K_r$, and maximum number of training iterations $k_{\max}$

\Ensure Trained network parameters $\theta^\ast$

\State Initialize the collocation set $\mathcal{T}_f$ of size $\mathcal{N}_f$ using a prescribed sampling strategy

\State Train the SCV-PINN for $K_r$ iterations

\While{total training iterations $< k_{\max}$}

\State Draw a candidate set
\[
\mathcal{C}=\{(x_i,t_i)\}_{i=1}^{\mathcal{N}_{\mathrm{cand}}}
\sim
\mathcal{U}(\Omega \times \mathcal{T})
\]

\For{each $(x_i,t_i)\in\mathcal{C}$}

\State Evaluate the complex-valued PDE residual:
\[
|f_i|^2
=
|f(x_i,t_i;\theta)|^2
=
(\Re f_i)^2+(\Im f_i)^2
\]

\State Compute the RAD weight:
\[
w_i=
\frac{|f_i|^{\kappa}}
{\frac{1}{\mathcal{N}_{\mathrm{cand}}}\sum_j |f_j|^{\kappa}}
+c
\]

\EndFor

\State Normalize the sampling probabilities:
\[
p_i=
\frac{w_i}
{\sum_j w_j}
\]

\State Resample $\mathcal{N}_f$ collocation points from $\mathcal{C}$ according to $\{p_i\}$

\State Replace the previous collocation set $\mathcal{T}_f$

\State Train the SCV-PINN for another $K_r$ iterations

\EndWhile

\State \Return $\theta^\ast$

\end{algorithmic}
\end{algorithm}

\section{Numerical Experiments and Applications}
\label{sec:sec5}

The proposed SCV-PINN framework is systematically assessed across a diverse collection of forward and inverse problems to examine its versatility, predictive performance, and parameter identification capabilities. The forward benchmarks include the nonlinear Schrödinger, nonlinear Poisson, Burgers', Allen--Cahn, Korteweg--de Vries (KdV), and Helmholtz equations, Poisson and Helmholtz problems on irregular domains, Kovasznay flow ($\mathrm{Re}=20$), lid-driven cavity flow ($\mathrm{Re}=100$), and three-dimensional Beltrami flow. The inverse studies involve the Lorenz system, Burgers' equation, and two-dimensional Navier--Stokes equations. A unified complex-valued formulation is adopted for both real- and complex-valued systems; for real-valued problems, the imaginary component serves as an auxiliary latent representation, while the physical solution is obtained from the real component. To isolate the influence of the proposed representation, SCV-PINN and conventional real-valued PINN are evaluated using identical network architectures, training data, and optimization schedules comprising 5000 Adam iterations followed by 5000 L-BFGS iterations (Table~\ref{tab:tab10}). Predictive performance is quantified using the relative $L_2$ error and mean squared error (MSE),
\begin{equation}
\text{relative }L_2\text{ error}
=
\frac{
\sqrt{
\sum_{i=1}^{N_e}
\left|
\hat{u}(\mathbf{x}_i,t_i)-u(\mathbf{x}_i,t_i)
\right|^2
}
}{
\sqrt{
\sum_{i=1}^{N_e}
\left|
u(\mathbf{x}_i,t_i)
\right|^2
}
},
\qquad
\mathrm{MSE}
=
\frac{1}{N_e}
\sum_{i=1}^{N_e}
\left|
\hat{u}(\mathbf{x}_i,t_i)
-
u(\mathbf{x}_i,t_i)
\right|^2,
\label{eq:l2_error}
\end{equation}
where $u$ and $\hat{u}$ denote the reference and predicted solutions, respectively, evaluated at $N_e$ test points.

\subsection{Forward Problem: Nonlinear Schrödinger Equation}

We consider the nonlinear Schrödinger (NLS) equation, a complex-valued PDE given by
\begin{equation}
i\psi_t + \frac{1}{2}\psi_{xx} + |\psi|^2\psi = 0,
\label{eq:nls}
\end{equation}
where $\psi(x,t)=u(x,t)+iv(x,t)$. Decomposing Eq.~\eqref{eq:nls} into real and imaginary components yields
\begin{align}
f_u &:= v_t - \frac{1}{2}u_{xx} - (u^2+v^2)u = 0,\\
f_v &:= u_t + \frac{1}{2}v_{xx} + (u^2+v^2)v = 0.
\end{align}
The computational domain is $x\in[-5,5]$ and $t\in[0,\pi/2]$, with $\psi(x,0)=2\,\mathrm{sech}(x)$ and periodic boundary conditions. The CV-MLP comprises five hidden layers of 100 neurons each with CGELU activation. Training employs $20{,}000$ Latin hypercube sampling (LHS) collocation points, followed by 5000 Adam and 5000 L-BFGS iterations.

Figures~\ref{fig:fig4}--\ref{fig:fig6} and Tables~\ref{tab:tab9}--\ref{tab:tab10} summarize the results. Figure~\ref{fig:fig4} compares the predicted modulus $|\psi|$ with the reference solution~\cite{raissi2019physics}, yielding a relative $L_2$ error of $8.45\times10^{-5}$. The predicted $u$, $v$, and $|\psi|$ profiles at $t=0.78\,\mathrm{s}$ (Fig.~\ref{fig:fig5}), together with $|\psi|$ at $t=0.59$, $0.79$, and $0.98\,\mathrm{s}$ (Fig.~\ref{fig:fig6}), show close agreement with the reference. Under identical settings (Table~\ref{tab:tab10}), SCV-PINN achieves a relative $L_2$ error of $1.05\times10^{-3}$ and an MSE of $8.83\times10^{-7}$, compared with $6.04\times10^{-2}$ and $2.91\times10^{-3}$ for the real-valued PINN. The corresponding loss histories (Fig.~\ref{fig:fig27}(a)) further indicate faster convergence to a lower loss, highlighting the effectiveness of split complex-valued representations for oscillatory wave dynamics.

\begin{figure}[pos=t]
\normalfont\normalsize\leftskip=0pt\rightskip=0pt\parfillskip=0pt plus 1fil\relax
\refstepcounter{figure}

\begin{center}
\includegraphics[width=\textwidth]{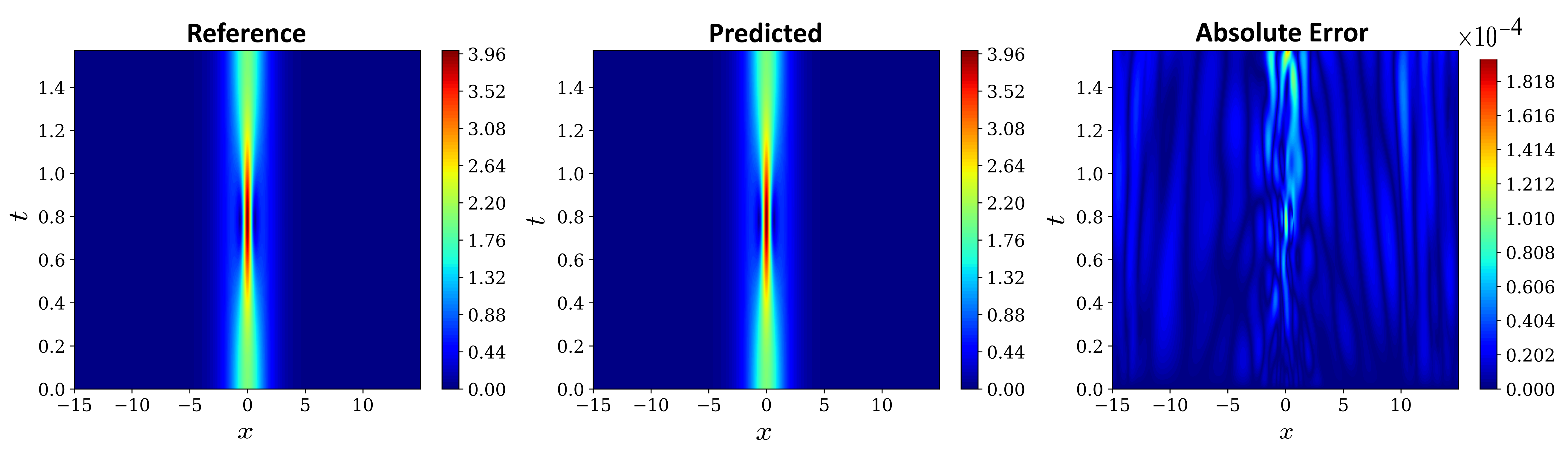}
\end{center}

\vspace{0.5em}

\noindent\textbf{Fig. \thefigure.}
Reference solution, SCV-PINN predicted solution, and corresponding absolute error contours for the nonlinear Schrödinger equation in terms of the modulus field $|\psi(x,t)|$. The proposed SCV-PINN achieves a relative $L_2$ error of $8.45\times10^{-5}$.

\label{fig:fig4}
\end{figure}

\begin{figure}[pos=t]
\normalfont\normalsize\leftskip=0pt\rightskip=0pt\parfillskip=0pt plus 1fil\relax
\refstepcounter{figure}

\begin{center}
\includegraphics[width=\textwidth]{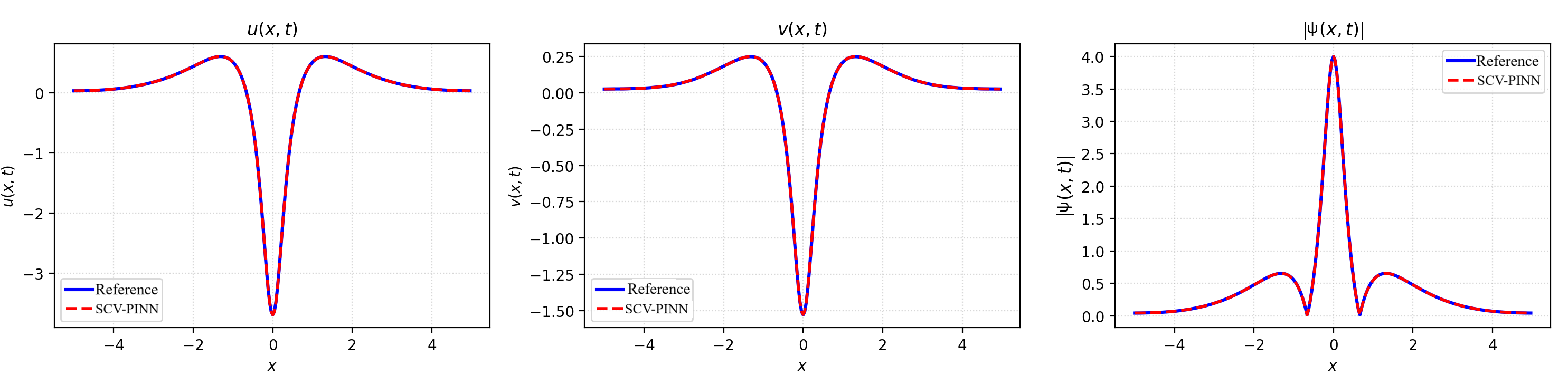}
\end{center}

\vspace{0.5em}

\noindent\textbf{Fig. \thefigure.}
Reference and predicted solution comparison for the nonlinear Schrödinger equation at $t=0.78\,\mathrm{s}$ for the real part $u(x,t)$, imaginary part $v(x,t)$, and modulus field $|\psi(x,t)|$. The proposed SCV-PINN accurately captures the complex-valued wave dynamics with close agreement to the reference solution.

\label{fig:fig5}
\end{figure}

\begin{figure}[pos=t]
\normalfont\normalsize\leftskip=0pt\rightskip=0pt\parfillskip=0pt plus 1fil\relax
\refstepcounter{figure}

\begin{center}
\includegraphics[width=\textwidth]{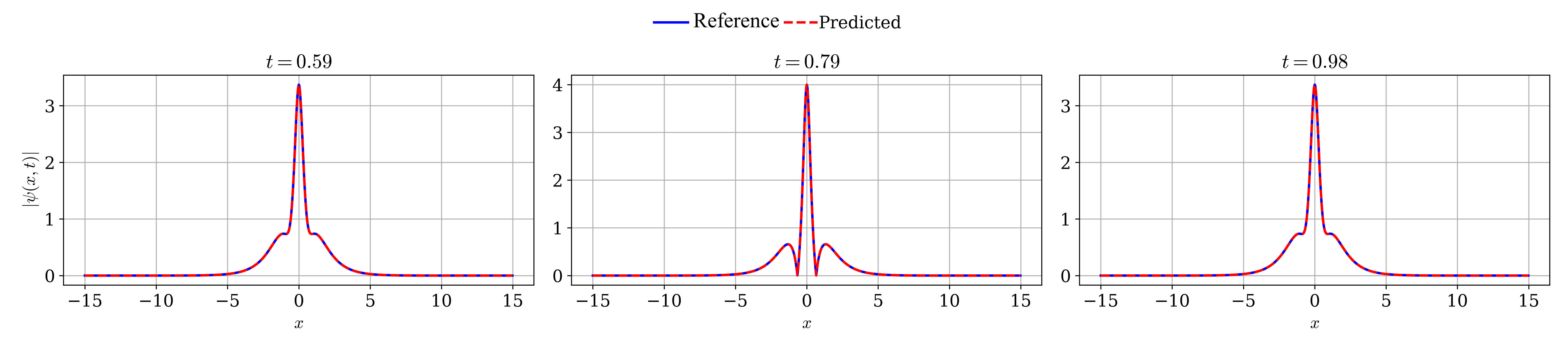}
\end{center}

\vspace{0.5em}

\noindent\textbf{Fig. \thefigure.}
Comparison of the reference and predicted solutions of the nonlinear Schrödinger equation for the modulus field $|\psi(x,t)|$ at three different time instances: $t=0.59\,\mathrm{s}$, $t=0.79\,\mathrm{s}$, and $t=0.98\,\mathrm{s}$. The proposed SCV-PINN predictions closely agree with the reference solution at all temporal snapshots.

\label{fig:fig6}
\end{figure}

\subsection{Forward Problem: Nonlinear Poisson Equation}

We next examine the nonlinear Poisson equation with an exponential reaction term,
\begin{equation}
\nabla^2\phi-e^{\phi}=r(x,y),
\label{eq:nlp}
\end{equation}
on $\Omega=[0,1]\times[0,1]$. For $r(x,y)=0$, Eq.~\eqref{eq:nlp} corresponds to the classical Liouville equation, encountered in differential geometry, steady-flow vorticity models, and Chern--Simons field theories. The source term $r(x,y)$ is prescribed using the exact solution
\begin{equation}
\phi(x,y)=1+\sin(k\pi x)\cos(k\pi y),
\label{eq:nlp_exact}
\end{equation}
where $k\in\mathbb{Z}$ determines the spatial oscillations. We choose $k=4$ to assess the framework under pronounced oscillatory behavior. The associated Dirichlet conditions are
\begin{align}
\phi(0,y) &= \phi(1,y) = 1,
\\
\phi(x,0) &= 1+\sin(k\pi x),
\\
\phi(x,1) &= 1+\sin(k\pi x)\cos(k\pi).
\end{align}
These boundary conditions are enforced exactly through
\begin{align}
f_b(x,y)
&=
1+
\Bigl[
1-y\bigl(1-\cos(k\pi)\bigr)
\Bigr]
\sin(k\pi x),
\\
h_b(x,y)
&=
xy(1-x)(1-y).
\end{align}
The SCV-PINN employs two hidden layers with 30 neurons each and CGELU activation. Collocation points are selected using Latin hypercube sampling (LHS), followed by $10{,}000$ Adam and $20{,}000$ L-BFGS iterations. As shown in Fig.~\ref{fig:fig7}, the predicted solution closely matches the oscillatory reference field, achieving a relative $L_2$ error of $4.45\times10^{-7}$ (Table~\ref{tab:tab9}).

\begin{figure}[pos=t]
\normalfont\normalsize\leftskip=0pt\rightskip=0pt\parfillskip=0pt plus 1fil\relax
\refstepcounter{figure}

\begin{center}
\includegraphics[width=\textwidth]{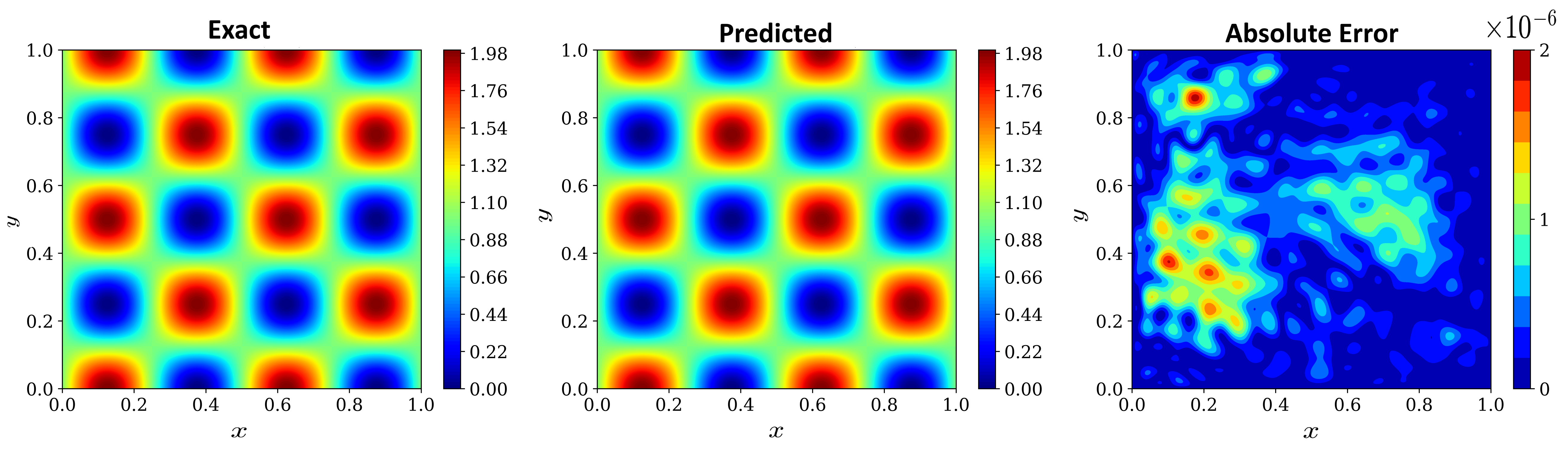}
\end{center}

\vspace{0.5em}

\noindent\textbf{Fig. \thefigure.}
Exact solution, SCV-PINN predicted solution, and corresponding absolute error for the nonlinear Poisson equation on a regular domain with $k=4$. The proposed SCV-PINN achieves a relative $L_2$ error of $4.45\times10^{-7}$.

\label{fig:fig7}
\end{figure}

\subsection{Forward Problem: Poisson Equation on an L-Shaped Domain}

To examine the applicability of SCV-PINN to irregular geometries, we consider the Poisson equation on an L-shaped domain,
\begin{equation}
-\Delta u(x,y)=1,
\qquad (x,y)\in\Omega,
\label{eq:lshape_poisson}
\end{equation}
with homogeneous Dirichlet boundary conditions,
\begin{equation}
u(x,y)=0,
\qquad (x,y)\in\partial\Omega,
\label{eq:lshape_bc}
\end{equation}
where $\Omega=[-1,1]^2\setminus[0,1]^2$. The re-entrant corner introduces a localized singularity that challenges numerical approximation. The network comprises four hidden layers with 50 neurons each and CGELU activation, using $1200$ interior Latin hypercube sampling (LHS) points and $120$ boundary points. Figure~\ref{fig:fig8} illustrates the predicted solution, including the region near the re-entrant corner, with a relative $L_2$ error of $1.26\times10^{-2}$ (Table~\ref{tab:tab9}). Under identical settings (Table~\ref{tab:tab10}), SCV-PINN yields an MSE of $8.84\times10^{-5}$ and a relative $L_2$ error of $1.26\times10^{-2}$, compared with $1.41\times10^{-4}$ and $1.59\times10^{-1}$ for the real-valued PINN. The loss histories in Fig.~\ref{fig:fig27}(b) further indicate more favorable convergence behavior for SCV-PINN in this benchmark.

\begin{figure}[pos=t]
\normalfont\normalsize\leftskip=0pt\rightskip=0pt\parfillskip=0pt plus 1fil\relax
\refstepcounter{figure}

\begin{center}
\includegraphics[width=\textwidth]{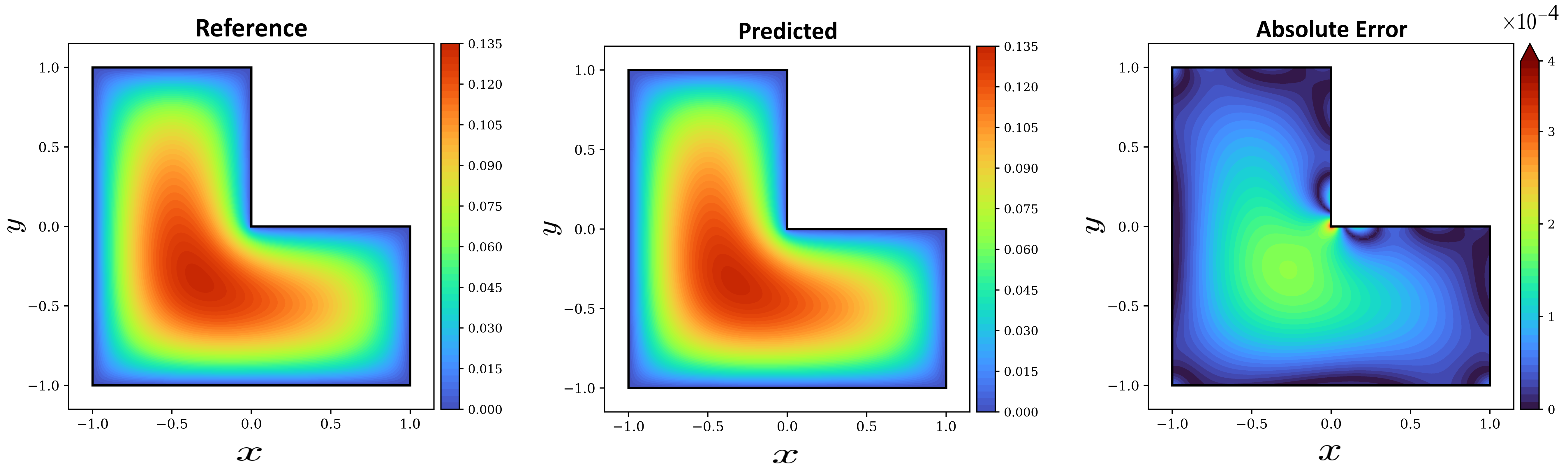}
\end{center}

\vspace{0.5em}

\noindent\textbf{Fig. \thefigure.}
Comparison of the reference solution, SCV-PINN predicted solution, and corresponding absolute error for the Poisson equation on an irregular L-shaped domain.

\label{fig:fig8}
\end{figure}

\subsection{Forward Problem: Viscous Burgers' Equation}

We consider the viscous Burgers' equation,
\begin{equation}
u_t + u u_x - \frac{0.01}{\pi}u_{xx}=0,
\qquad x\in[-1,1], \quad t\in[0,1],
\label{eq:burgers}
\end{equation}
with the initial condition
\begin{equation}
u(0,x)=-\sin(\pi x),
\label{eq:burgers_ic}
\end{equation}
and homogeneous Dirichlet boundary conditions,
\begin{equation}
u(t,-1)=u(t,1)=0.
\label{eq:burgers_bc}
\end{equation}
The interplay between convection and diffusion produces a sharp solution gradient, providing a benchmark for nonlinear PDE solvers. The CV-MLP employs eight hidden layers with 20 neurons each and CTanh activation, using $\mathcal{N}_u=100$ initial and boundary points and $\mathcal{N}_f=10{,}000$ Latin hypercube sampling (LHS) collocation points. Figure~\ref{fig:fig9} compares the SCV-PINN prediction with the reference solution~\cite{raissi2019physics} and its absolute error, while Fig.~\ref{fig:fig10} presents solution profiles at $t=0.000$, $0.333$, $0.667$, and $1.000\,\mathrm{s}$. SCV-PINN achieves a relative $L_2$ error of $6.73\times10^{-5}$ (Table~\ref{tab:tab9}), compared with $6.7\times10^{-4}$ for PINN~\cite{raissi2019physics}, $4.80\times10^{-4}$ for SA-PINN~\cite{mcclenny2023self}, and $2.83\times10^{-4}$ for LA-PINN~\cite{song2024loss}, as reported in Table~\ref{tab:tab1}. Under matched architecture and training settings, the comparison with PINN~\cite{raissi2019physics} and SA-PINN~\cite{mcclenny2023self} in Fig.~\ref{fig:fig29} further demonstrates the lower prediction error of SCV-PINN. The controlled comparison with the real-valued PINN~\cite{raissi2019physics} (Table~\ref{tab:tab10}) and the loss histories in Fig.~\ref{fig:fig28}(a) additionally indicate improved predictive performance and faster convergence for the proposed framework.

\begin{table}[pos=t]
\centering
\caption{
Comparison of relative $L_2$ error for the 1-D Burgers' equation using different physics-informed neural network frameworks. For a fair comparison, all models were trained using similar network hyperparameters.
}
\label{tab:tab1}

\vspace{0.3em}

\renewcommand{\arraystretch}{1.15}

\begin{tabular}{lc}
\hline
\textbf{Method} & \textbf{relative $L_2$ error} \\
\hline
PINN~\cite{raissi2019physics} & $6.7 \times 10^{-4}$ \\
SA-PINN~\cite{mcclenny2023self} & $4.80 \times 10^{-4}$ \\
LA-PINN~\cite{song2024loss} & $2.83 \times 10^{-4}$ \\
SCV-PINN (Proposed) & $\mathbf{6.73 \times 10^{-5}}$ \\
\hline
\end{tabular}

\end{table}

\begin{figure}[pos=t]
\normalfont\normalsize\leftskip=0pt\rightskip=0pt\parfillskip=0pt plus 1fil\relax
\refstepcounter{figure}

\begin{center}
\includegraphics[width=\textwidth]{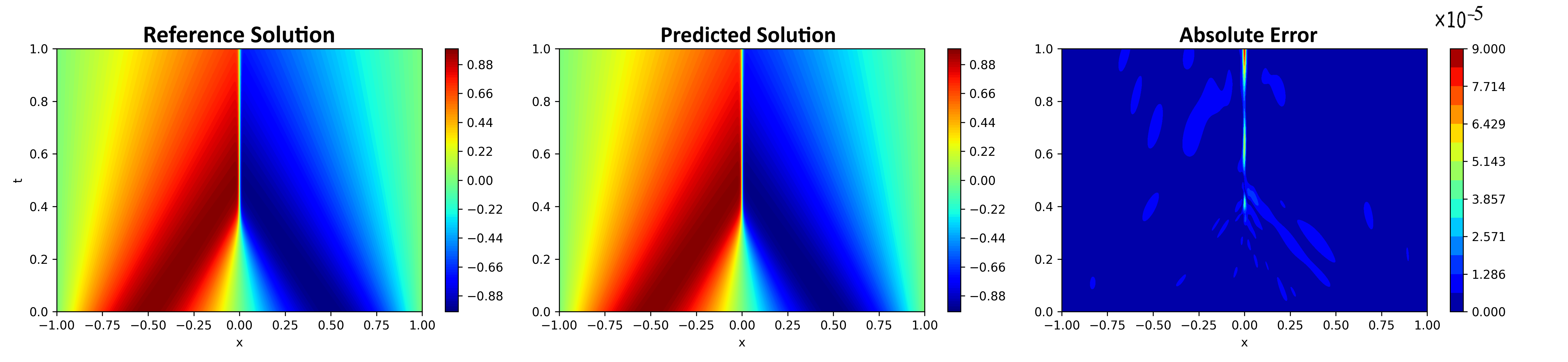}
\end{center}

\vspace{0.5em}

\noindent\textbf{Fig. \thefigure.}
Comparison of the reference solution, predicted solution, and corresponding absolute error for the 1-D Burgers' equation. The proposed SCV-PINN accurately captures the nonlinear shock behavior with a relative $L_2$ error of $6.73\times10^{-5}$.

\label{fig:fig9}
\end{figure}

\begin{figure}[pos=t]
\normalfont\normalsize\leftskip=0pt\rightskip=0pt\parfillskip=0pt plus 1fil\relax
\refstepcounter{figure}

\begin{center}
\includegraphics[width=0.8\textwidth]{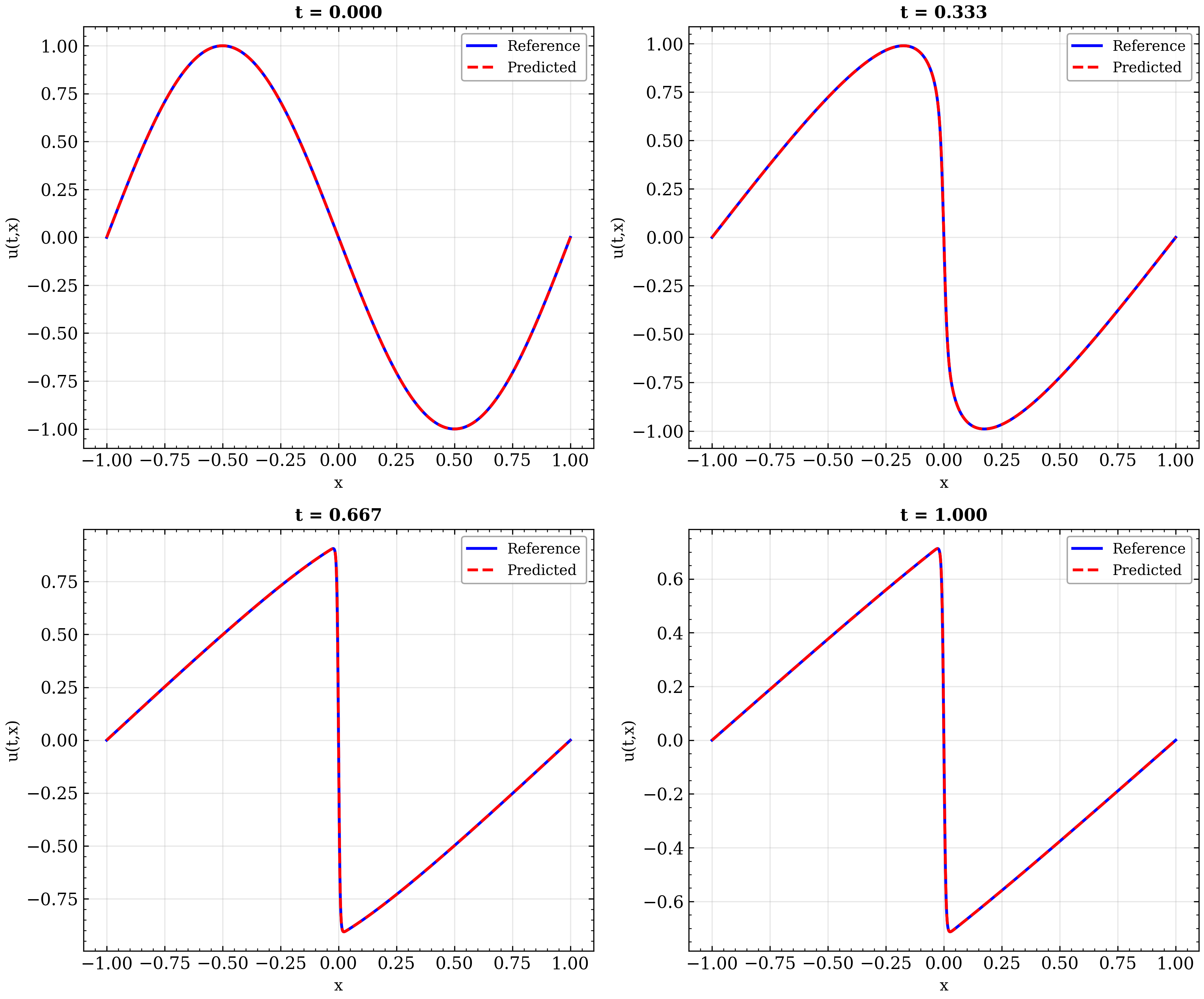}
\end{center}

\vspace{0.5em}

\noindent\textbf{Fig. \thefigure.}
Comparison of the reference and predicted solutions of the 1-D Burgers' equation at four different time instants: $t=0.000\,\mathrm{s}$, $t=0.333\,\mathrm{s}$, $t=0.667\,\mathrm{s}$, and $t=1.000\,\mathrm{s}$.

\label{fig:fig10}
\end{figure}

\subsection{Forward Problem: Allen--Cahn Equation}

We investigate the Allen--Cahn equation, a nonlinear reaction--diffusion model,
\begin{equation}
u_t - 0.0001\,u_{xx} + 5u^3 - 5u = 0,
\qquad x \in [-1,1], \quad t \in [0,1],
\label{eq:allen_cahn}
\end{equation}
with the initial condition
\begin{equation}
u(0,x)=x^2\cos(\pi x),
\label{eq:allen_initial}
\end{equation}
and periodic boundary conditions,
\begin{equation}
u(t,-1)=u(t,1),
\qquad
u_x(t,-1)=u_x(t,1).
\label{eq:allen_bc}
\end{equation}
The nonlinear reaction and diffusion terms govern phase separation and the development of sharp interfaces. The SCV-PINN employs four hidden layers with 128 neurons each and CGELU activation, using $\mathcal{N}_f=20{,}000$ collocation points, $\mathcal{N}_0=100$ initial points, and $\mathcal{N}_b=100$ LHS-sampled boundary pairs ($200$ boundary points), followed by Adam and L-BFGS optimization. Figures~\ref{fig:fig11}--\ref{fig:fig12} show close agreement with the reference solution~\cite{raissi2019physics}, including profiles at $t=0.10$, $0.45$, and $0.90\,\mathrm{s}$, with a relative $L_2$ error of $4.39\times10^{-4}$ (Table~\ref{tab:tab9}). The comparison in Table~\ref{tab:tab2} further shows that SCV-PINN achieves a lower error than PINN~\cite{raissi2019physics}, SA-PINN~\cite{mcclenny2023self}, and LA-PINN~\cite{song2024loss} for this benchmark, demonstrating the potential of split complex-valued representations for nonlinear interface dynamics.

\begin{table}[pos=t]
\centering
\caption{
Comparison of relative $L_2$ error for the Allen--Cahn equation using different physics-informed neural network frameworks.
}
\label{tab:tab2}

\vspace{0.3em}

\renewcommand{\arraystretch}{1.15}

\begin{tabular}{lc}
\hline
\textbf{Method} & \textbf{relative $L_2$ error} \\
\hline
PINN~\cite{raissi2019physics} & $9.53 \times 10^{-1}$ \\
SA-PINN~\cite{mcclenny2023self}      & $3.34 \times 10^{-2}$ \\
LA-PINN~\cite{song2024loss}      & $8.22 \times 10^{-3}$ \\
SCV-PINN (Proposed) & $\mathbf{4.39 \times 10^{-4}}$ \\
\hline
\end{tabular}

\end{table}

\begin{figure}[pos=t]
\normalfont\normalsize\leftskip=0pt\rightskip=0pt\parfillskip=0pt plus 1fil\relax
\refstepcounter{figure}

\begin{center}
\includegraphics[width=\textwidth]{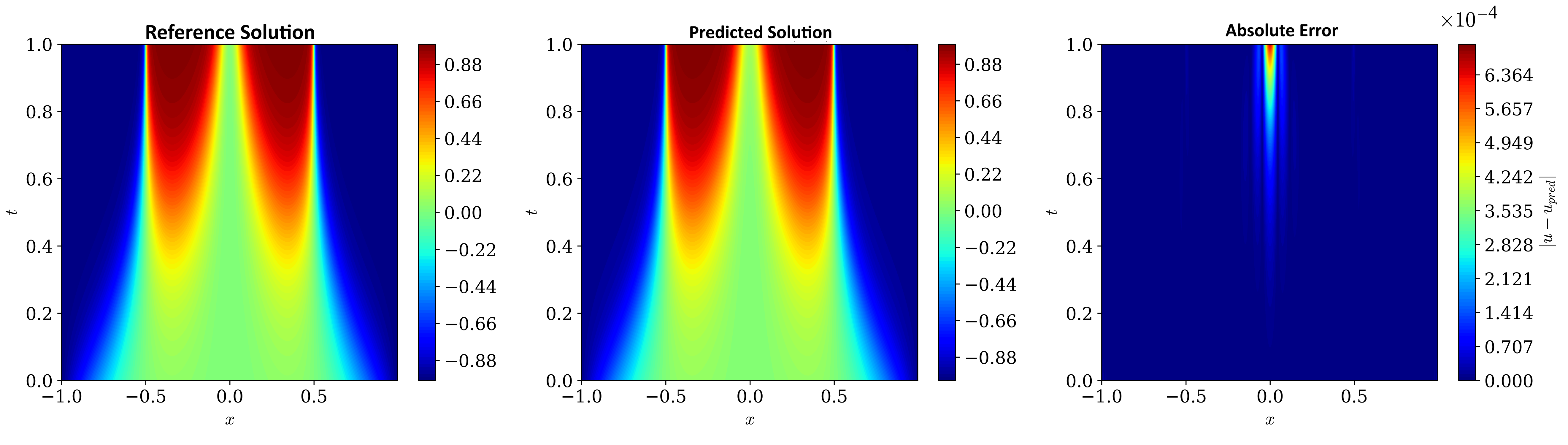}
\end{center}

\vspace{0.5em}

\noindent\textbf{Fig. \thefigure.}
Comparison of the reference solution, predicted solution, and corresponding absolute error for the Allen--Cahn equation.

\label{fig:fig11}
\end{figure}
\begin{figure}[pos=t]
\normalfont\normalsize\leftskip=0pt\rightskip=0pt\parfillskip=0pt plus 1fil\relax
\refstepcounter{figure}

\begin{center}
\includegraphics[width=\textwidth]{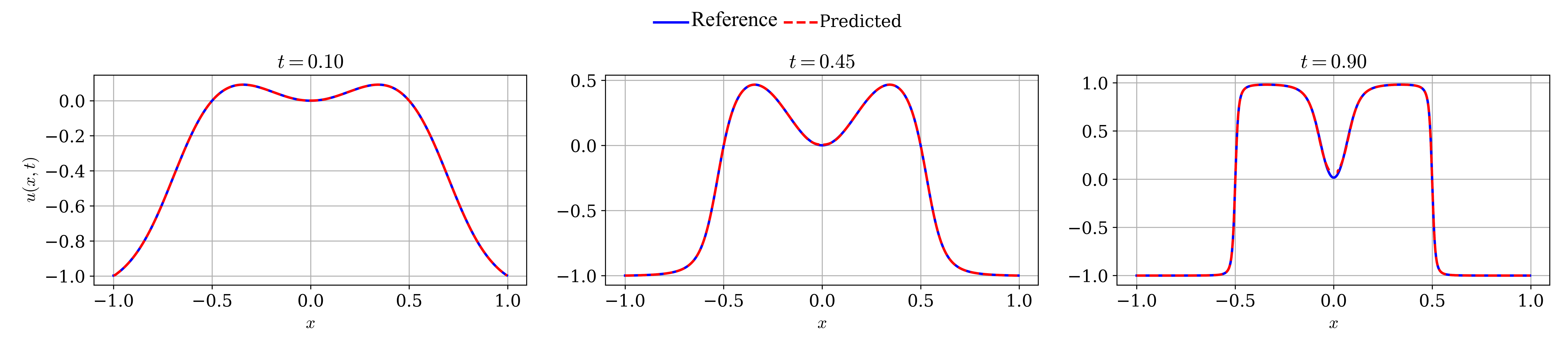}
\end{center}

\vspace{0.5em}

\noindent\textbf{Fig. \thefigure.}
Comparison of the Allen--Cahn solution at three different time instants: $t=0.10\,\mathrm{s}$, $t=0.45\,\mathrm{s}$, and $t=0.90\,\mathrm{s}$.

\label{fig:fig12}
\end{figure}

\subsection{Forward Problem: Korteweg--de Vries (KdV) Equation}

We examine nonlinear dispersive wave dynamics using the Korteweg--de Vries (KdV) equation,
\begin{equation}
\alpha \frac{\partial u}{\partial t}
+
\beta u \frac{\partial u}{\partial x}
+
\gamma \frac{\partial^3 u}{\partial x^3}
=0,
\label{eq:kdv}
\end{equation}
with $\alpha=1$, $\beta=6$, and $\gamma=1$. This equation describes soliton dynamics in shallow-water flows, plasmas, and nonlinear optical systems~\cite{segur1973korteweg}. The analytical two-soliton solution is expressed as
\begin{equation}
u_{\mathrm{an}}(x,t)=
\frac{
2(c_1-c_2)
\left[
c_1 \cosh^2\!\left(\dfrac{\sqrt{c_2}\zeta_2}{2}\right)
+
c_2 \sinh^2\!\left(\dfrac{\sqrt{c_1}\zeta_1}{2}\right)
\right]
}{
\left[
(\sqrt{c_1}-\sqrt{c_2})
\cosh\!\left(
\dfrac{\sqrt{c_1}\zeta_1+\sqrt{c_2}\zeta_2}{2}
\right)
+
(\sqrt{c_1}+\sqrt{c_2})
\cosh\!\left(
\dfrac{\sqrt{c_1}\zeta_1-\sqrt{c_2}\zeta_2}{2}
\right)
\right]^2
},
\label{eq:kdv_exact}
\end{equation}
where $\zeta_i=x-c_i t-x_i$ $(i=1,2)$, with $x_1=-2$, $x_2=2$, $c_1=6$, and $c_2=2$. The initial and boundary data are prescribed from the analytical solution:
\begin{align}
u(0,x) &= u(x,0),\\
u(t,x_0) &= u(x_0,t),\\
u(t,x_0+L) &= u(x_0+L,t),\\
u_x(t,x_0+L) &= \frac{\partial u}{\partial x}(x_0+L,t).
\end{align}
The SCV-PINN comprises three hidden layers with 30 neurons each and CGELU activation, trained using $10{,}000$ Adam iterations followed by $20{,}000$ L-BFGS iterations. Figure~\ref{fig:fig13} demonstrates close agreement with the reference solution~\cite{urban2025unveiling} during soliton propagation and interaction, yielding a relative $L_2$ error of $5.54\times10^{-4}$ (Table~\ref{tab:tab9}).

\begin{figure}[pos=t]
\normalfont\normalsize\leftskip=0pt\rightskip=0pt\parfillskip=0pt plus 1fil\relax
\refstepcounter{figure}

\begin{center}
\includegraphics[width=\textwidth]{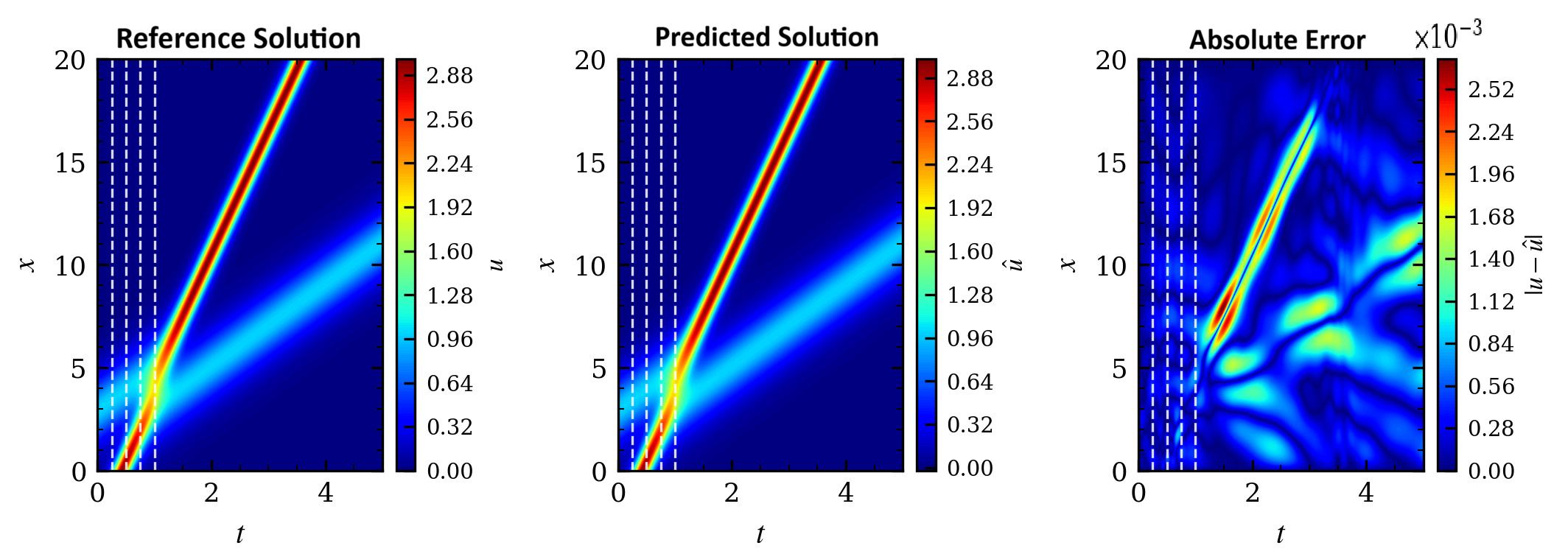}
\end{center}

\vspace{0.5em}

\noindent\textbf{Fig. \thefigure.}
Comparison of the reference solution, predicted solution, and corresponding absolute error for the Korteweg--de Vries (KdV) equation.

\label{fig:fig13}
\end{figure}

\subsection{Forward Problem: Helmholtz Equation}

We investigate oscillatory wave solutions through the two-dimensional Helmholtz equation on $(x,y)\in[-1,1]^2$,
\begin{equation}
u_{xx}+u_{yy}+k^2u-q(x,y)=0,
\label{eq:helmholtz}
\end{equation}
subject to homogeneous Dirichlet conditions,
\begin{equation}
u(-1,y)=u(1,y)=u(x,-1)=u(x,1)=0.
\label{eq:helmholtz_bc}
\end{equation}
The source term is prescribed as
\begin{align}
q(x,y) =\,&
-(a_1\pi)^2 \sin(a_1\pi x)\sin(a_2\pi y)
\nonumber\\
&
-(a_2\pi)^2 \sin(a_1\pi x)\sin(a_2\pi y)
\nonumber\\
&
+k^2 \sin(a_1\pi x)\sin(a_2\pi y),
\label{eq:helmholtz_source}
\end{align}
yielding the analytical solution
\begin{equation}
u(x,y)=\sin(a_1\pi x)\sin(a_2\pi y),
\label{eq:helmholtz_exact}
\end{equation}
with $a_1=1$ and $a_2=4$. This benchmark examines wave-like spatial oscillations relevant to acoustics, electromagnetics, and optics. The CV-MLP comprises four hidden layers with 50 neurons each and CTanh activation, using Latin hypercube sampling (LHS) collocation points and sequential Adam--L-BFGS optimization. Figure~\ref{fig:fig14} shows close agreement between the SCV-PINN prediction and the analytical wave field, yielding a relative $L_2$ error of $2.17\times10^{-4}$ (Table~\ref{tab:tab9}). The solution profiles at $y=-0.5$, $0.0$, and $0.5$ in Fig.~\ref{fig:fig15} further support this agreement. Under identical settings, SCV-PINN achieves a lower prediction error than the real-valued PINN (Table~\ref{tab:tab10}), while the loss histories in Fig.~\ref{fig:fig27}(g) indicate faster convergence.

\begin{figure}[pos=t]
\normalfont\normalsize\leftskip=0pt\rightskip=0pt\parfillskip=0pt plus 1fil\relax
\refstepcounter{figure}

\begin{center}
\includegraphics[width=\textwidth]{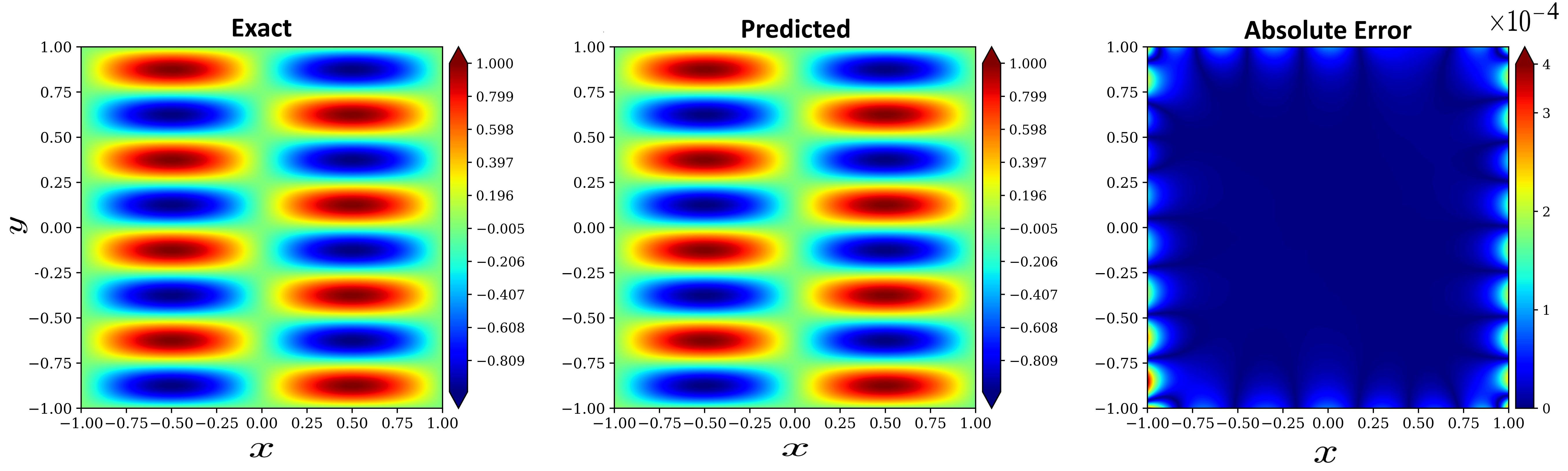}
\end{center}

\vspace{0.5em}

\noindent\textbf{Fig. \thefigure.}
Comparison of the exact and predicted solutions of the Helmholtz equation (regular) for $a_1=1$ and $a_2=4$.

\label{fig:fig14}
\end{figure}

\begin{figure}[pos=t]
\normalfont\normalsize\leftskip=0pt\rightskip=0pt\parfillskip=0pt plus 1fil\relax
\refstepcounter{figure}

\begin{center}
\includegraphics[width=\textwidth]{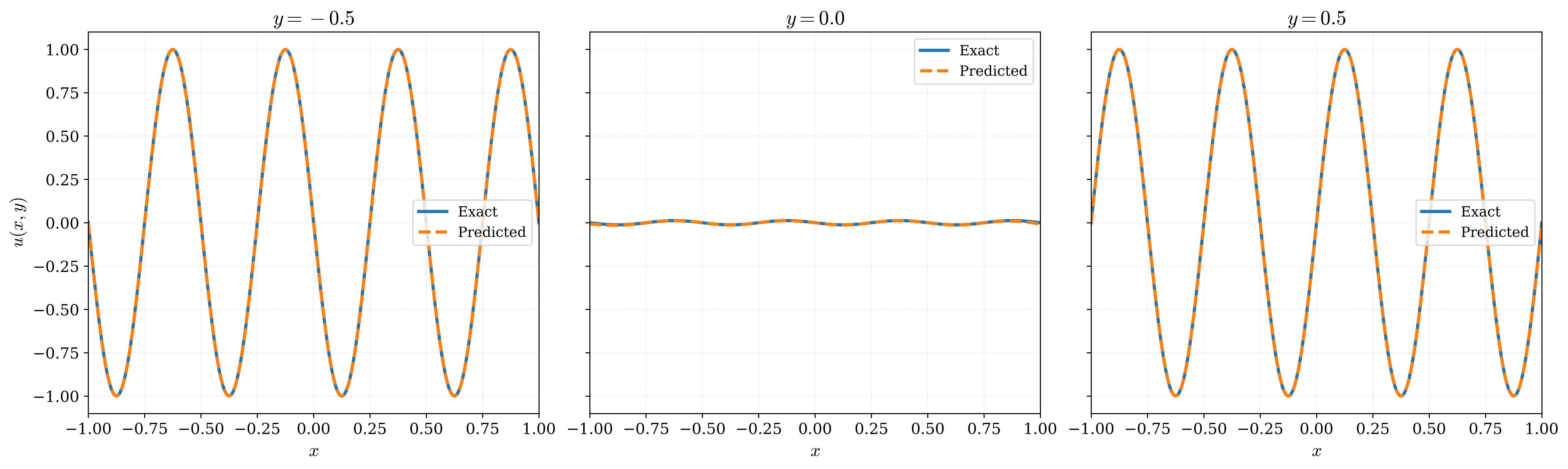}
\end{center}

\vspace{0.5em}

\noindent\textbf{Fig. \thefigure.}
Comparison of the exact and predicted slices of the Helmholtz equation (regular) at $y=0.5$, $y=0.0$, and $y=-0.5$.

\label{fig:fig15}
\end{figure}

\subsection{Forward Problem: Helmholtz Equation over a Square Domain with a Circular Hole}

We examine the Helmholtz equation on a square domain containing a circular cavity, $\Omega=[-0.5,0.5]^2\setminus D(0,R)$, where $D(0,R)$ denotes a disk of radius $R=1/4$ centered at the origin:
\begin{equation}
-\Delta u-k_0^2u
=
k_0^2\sin(k_0x)\sin(k_0y),
\qquad (x,y)\in\Omega,
\label{eq:helmholtz_hole}
\end{equation}
with $k_0=2\pi n$ and $n=1$. The manufactured solution is
\begin{equation}
u_{\mathrm{exact}}(x,y)
=
\sin(k_0x)\sin(k_0y).
\label{eq:helmholtz_hole_exact}
\end{equation}
Dirichlet conditions prescribed by Eq.~\eqref{eq:helmholtz_hole_exact} are enforced on $\Gamma_{\mathrm{outer}}$, while the inner boundary $\Gamma_{\mathrm{inner}}$ satisfies
\begin{equation}
\nabla u\cdot\hat{n}=g(x,y),
\end{equation}
where $\hat{n}$ is the outward unit normal and
\begin{equation}
g(x,y)=
\left[
k_0\cos(k_0x)\sin(k_0y),
\;
k_0\sin(k_0x)\cos(k_0y)
\right]\cdot\hat{n}.
\end{equation}
This configuration evaluates the framework on an irregular domain with oscillatory solutions and mixed boundary conditions. The SCV-PINN employs three hidden layers of 350 neurons each with CGELU activation, using Latin hypercube sampling (LHS) to generate 200 Dirichlet, 200 Neumann, and 2000 interior points. Spatial derivatives, including the Neumann residual, are computed through automatic differentiation. Figure~\ref{fig:fig16} shows close agreement with the manufactured solution, yielding a relative $L_2$ error of $1.17\times10^{-4}$ (Table~\ref{tab:tab9}). Under identical settings, SCV-PINN achieves lower MSE and relative $L_2$ errors than the real-valued PINN (Table~\ref{tab:tab10}), with faster convergence observed in Fig.~\ref{fig:fig27}(f).

\begin{figure}[pos=t]
\normalfont\normalsize\leftskip=0pt\rightskip=0pt\parfillskip=0pt plus 1fil\relax
\refstepcounter{figure}

\begin{center}
\includegraphics[width=\textwidth]{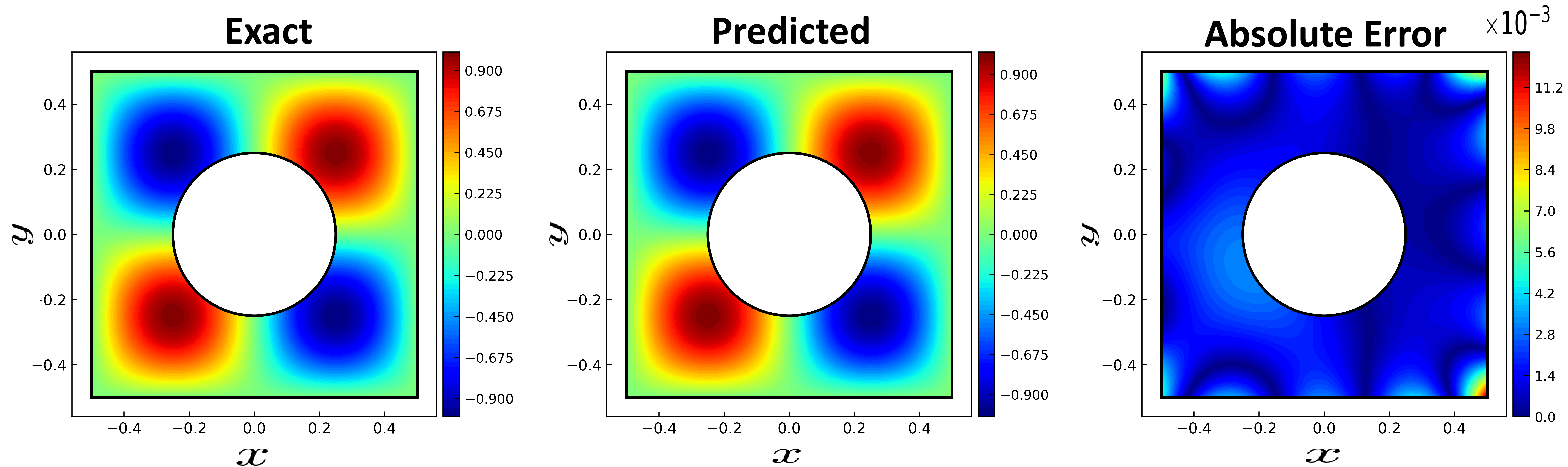}
\end{center}

\vspace{0.5em}

\noindent\textbf{Fig. \thefigure.}
Comparison of the exact solution, predicted solution, and corresponding absolute error for the Helmholtz equation over a domain with a circular hole of radius $R=1/4$.

\label{fig:fig16}
\end{figure}

\subsection{Forward Problem: Kovasznay Flow at $\mathrm{Re}=20$}

We investigate steady incompressible flow using the Kovasznay benchmark governed by the Navier--Stokes equations,
\begin{align}
u\frac{\partial u}{\partial x}
+
v\frac{\partial u}{\partial y}
+
\frac{\partial p}{\partial x}
-
\frac{1}{\mathrm{Re}}
\left(
\frac{\partial^2 u}{\partial x^2}
+
\frac{\partial^2 u}{\partial y^2}
\right)
&=0,
\\
u\frac{\partial v}{\partial x}
+
v\frac{\partial v}{\partial y}
+
\frac{\partial p}{\partial y}
-
\frac{1}{\mathrm{Re}}
\left(
\frac{\partial^2 v}{\partial x^2}
+
\frac{\partial^2 v}{\partial y^2}
\right)
&=0,
\\
\frac{\partial u}{\partial x}
+
\frac{\partial v}{\partial y}
&=0,
\end{align}
on $(x,y)\in[-0.5,1.0]\times[-0.5,1.5]$, with $\mathrm{Re}=20$, $\nu=1/\mathrm{Re}=0.05$, and
\begin{equation}
\lambda
=
\frac{1}{2\nu}
-
\sqrt{
\frac{1}{4\nu^2}
+
4\pi^2
}.
\end{equation}
The analytical solution is
\begin{align}
u(x,y)
&=
1-e^{\lambda x}\cos(2\pi y),
\\
v(x,y)
&=
\frac{\lambda}{2\pi}
e^{\lambda x}\sin(2\pi y),
\\
p(x,y)
&=
\frac{1}{2}
\left(
1-e^{2\lambda x}
\right).
\end{align}
Velocity conditions are prescribed from the analytical solution along all boundaries, with pressure specified at the outlet. The SCV-PINN employs four hidden layers with 50 neurons each and CTanh activation, predicting $(u,v,p)$ from 2601 interior, 400 velocity-boundary, and 100 pressure-boundary points. Figure~\ref{fig:fig17} shows close agreement with the analytical fields, yielding relative $L_2$ errors of $8.20\times10^{-5}$, $3.62\times10^{-5}$, and $2.31\times10^{-5}$ for $u$, $v$, and $p$, respectively. The predicted streamlines in Fig.~\ref{fig:fig18} further reproduce the reference flow structure. Under identical settings, SCV-PINN achieves lower prediction errors than the real-valued PINN (Table~\ref{tab:tab10}), while Fig.~\ref{fig:fig27}(e) indicates more favorable convergence behavior.

\begin{figure}[pos=t]
\normalfont\normalsize\leftskip=0pt\rightskip=0pt\parfillskip=0pt plus 1fil\relax
\refstepcounter{figure}

\begin{center}
\includegraphics[width=\textwidth]{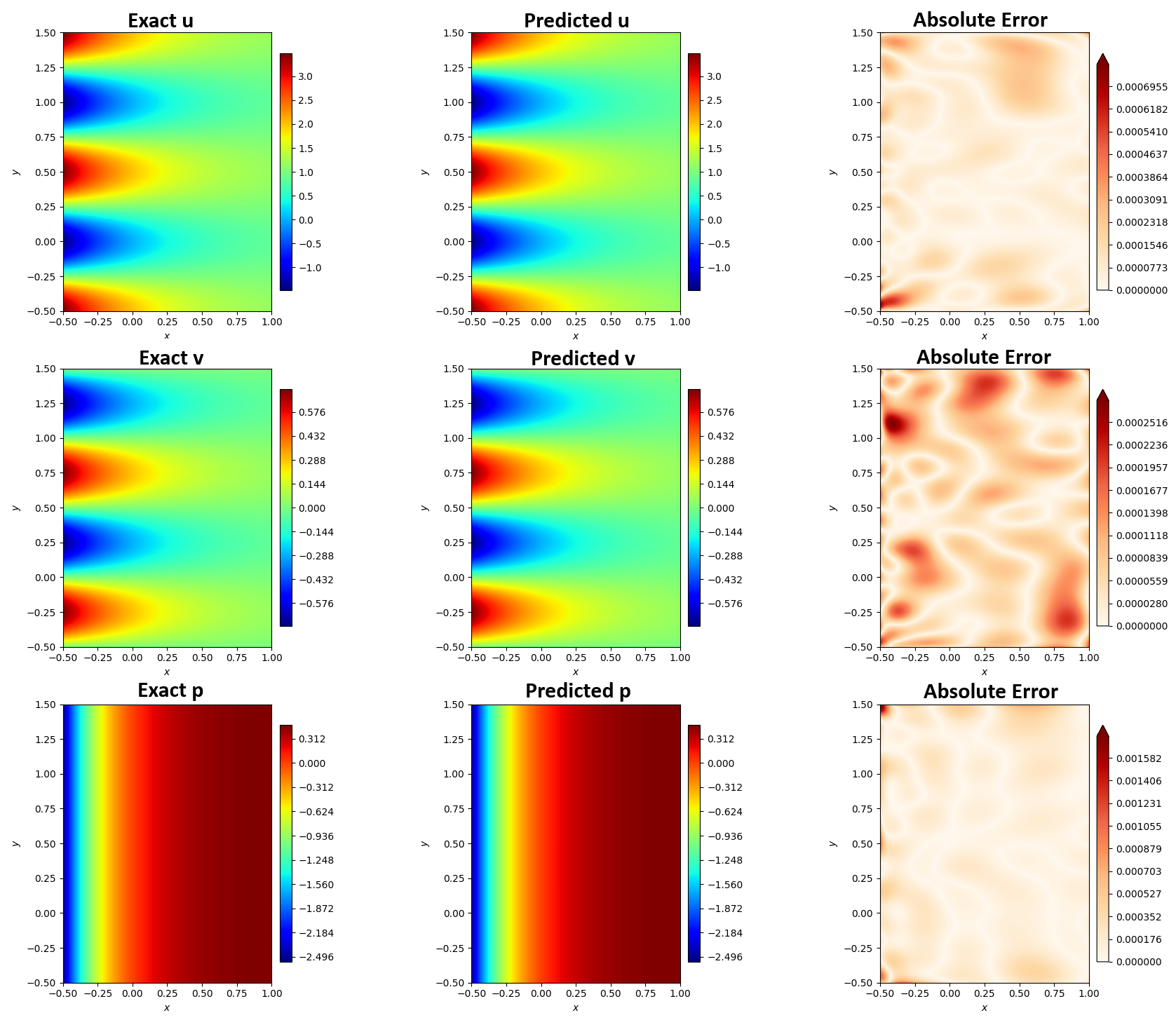}
\end{center}

\vspace{0.5em}

\noindent\textbf{Fig. \thefigure.}
Comparison of the exact and predicted $u$-velocity, $v$-velocity, and pressure fields for the Kovasznay flow at $\mathrm{Re}=20$, along with the corresponding absolute error.

\label{fig:fig17}
\end{figure}

\begin{figure}[pos=t]
\normalfont\normalsize\leftskip=0pt\rightskip=0pt\parfillskip=0pt plus 1fil\relax
\refstepcounter{figure}

\begin{center}
\includegraphics[width=\textwidth]{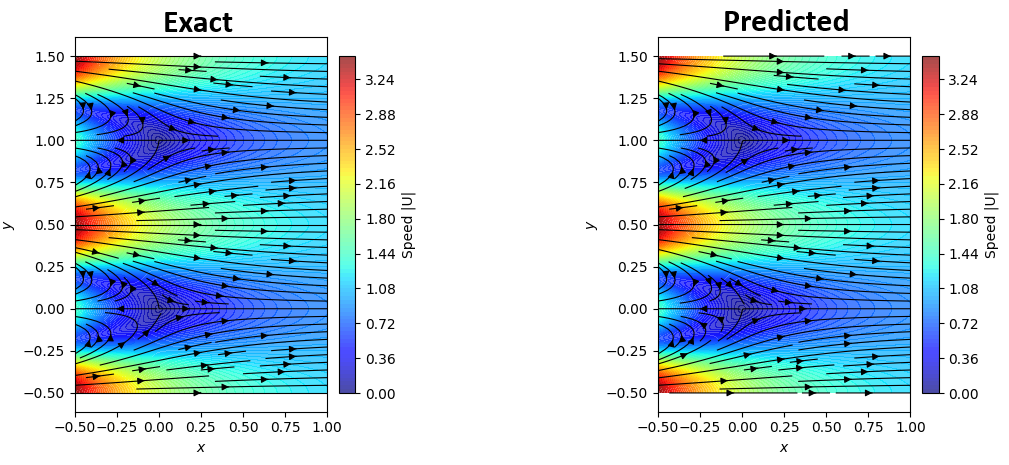}
\end{center}

\vspace{0.5em}

\noindent\textbf{Fig. \thefigure.}
Comparison of the exact and predicted streamline of the Kovasznay flow at $\mathrm{Re}=20$.

\label{fig:fig18}
\end{figure}

\subsection{Forward Problem: Lid-Driven Cavity Flow at $\mathrm{Re}=100$}

We investigate the lid-driven cavity flow at $\mathrm{Re}=100$, governed by the steady incompressible Navier--Stokes equations,
\begin{align}
u_x + v_y &= 0, \\
u u_x + v u_y + \frac{1}{\rho}p_x - \nu (u_{xx}+u_{yy}) &= 0, \\
u v_x + v v_y + \frac{1}{\rho}p_y - \nu (v_{xx}+v_{yy}) &= 0,
\end{align}
where $(u,v)$ are the velocity components, $p$ denotes pressure, $\rho=1$, and $\nu=0.01$. The computational domain is $(x,y)\in[0,1]^2$, with lid velocity $(u,v)=(1,0)$ and no-slip conditions $(u,v)=(0,0)$ on the remaining walls. The SCV-PINN employs eight hidden layers with 20 neurons each and CGELU activation, using $\mathcal{N}_f=20{,}000$ Latin hypercube sampling (LHS) collocation points and $\mathcal{N}_b=4{,}000$ boundary points. Figure~\ref{fig:fig19} illustrates the predicted primary vortex and pressure field, with relative $L_2$ errors of $2.64\times10^{-3}$, $8.53\times10^{-3}$, and $6.04\times10^{-2}$ for $u$, $v$, and $p$, respectively. These values are lower than the corresponding LA-PINN~\cite{song2024loss} errors of $3.84\times10^{-2}$, $4.92\times10^{-2}$, and $2.04\times10^{-1}$ (Table~\ref{tab:tab3}). Under identical settings, SCV-PINN also achieves lower prediction errors than the real-valued PINN (Table~\ref{tab:tab10}), while the loss histories in Fig.~\ref{fig:fig27}(h) indicate a lower final training loss.

\begin{table}[pos=t]
\centering
\caption{
Comparison of relative $L_2$ error for the lid-driven cavity flow problem at $\mathrm{Re}=100$ using LA-PINN~\cite{song2024loss} and the proposed SCV-PINN framework.
}
\label{tab:tab3}

\vspace{0.3em}

\renewcommand{\arraystretch}{1.15}

\begin{tabular}{lccc}
\hline
\textbf{Method} & \textbf{$u$-velocity} & \textbf{$v$-velocity} & \textbf{Pressure} \\
\hline
LA-PINN~\cite{song2024loss} & $3.84 \times 10^{-2}$ & $4.92 \times 10^{-2}$ & $2.04 \times 10^{-1}$ \\

SCV-PINN (Proposed) &
$\mathbf{2.64 \times 10^{-3}}$ &
$\mathbf{8.53 \times 10^{-3}}$ &
$\mathbf{6.04 \times 10^{-2}}$ \\
\hline
\end{tabular}

\end{table}

\begin{figure}[pos=t]
\normalfont\normalsize\leftskip=0pt\rightskip=0pt\parfillskip=0pt plus 1fil\relax
\refstepcounter{figure}

\begin{center}
\includegraphics[width=\textwidth]{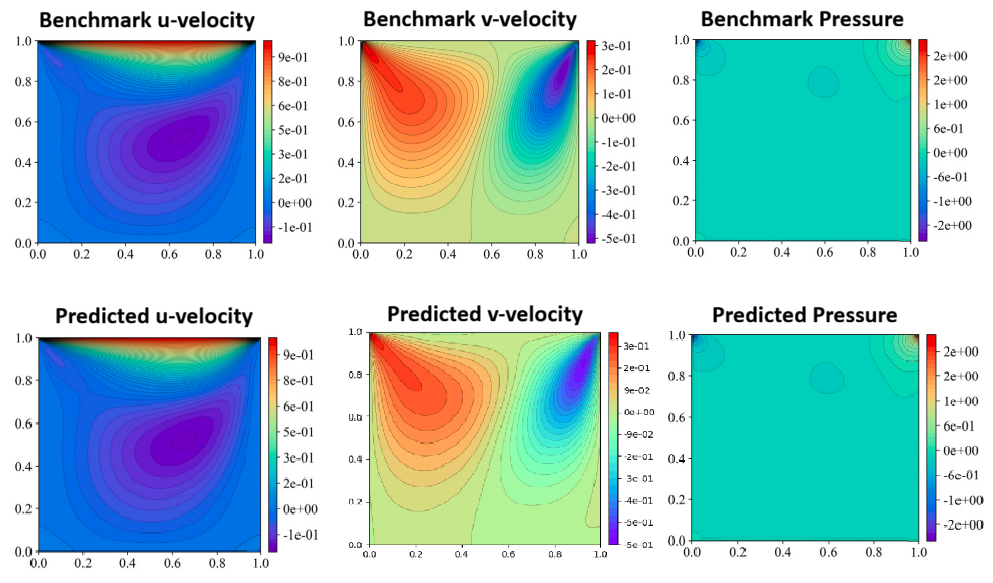}
\end{center}

\vspace{0.5em}

\noindent\textbf{Fig. \thefigure.}
Comparison of the benchmark and predicted $u$-velocity, $v$-velocity, and pressure fields for the lid-driven cavity flow at $\mathrm{Re}=100$. The proposed SCV-PINN accurately captures the incompressible flow structures and pressure distribution. The corresponding relative $L_2$ errors for $u$, $v$, and $p$ are $2.64\times10^{-3}$, $8.53\times10^{-3}$, and $6.04\times10^{-2}$, respectively.

\label{fig:fig19}
\end{figure}

\subsection{Inverse Problem: Lorenz Dynamical System}

We first investigate inverse parameter identification using the Lorenz system,
\begin{align}
\frac{dy_1}{dt}&=\lambda_1(y_2-y_1),
\qquad
\frac{dy_2}{dt}=y_1(\lambda_2-y_3)-y_2,
\qquad
\frac{dy_3}{dt}=y_1y_2-\lambda_3y_3,
\end{align}
with initial conditions $y_1(0)=-8$, $y_2(0)=7$, and $y_3(0)=27$ over $t\in[0,3]$. The reference parameters are $\lambda_1=10$, $\lambda_2=15$, and $\lambda_3=8/3$. The SCV-PINN employs three hidden layers with 40 neurons each and CGELU activation, using Latin hypercube sampling (LHS) residual points and only 25 observations for simultaneous state reconstruction and parameter identification~\cite{lu2021deepxde}. The governing residuals are
\begin{equation}
f_1=\frac{dy_1}{dt}-\lambda_1(y_2-y_1),
\qquad
f_2=\frac{dy_2}{dt}-y_1(\lambda_2-y_3)+y_2,
\qquad
f_3=\frac{dy_3}{dt}-y_1y_2+\lambda_3y_3.
\end{equation}
Figure~\ref{fig:fig20} illustrates the convergence of the identified parameters toward their reference values, while Fig.~\ref{fig:fig21} presents the reconstructed Lorenz trajectory from sparse observations. The identified parameters and governing equations (Tables~\ref{tab:tab4}--\ref{tab:tab5}) show closer agreement with the reference values than those obtained using the conventional PINN. Under identical settings, SCV-PINN achieves lower state and parameter estimation errors (Table~\ref{tab:tab10}), with more favorable convergence behavior observed in Fig.~\ref{fig:fig27}(d).

\begin{table}[pos=t]
\centering
\caption{
Comparison of parameter identification results for the Lorenz system using PINN~\cite{lu2021deepxde} and the proposed SCV-PINN framework.
}
\label{tab:tab4}

\vspace{0.3em}

\renewcommand{\arraystretch}{1.15}

\begin{tabular}{lccc}
\hline
\textbf{Method} & $\boldsymbol{\lambda_1}$ & $\boldsymbol{\lambda_2}$ & $\boldsymbol{\lambda_3}$ \\
\hline
True Values & $10.000$ & $15.000$ & $2.667$ \\

PINN~\cite{lu2021deepxde} & $10.002$ & $14.999$ & $2.668$ \\

SCV-PINN (Proposed) &
$\mathbf{10.001}$ &
$\mathbf{15.000}$ &
$\mathbf{2.667}$ \\
\hline
\end{tabular}

\end{table}

\begin{table}[pos=t]
\centering
\caption{
Comparison of the exact and identified Lorenz system equations obtained using PINN~\cite{lu2021deepxde} and the proposed SCV-PINN framework.
}
\label{tab:tab5}

\vspace{0.3em}

\renewcommand{\arraystretch}{1.35}

\begin{tabular}{|c|c|}
\hline

\textbf{Equation Type} & \textbf{Lorenz System Equations} \\

\hline

Correct System &
$\begin{aligned}
\frac{dy_1}{dt} &= 10.000(y_2 - y_1) \\
\frac{dy_2}{dt} &= y_1(15.000 - y_3) - y_2 \\
\frac{dy_3}{dt} &= y_1y_2 - 2.667y_3
\end{aligned}$ \\

\hline

PINN~\cite{lu2021deepxde} Identified System &
$\begin{aligned}
\frac{dy_1}{dt} &= 10.002(y_2 - y_1) \\
\frac{dy_2}{dt} &= y_1(14.999 - y_3) - y_2 \\
\frac{dy_3}{dt} &= y_1y_2 - 2.668y_3
\end{aligned}$ \\

\hline

SCV-PINN Identified System &
$\begin{aligned}
\frac{dy_1}{dt} &= 10.001(y_2 - y_1) \\
\frac{dy_2}{dt} &= y_1(15.000 - y_3) - y_2 \\
\frac{dy_3}{dt} &= y_1y_2 - 2.667y_3
\end{aligned}$ \\

\hline

\end{tabular}

\end{table}

\begin{figure}[pos=t]
\normalfont\normalsize\leftskip=0pt\rightskip=0pt\parfillskip=0pt plus 1fil\relax
\refstepcounter{figure}

\begin{center}
\includegraphics[width=\textwidth]{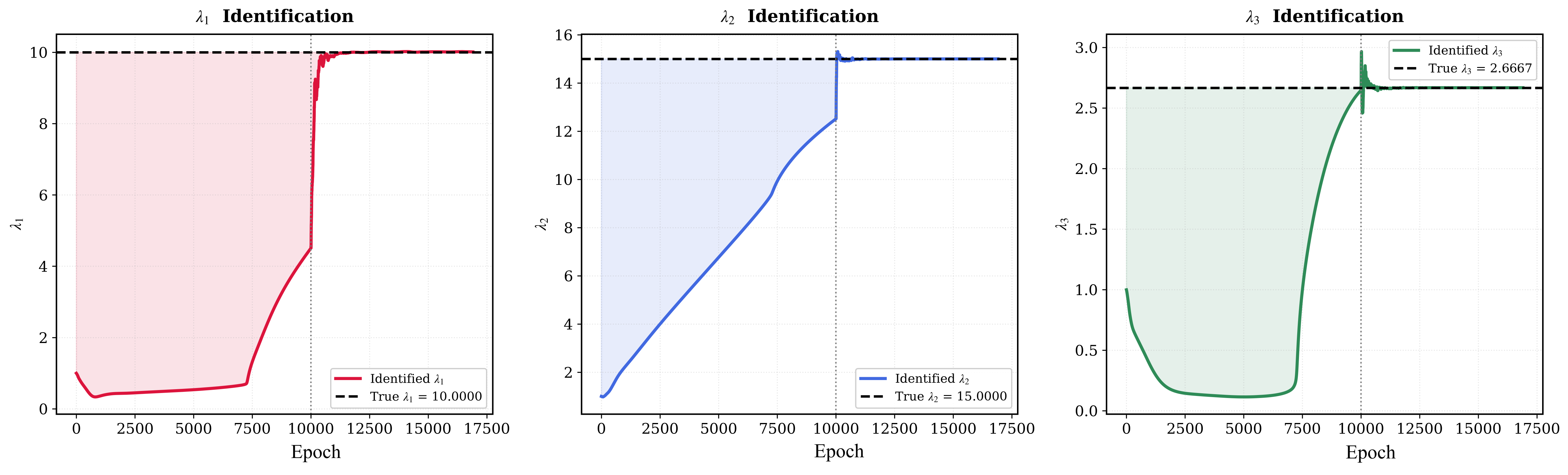}
\end{center}

\vspace{0.5em}

\noindent\textbf{Fig. \thefigure.}
Inverse identification history of the Lorenz system parameters: Prandtl number $(\lambda_1)$, Rayleigh number $(\lambda_2)$, and geometric aspect ratio $(\lambda_3)$ with respect to training epochs. The exact values are $(10.000,\,15.000,\,2.667)$, while the proposed SCV-PINN predicts $(10.001,\,15.000,\,2.667)$, showing excellent agreement with the reference parameters.

\label{fig:fig20}
\end{figure}

\begin{figure}[pos=t]
\normalfont\normalsize\leftskip=0pt\rightskip=0pt\parfillskip=0pt plus 1fil\relax
\refstepcounter{figure}

\begin{center}
\includegraphics[width=\textwidth]{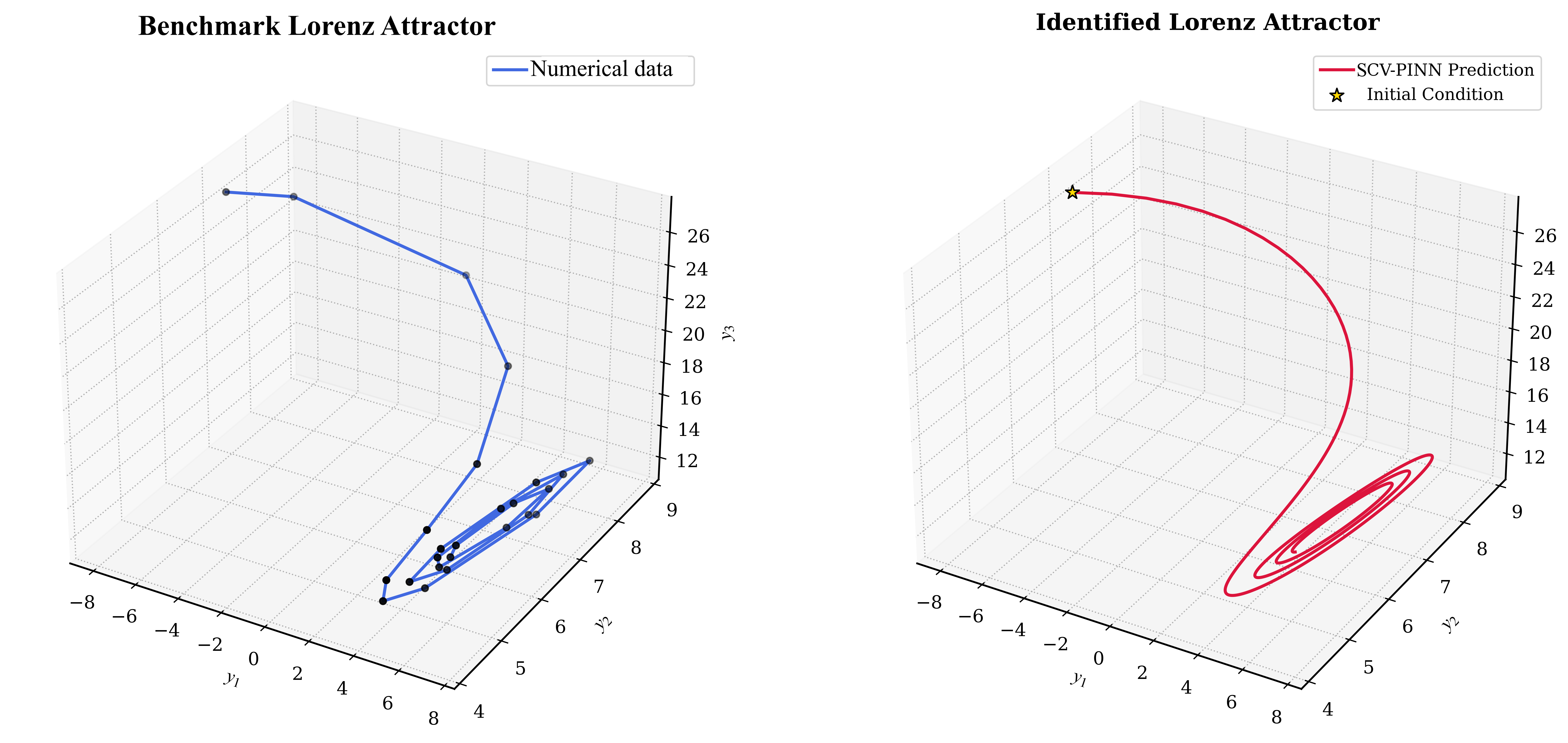}
\end{center}

\vspace{0.5em}

\noindent\textbf{Fig. \thefigure.}
Lorenz system dynamics: left, numerical data points; right, predicted solution trajectory obtained using the proposed SCV-PINN. The proposed framework accurately captures the chaotic butterfly-shaped attractor and nonlinear temporal evolution.

\label{fig:fig21}
\end{figure}

\subsection{Inverse Problem: Burgers' Equation}

We next examine inverse parameter identification for the Burgers' equation,
\begin{equation}
u_t+\lambda_1uu_x-\lambda_2u_{xx}=0,
\end{equation}
on $x\in[-1,1]$ and $t\in[0,1]$, with reference coefficients $\lambda_1=1.0$ and $\lambda_2=0.01/\pi$. The SCV-PINN employs seven hidden layers with 20 neurons each and CTanh activation, using $\mathcal{N}_u=2000$ observations and sequential Adam--L-BFGS optimization. Table~\ref{tab:tab6} presents the identified equation, showing close agreement with the reference coefficients despite using seven hidden layers compared with nine in the reference PINN. Figure~\ref{fig:fig22} compares the predicted solution, absolute error, and profiles at $t=0.25$, $0.50$, and $0.75\,\mathrm{s}$ with the reference data. For a controlled comparison using eight hidden layers and 5000 Adam followed by 5000 L-BFGS iterations, Fig.~\ref{fig:fig23} shows that the real-valued PINN estimate of $\lambda_2$ remains near $4.8\times10^{-3}$, whereas SCV-PINN approaches the reference value $0.01/\pi$. The corresponding results in Table~\ref{tab:tab10} indicate lower estimation errors for $\lambda_1$, $\lambda_2$, and $u$, while Fig.~\ref{fig:fig28}(b) shows a faster reduction in training loss for SCV-PINN.

\begin{table}[pos=h]
\centering
\caption{Comparison of the identified Burgers' equation using the conventional PINN and the proposed SCV-PINN frameworks. For a fair comparison, the PINN~\cite{raissi2019physics} model employs 9 hidden layers with 20 neurons per layer, whereas the proposed SCV-PINN uses 7 hidden layers with 20 neurons per layer.}
\label{tab:tab6}
\renewcommand{\arraystretch}{1.2}
\begin{tabular}{|c|c|}
\hline
\textbf{Method} & \textbf{Identified PDE} \\
\hline

Correct PDE 
&
$
u_t + u u_x - 0.0031831\,u_{xx} = 0
$
\\
\hline

PINN~\cite{raissi2019physics} 
(9 hidden layers, 20 neurons/layer)
&
$
u_t + 0.99915\,u u_x - 0.0031794\,u_{xx} = 0
$
\\
\hline

Proposed SCV-PINN 
(7 hidden layers, 20 neurons/layer)
&
$
u_t + 1.00005\,u u_x - 0.0031864\,u_{xx} = 0
$
\\
\hline

\end{tabular}
\end{table}

\begin{figure}[pos=t]
\normalfont\normalsize\leftskip=0pt\rightskip=0pt\parfillskip=0pt plus 1fil\relax
\refstepcounter{figure}

\begin{center}
\includegraphics[width=\textwidth]{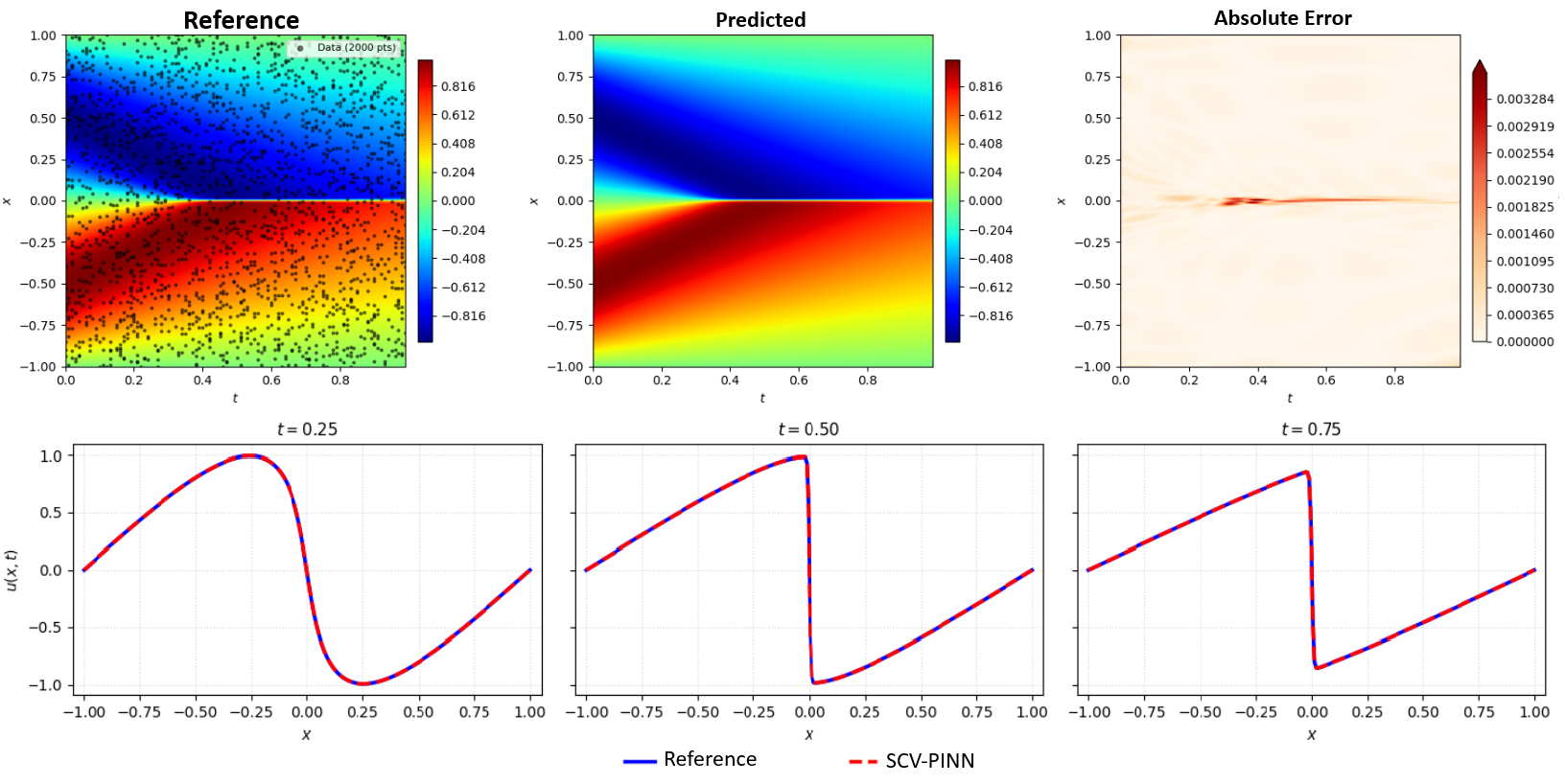}
\end{center}

\vspace{0.5em}

\noindent\textbf{Fig. \thefigure.}
Inverse Burgers' equation results obtained using the proposed SCV-PINN with a $7$-hidden-layer and $20$-neurons-per-layer architecture. The top row shows the exact solution with $2000$ observation data points (left), predicted solution (middle), and corresponding absolute error (right). The bottom row presents the reference and predicted solution comparisons at $t=0.25\,\mathrm{s}$, $t=0.50\,\mathrm{s}$, and $t=0.75\,\mathrm{s}$.

\label{fig:fig22}
\end{figure}

\begin{figure}[pos=t]
\normalfont\normalsize\leftskip=0pt\rightskip=0pt\parfillskip=0pt plus 1fil\relax
\refstepcounter{figure}

\begin{center}
\includegraphics[width=\textwidth]{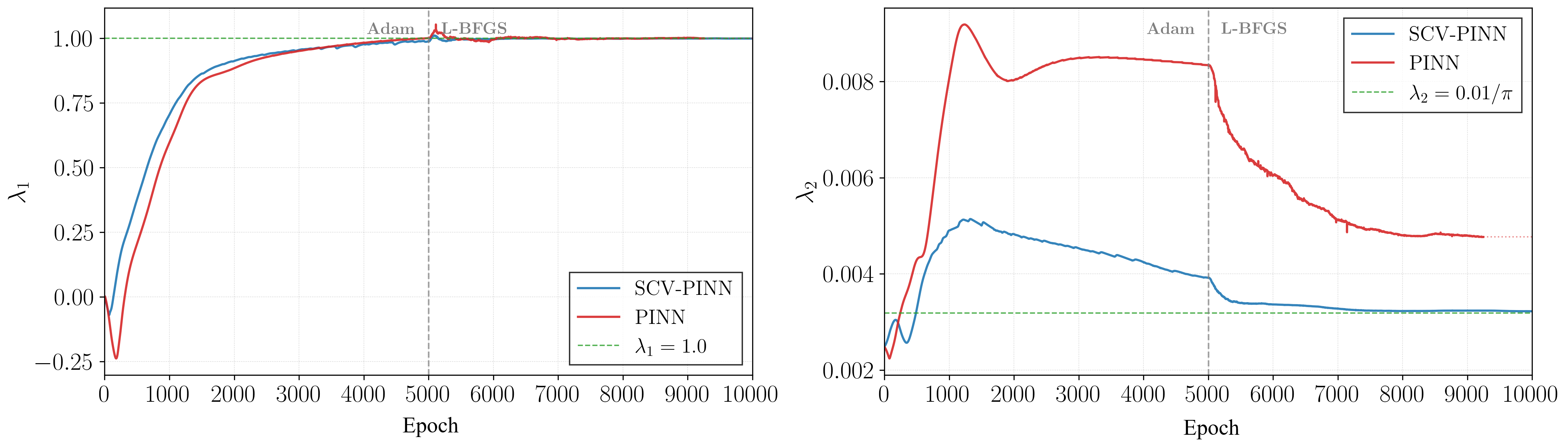}
\end{center}

\vspace{0.5em}

\noindent\textbf{Fig. \thefigure.}
Convergence history of the Burgers' inverse problem parameters: convection coefficient $(\lambda_1)$ and diffusion coefficient $(\lambda_2)$ with respect to training epochs. The comparison between PINN and the proposed SCV-PINN is performed using the same architecture consisting of $8$ hidden layers with $20$ neurons per layer.

\label{fig:fig23}
\end{figure}

\subsection{Inverse Problem: Navier--Stokes Equation}

Finally, we investigate inverse parameter identification for the two-dimensional incompressible Navier--Stokes equations,
\begin{align}
u_t+\lambda_1(uu_x+vu_y)
&=
-p_x+\lambda_2(u_{xx}+u_{yy}),
\\
v_t+\lambda_1(uv_x+vv_y)
&=
-p_y+\lambda_2(v_{xx}+v_{yy}),
\end{align}
subject to $u_x+v_y=0$, where $(u,v,p)$ denote the velocity and pressure fields, and the unknown coefficients have reference values $\lambda_1=1.0$ and $\lambda_2=0.01$. The computational domain is $x\in[1,8]$, $y\in[-2,2]$, and $t\in[0,7]$, with incompressibility enforced through the stream-function formulation $u=\psi_y$ and $v=-\psi_x$. The SCV-PINN employs nine hidden layers with 20 neurons each and CTanh activation, using 5000 observations and 5000 collocation points for simultaneous field reconstruction and parameter identification. Table~\ref{tab:tab7} compares the identified coefficients with PINN~\cite{raissi2019physics} using the same network size. Both approaches estimate $\lambda_1$ within $0.14\%$, whereas SCV-PINN reduces the viscosity identification error to $0.7\%$, compared with $4.7\%$ for PINN. The corresponding identified equations are reported in Table~\ref{tab:tab8}. Figure~\ref{fig:fig24} presents the parameter convergence histories, while Fig.~\ref{fig:fig25} shows close agreement between the predicted and reference $(u,v,p)$ fields. Under identical settings, SCV-PINN achieves lower field and parameter estimation errors (Table~\ref{tab:tab10}), with faster and more stable convergence observed in Fig.~\ref{fig:fig27}(c).

\begin{table}[pos=t]
\centering
\caption{
Comparison of parameter identification results for the inverse Navier--Stokes equation using PINN~\cite{raissi2019physics} and the proposed SCV-PINN framework.
}
\label{tab:tab7}

\vspace{0.3em}

\renewcommand{\arraystretch}{1.15}

\begin{tabular}{lcc}
\hline
\textbf{Method} & \textbf{Convection ($\lambda_1$)} & \textbf{Viscosity ($\lambda_2$)} \\
\hline
True Values & $1.00000$ & $0.01000$ \\

PINN~\cite{raissi2019physics} & $\mathbf{0.99900}$ & $0.01047$ \\

SCV-PINN (Proposed) &
$1.00137$ &
$\mathbf{0.01007}$ \\
\hline
\end{tabular}

\end{table}

\begin{table}[pos=t]
\centering
\caption{
Comparison of the exact and identified inverse Navier--Stokes equations obtained using PINN~\cite{raissi2019physics} and the proposed SCV-PINN framework.
}
\label{tab:tab8}

\vspace{0.3em}

\renewcommand{\arraystretch}{1.35}

\begin{tabular}{|c|c|}
\hline

\textbf{Equation Type} & \textbf{Identified PDE System} \\

\hline

Correct PDE &
$\begin{aligned}
u_t + (u u_x + v u_y) &= -p_x + 0.01(u_{xx} + u_{yy}) \\
v_t + (u v_x + v v_y) &= -p_y + 0.01(v_{xx} + v_{yy})
\end{aligned}$ \\

\hline

PINN~\cite{raissi2019physics} Identified PDE &
$\begin{aligned}
u_t + 0.999(u u_x + v u_y) &= -p_x + 0.01047(u_{xx} + u_{yy}) \\
v_t + 0.999(u v_x + v v_y) &= -p_y + 0.01047(v_{xx} + v_{yy})
\end{aligned}$ \\

\hline

SCV-PINN Identified PDE &
$\begin{aligned}
u_t + 1.001(u u_x + v u_y) &= -p_x + 0.01007(u_{xx} + u_{yy}) \\
v_t + 1.001(u v_x + v v_y) &= -p_y + 0.01007(v_{xx} + v_{yy})
\end{aligned}$ \\

\hline

\end{tabular}

\end{table}

\begin{figure}[pos=t]
\normalfont\normalsize\leftskip=0pt\rightskip=0pt\parfillskip=0pt plus 1fil\relax
\refstepcounter{figure}

\begin{center}
\includegraphics[width=\textwidth]{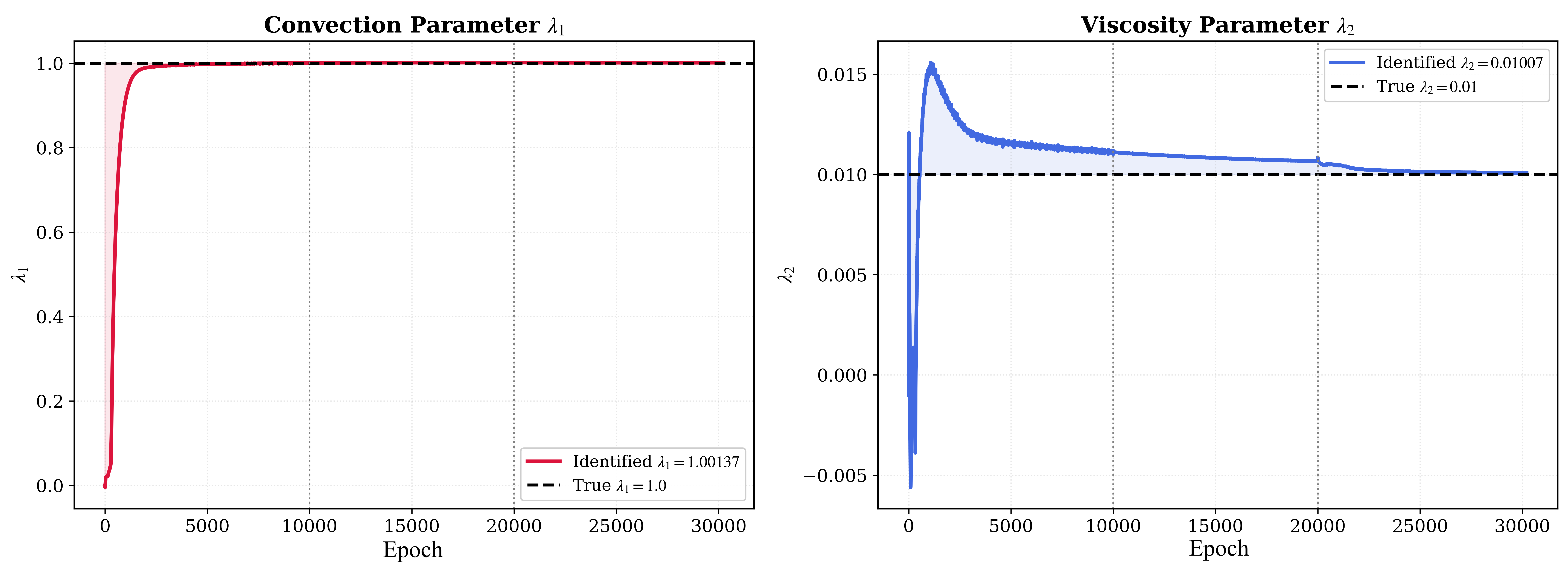}
\end{center}

\vspace{0.5em}

\noindent\textbf{Fig. \thefigure.}
Convergence history of the identified parameters $(\lambda_1,\lambda_2)$ with respect to training epochs for the two-dimensional Navier--Stokes inverse problem. The left and right panels correspond to $\lambda_1$ and $\lambda_2$, respectively. The identified values are $\lambda_1=1.00137$ and $\lambda_2=0.01007$, while the exact values are $\lambda_1=1.0$ and $\lambda_2=0.01$.

\label{fig:fig24}
\end{figure}

\begin{figure}[pos=t]
\normalfont\normalsize\leftskip=0pt\rightskip=0pt\parfillskip=0pt plus 1fil\relax
\refstepcounter{figure}

\begin{center}
\includegraphics[width=\textwidth]{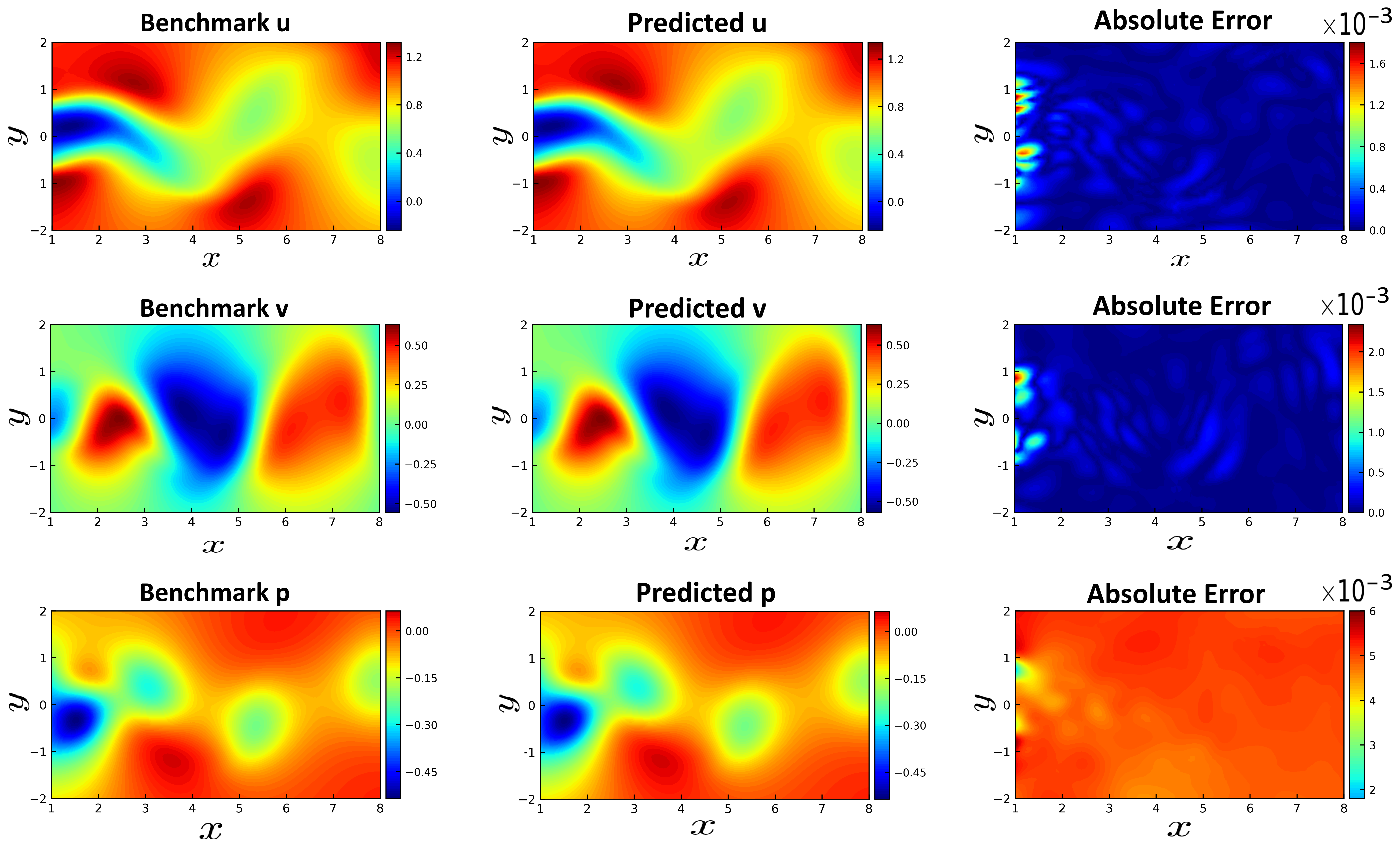}
\end{center}

\vspace{0.5em}

\noindent\textbf{Fig. \thefigure.}
Inverse Navier--Stokes flow field reconstruction at $\mathrm{Re}=100$: comparison of the benchmark and predicted $u$-velocity, $v$-velocity, and pressure fields together with their corresponding absolute error distributions obtained using the proposed SCV-PINN.

\label{fig:fig25}
\end{figure}

\subsection{3D Navier--Stokes Equation: Beltrami Flow}

To examine the applicability of SCV-PINN to three-dimensional nonlinear flows, we consider the incompressible Navier--Stokes equations,
\begin{equation}
\frac{\partial \mathbf{u}}{\partial t}
+
(\mathbf{u}\cdot\nabla)\mathbf{u}
=
-\frac{1}{\rho_0}\nabla p
+
\nu\nabla^2\mathbf{u},
\end{equation}
subject to
\begin{equation}
\nabla\cdot\mathbf{u}=0,
\end{equation}
where $\mathbf{u}=(u,v,w)$, $p$, $\rho_0$, and $\nu$ denote the velocity, pressure, constant density, and kinematic viscosity, respectively. The analytical Beltrami-flow solution~\cite{ethier1994exact} is
\begin{align}
u(x,y,z,t)
&=
-a
\left[
e^{ax}\sin(ay+dz)
+
e^{az}\cos(ax+dy)
\right]
e^{-d^{2}t},
\\
v(x,y,z,t)
&=
-a
\left[
e^{ay}\sin(az+dx)
+
e^{ax}\cos(ay+dz)
\right]
e^{-d^{2}t},
\\
w(x,y,z,t)
&=
-a
\left[
e^{az}\sin(ax+dy)
+
e^{ay}\cos(az+dx)
\right]
e^{-d^{2}t}.
\end{align}
The corresponding pressure is
\begin{align}
p(x,y,z,t)
=
-\frac{a^{2}}{2}
\Big[
&
e^{2ax}
+
e^{2ay}
+
e^{2az}
\nonumber\\
&
+2e^{a(y+z)}
\sin(ax+dy)\cos(az+dx)
\nonumber\\
&
+2e^{a(z+x)}
\sin(ay+dz)\cos(ax+dy)
\nonumber\\
&
+2e^{a(x+y)}
\sin(az+dx)\cos(ay+dz)
\Big]
e^{-2d^{2}t}.
\end{align}
The normalized coordinates $(x,y,z,t)$ are embedded into the complex domain, and the SCV-PINN employs two hidden layers with 40 neurons each and CGELU activation. Dirichlet conditions are prescribed for all velocity components, while pressure is anchored at a reference point. Training consists of $10{,}000$ Adam iterations followed by $20{,}000$ L-BFGS iterations. Figure~\ref{fig:fig26} shows close agreement between the predicted and analytical velocity and pressure fields, yielding a relative $L_2$ error of $4.07\times10^{-5}$ and demonstrating the applicability of SCV-PINN to this three-dimensional flow benchmark.

\begin{table}[pos=h]
\centering
\caption{Benchmark problems considered for evaluating the proposed SCV-PINN framework together with the corresponding relative $L_2$ errors.}
\label{tab:tab9}

\vspace{0.3em}

\renewcommand{\arraystretch}{1.15}
\begin{tabular}{lc}
\hline
\textbf{Benchmark Problem} & \textbf{relative $L_2$ error} \\
\hline
Burgers' equation                          & $6.73 \times 10^{-5}$ \\
Allen--Cahn equation                       & $4.39 \times 10^{-4}$ \\
Korteweg--de Vries (KdV) equation          & $5.54 \times 10^{-4}$ \\
Nonlinear Schrödinger equation             & $8.45 \times 10^{-5}$ \\
Helmholtz equation (regular domain)        & $2.17 \times 10^{-4}$ \\
Helmholtz equation (irregular domain)      & $1.17 \times 10^{-4}$ \\
Poisson equation (regular domain)          & $4.45 \times 10^{-7}$ \\
Poisson equation (irregular domain)        & $1.26 \times 10^{-2}$ \\
Kovasznay flow (u) ($\mathrm{Re}=20$)                   &$8.20 \times 10^{-5}$ \\
Lid-driven cavity flow (u) ($\mathrm{Re}=100$)          & $2.64 \times 10^{-3}$ \\
\hline
\end{tabular}
\end{table}

\begin{table}[pos=t]
\centering
\caption{
Comparison of the proposed SCV-PINN and conventional real-valued PINN~\cite{raissi2019physics} across different PDE benchmark problems. Bold values indicate better performance. All models were trained using Adam optimization for 5000 iterations followed by L-BFGS optimization for 5000 iterations with identical hyperparameters to ensure a fair comparison.
}
\label{tab:tab10}

\vspace{0.3em}

\renewcommand{\arraystretch}{1.10}
\scriptsize

\begin{tabular}{llcccc}
\hline
\textbf{Problem} & \textbf{Field} &
\multicolumn{2}{c}{\textbf{MSE}} &
\multicolumn{2}{c}{\textbf{relative $L_2$ error}} \\
\cline{3-6}
 & & \textbf{SCV-PINN} & \textbf{PINN~\cite{raissi2019physics}} & \textbf{SCV-PINN} & \textbf{PINN~\cite{raissi2019physics}} \\
\hline

Helmholtz & $u$ &
$\mathbf{9.34 \times 10^{-6}}$ & $1.57 \times 10^{-4}$ &
$\mathbf{6.14 \times 10^{-3}}$ & $2.52 \times 10^{-2}$ \\

Helmholtz (hole) & $u$ &
$\mathbf{7.75 \times 10^{-6}}$ & $5.62 \times 10^{-5}$ &
$\mathbf{5.31 \times 10^{-3}}$ & $1.43 \times 10^{-2}$ \\

\hline

Kovasznay & $u$ &
$\mathbf{2.36 \times 10^{-8}}$ & $2.15 \times 10^{-6}$ &
$\mathbf{1.22 \times 10^{-4}}$ & $1.17 \times 10^{-3}$ \\

Kovasznay & $v$ &
$\mathbf{1.95 \times 10^{-8}}$ & $4.24 \times 10^{-7}$ &
$\mathbf{6.46 \times 10^{-4}}$ & $3.01 \times 10^{-3}$ \\

Kovasznay & $p$ &
$\mathbf{7.21 \times 10^{-8}}$ & $1.49 \times 10^{-6}$ &
$\mathbf{3.58 \times 10^{-4}}$ & $1.63 \times 10^{-3}$ \\

\hline

Lid cavity & $u$ &
$\mathbf{1.48 \times 10^{-4}}$ & $2.03 \times 10^{-4}$ &
$\mathbf{2.67 \times 10^{-2}}$ & $3.12 \times 10^{-2}$ \\

Lid cavity & $v$ &
$\mathbf{1.53 \times 10^{-4}}$ & $2.17 \times 10^{-4}$ &
$\mathbf{9.27 \times 10^{-2}}$ & $1.10 \times 10^{-1}$ \\

\hline

Burgers (forward) & $u$ &
$\mathbf{4.48 \times 10^{-5}}$ & $8.76 \times 10^{-2}$ &
$\mathbf{1.09 \times 10^{-2}}$ & $4.82 \times 10^{-1}$ \\

\hline

Burgers inverse & $u$ &
$\mathbf{2.64 \times 10^{-6}}$ & $1.62 \times 10^{-3}$ &
$\mathbf{2.64 \times 10^{-3}}$ & $6.54 \times 10^{-2}$ \\

Burgers inverse & $\lambda_1$ &
$\mathbf{5.61 \times 10^{-7}}$ & $2.99 \times 10^{-6}$ &
$\mathbf{7.49 \times 10^{-4}}$ & $1.73 \times 10^{-3}$ \\

Burgers inverse & $\lambda_2$ &
$\mathbf{1.44 \times 10^{-9}}$ & $2.51 \times 10^{-6}$ &
$\mathbf{1.19 \times 10^{-2}}$ & $4.97 \times 10^{-1}$ \\





\hline

Lorenz inverse & $y_1$ &
$\mathbf{4.74 \times 10^{-6}}$ & $1.29 \times 10^{-3}$ &
$\mathbf{3.74 \times 10^{-4}}$ & $6.17 \times 10^{-3}$ \\

Lorenz inverse & $y_2$ &
$\mathbf{6.35 \times 10^{-6}}$ & $1.70 \times 10^{-3}$ &
$\mathbf{4.02 \times 10^{-4}}$ & $6.58 \times 10^{-3}$ \\

Lorenz inverse & $y_3$ &
$\mathbf{8.26 \times 10^{-6}}$ & $3.15 \times 10^{-3}$ &
$\mathbf{1.79 \times 10^{-4}}$ & $3.50 \times 10^{-3}$ \\

Lorenz inverse & $\lambda_1$ &
$\mathbf{1.61 \times 10^{-5}}$ & $1.97 \times 10^{-3}$ &
$\mathbf{4.01 \times 10^{-4}}$ & $4.44 \times 10^{-3}$ \\

Lorenz inverse & $\lambda_2$ &
$\mathbf{2.72 \times 10^{-7}}$ & $6.93 \times 10^{-5}$ &
$\mathbf{3.48 \times 10^{-5}}$ & $5.55 \times 10^{-4}$ \\

Lorenz inverse & $\lambda_3$ &
$\mathbf{1.06 \times 10^{-6}}$ & $2.69 \times 10^{-5}$ &
$\mathbf{3.86 \times 10^{-4}}$ & $1.95 \times 10^{-3}$ \\

\hline

Navier--Stokes inverse & $u$ &
$\mathbf{8.44 \times 10^{-6}}$ & $1.09 \times 10^{-4}$ &
$\mathbf{3.21 \times 10^{-3}}$ & $1.15 \times 10^{-2}$ \\

Navier--Stokes inverse & $v$ &
$\mathbf{1.55 \times 10^{-5}}$ & $6.38 \times 10^{-5}$ &
$\mathbf{1.44 \times 10^{-2}}$ & $2.92 \times 10^{-2}$ \\

Navier--Stokes inverse & $p$ &
$\mathbf{5.44 \times 10^{-4}}$ & $9.47 \times 10^{-3}$ &
$\mathbf{1.77 \times 10^{-1}}$ & $7.39 \times 10^{-1}$ \\

Navier--Stokes inverse & $\lambda_1$ &
$\mathbf{1.98 \times 10^{-6}}$ & $6.48 \times 10^{-6}$ &
$\mathbf{1.41 \times 10^{-3}}$ & $2.55 \times 10^{-3}$ \\

Navier--Stokes inverse & $\lambda_2$ &
$\mathbf{8.04 \times 10^{-8}}$ & $2.30 \times 10^{-6}$ &
$\mathbf{2.84 \times 10^{-2}}$ & $1.52 \times 10^{-1}$ \\

\hline

Poisson(irregular/L-shape) & $u$ &
$\mathbf{8.84 \times 10^{-5}}$ & $1.41 \times 10^{-4}$ &
$\mathbf{1.26 \times 10^{-2}}$ & $1.59 \times 10^{-1}$ \\

\hline
Schrödinger & $\psi$ &
$\mathbf{8.83 \times 10^{-7}}$ & $2.91 \times 10^{-3}$ &
$\mathbf{1.05 \times 10^{-3}}$ & $6.04 \times 10^{-2}$ \\

\hline
\end{tabular}

\end{table}

\begin{figure}[pos=t]
\normalfont\normalsize\leftskip=0pt\rightskip=0pt\parfillskip=0pt plus 1fil\relax
\refstepcounter{figure}

\begin{center}
\includegraphics[width=\textwidth]{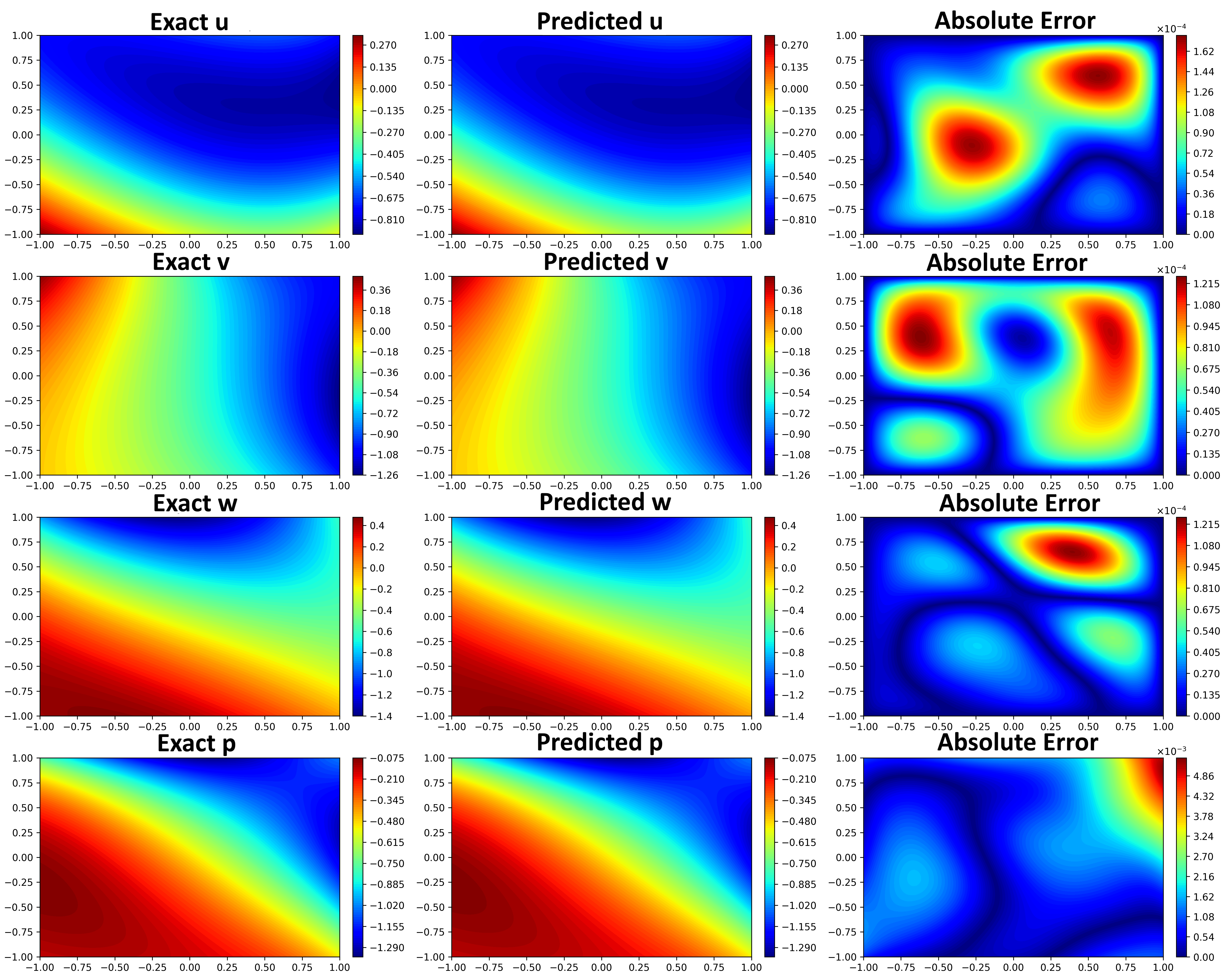}
\end{center}

\vspace{0.5em}

\noindent\textbf{Fig. \thefigure.}
Comparison of the exact and predicted $u$-velocity, $v$-velocity, $w$-velocity, and pressure fields together with their corresponding absolute error distributions for the three-dimensional Navier--Stokes Beltrami flow obtained using the proposed SCV-PINN.

\label{fig:fig26}
\end{figure}

\begin{figure}[pos=htbp]
\centering
\includegraphics[width=0.9\textwidth]{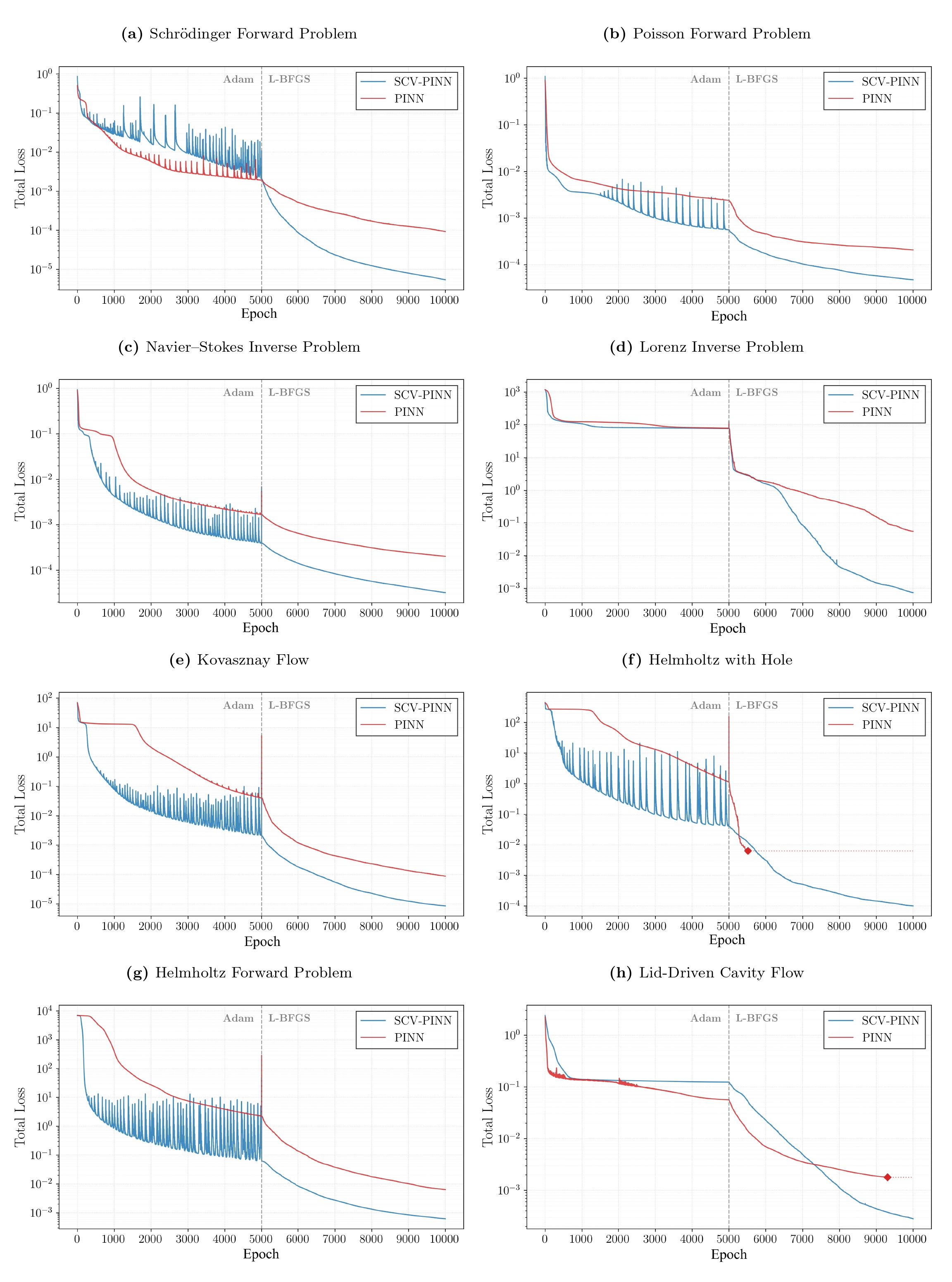}
\caption{Training loss convergence histories of the conventional PINN and the proposed SCV-PINN frameworks. The proposed SCV-PINN consistently exhibits faster convergence and lower final loss values across a range of benchmark problems.}

\label{fig:fig27}
\end{figure}

\begin{figure}[pos=htbp]
\centering
\includegraphics[width=\textwidth]{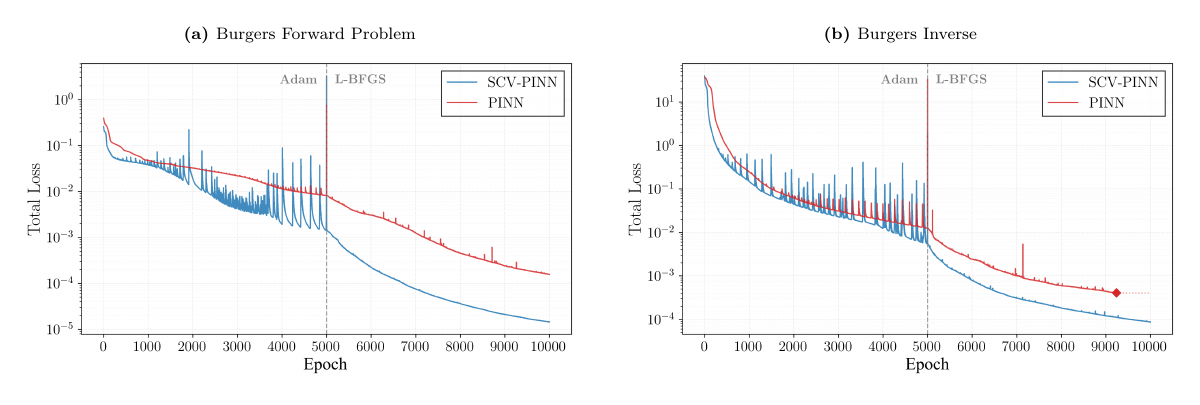}
\caption{Training loss convergence histories obtained using the conventional PINN and the proposed SCV-PINN frameworks for (a) the Burgers' forward problem and (b) the Burgers' inverse problem.}

\label{fig:fig28}
\end{figure}

\begin{figure}[pos=htbp]
\centering
\includegraphics[width=0.5\textwidth]{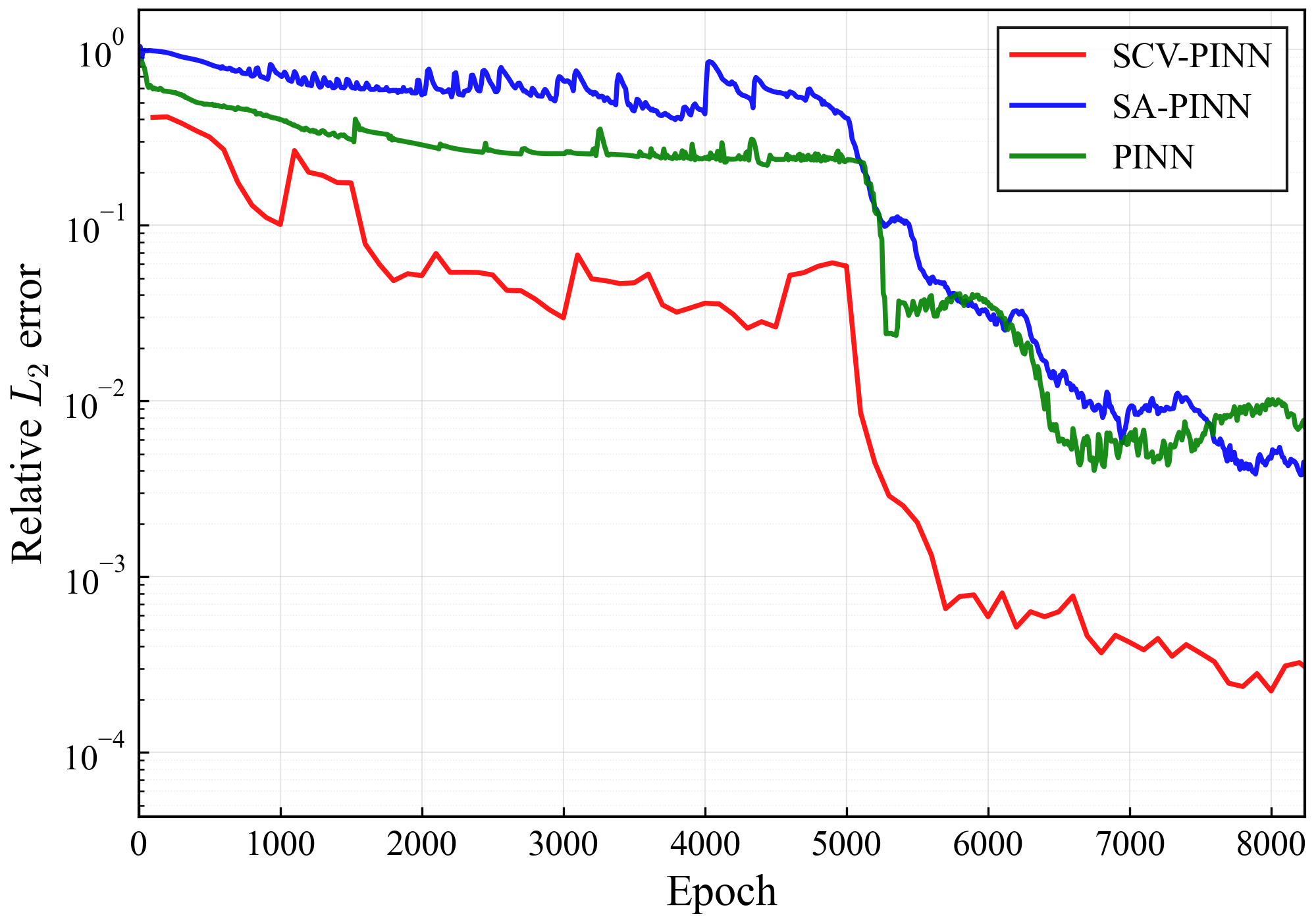}
\caption{Comparison of the relative $L_2$ error obtained using the conventional PINN, SA-PINN, and the proposed SCV-PINN frameworks under identical network architectures and training hyperparameters for the Burgers' equation (forward problem).}
\label{fig:fig29}
\end{figure}

\renewcommand{\thesection}{\arabic{section}}
\setcounter{section}{4}
\section{Conclusion}
\label{sec:sec6}
In this work, we introduced a generalized split complex-valued physics-informed neural network (SCV-PINN) framework for solving forward and inverse problems governed by ordinary and partial differential equations. By integrating complex-valued trainable parameters and latent representations with split activation functions, the proposed framework provides an alternative physics-informed formulation that incorporates coupled real and imaginary representations, offering greater flexibility in capturing oscillatory and phase-dependent solution features. Its computational capabilities were systematically investigated across a broad range of benchmark problems, including the Burgers, Allen--Cahn, Korteweg--de Vries, nonlinear Schrödinger, Helmholtz, and Poisson equations on regular and irregular domains, Kovasznay and lid-driven cavity flows, inverse Lorenz, Burgers, and Navier--Stokes problems, and three-dimensional Navier--Stokes Beltrami flow. The numerical results demonstrate that SCV-PINNs consistently provide high-fidelity solution approximations, reduced relative $L_2$ errors, improved parameter identification, and favorable convergence behavior compared with conventional real-valued PINNs and selected existing variants across the investigated configurations. While established PINN formulations remain effective computational tools for numerous physical problems, the proposed complex-valued representation offers complementary advantages, particularly for nonlinear, oscillatory, and phase-dependent systems. The results further demonstrate its applicability to irregular computational domains, sparse-data inverse problems, and three-dimensional flows. Ablation studies highlight the influence of split activation functions and collocation-point sampling strategies on training stability and predictive performance, with smooth activations such as CTanh and CGELU supporting stable L-BFGS refinement. Overall, the findings demonstrate the potential of split complex-valued representations to extend the computational capabilities of physics-informed learning without altering its underlying physical constraints. Future investigations will focus on multiphysics coupling, turbulent flow reconstruction, integration with neural operator architectures, and computational scalability for large-scale scientific applications.

\appendix
\section{Appendix A: Ablation Study of Split Complex-Valued Activation Functions}
\label{appen:appenA}

\begin{table}[pos=h]
\centering
\caption{Performance comparison of different split complex activation functions for the Burgers' equation using the proposed SCV-PINN framework.}
\label{tab:tab11}

\vspace{0.3em}

\renewcommand{\arraystretch}{1.15}

\begin{tabular}{lcc}
\hline
\textbf{Activation} & \textbf{MSE} & \textbf{relative $L_2$ error} \\
\hline
CGELU     & $5.842 \times 10^{-5}$ & $1.244 \times 10^{-2}$ \\
CTanh     & $\mathbf{4.730 \times 10^{-5}}$ & $\mathbf{1.119 \times 10^{-2}}$ \\
CSwish    & $1.246 \times 10^{-4}$ & $1.817 \times 10^{-2}$ \\
CSigmoid  & $8.577 \times 10^{-5}$ & $1.508 \times 10^{-2}$ \\
CSoftplus & $3.074 \times 10^{-3}$ & $9.024 \times 10^{-2}$ \\
CReLU     & $1.882 \times 10^{-1}$ & $7.062 \times 10^{-1}$ \\
$z$-ReLU  & $2.068 \times 10^{-1}$ & $7.402 \times 10^{-1}$ \\
\hline
\end{tabular}

\end{table}

The split activation functions CGELU, CTanh, CSwish, CSigmoid, CSoftplus, CReLU, and $z$-ReLU are compared for the Burgers' equation in Table~\ref{tab:tab11} and Fig.~\ref{fig:fig30}. All runs use the same 8-hidden-layer CV-MLP with 20 neurons per layer, LHS with $\mathcal{N}_u = 100$ and $\mathcal{N}_f = 10{,}000$, and 5000 Adam iterations followed by 5000 L-BFGS iterations.

CTanh gives the lowest error in this setting, while CGELU converges stably and performs well on several other benchmarks; therefore, both are used in this study depending on the governing equations. CReLU and $z$-ReLU fail to converge, and the loss stagnates during the L-BFGS stage. Since L-BFGS relies on curvature information to approximate the Hessian, the lack of continuous second derivatives in ReLU-type activations limits its convergence~\cite{nocedal2006numerical}.

\section{Appendix B: Comparison of Sampling Strategies for Burgers' Equation}
\label{appen:appenB}

\begin{table}[pos=h]
\centering
\caption{Performance comparison of different collocation point sampling strategies for the Burgers' equation using the proposed SCV-PINN framework.}
\label{tab:tab12}

\vspace{0.3em}

\renewcommand{\arraystretch}{1.15}
\begin{tabular}{lcc}
\hline
\textbf{Sampling Method} & \textbf{MSE} & \textbf{relative $L_2$ error} \\
\hline
LHS        & $5.842 \times 10^{-5}$ & $1.244 \times 10^{-2}$ \\
Halton     & $7.559 \times 10^{-5}$ & $1.415 \times 10^{-2}$ \\
Hammersley & $5.332 \times 10^{-5}$ & $1.189 \times 10^{-2}$ \\
Sobol      & $4.919 \times 10^{-5}$ & $1.142 \times 10^{-2}$ \\
Grid       & $2.137 \times 10^{-2}$ & $2.380 \times 10^{-1}$ \\
Random     & $\mathbf{4.483 \times 10^{-5}}$ & $\mathbf{1.090 \times 10^{-2}}$ \\
RAD        & $3.057 \times 10^{-2}$ & $2.846 \times 10^{-1}$ \\
\hline
\end{tabular}
\end{table}

The non-adaptive strategies (uniform random, grid, LHS, Halton, Hammersley, and Sobol) and the adaptive RAD strategy are compared for the forward Burgers' equation, and the point distributions of the non-adaptive methods are shown in Fig.~\ref{fig:fig31}. All runs use the same 8-hidden-layer CV-MLP with 20 neurons per layer and the CGELU activation, $\mathcal{N}_f = 10{,}000$, $\mathcal{N}_u = 100$, and 5000 Adam iterations followed by 5000 L-BFGS iterations.

The results are summarized in Table~\ref{tab:tab12} and Fig.~\ref{fig:fig30}. Within this fixed training budget, random sampling gives the lowest error, the Sobol, Hammersley, LHS, and Halton sequences give errors of the same order, and grid and RAD sampling give considerably larger errors. LHS was adopted as the default strategy because, after full convergence, it provided the best agreement for the Burgers' equation and stable performance across the other benchmarks; with this choice, the Burgers' equation attains a final relative $L_2$ error of $6.73\times10^{-5}$.

\begin{figure}[pos=htbp]
\centering
\includegraphics[width=\linewidth]{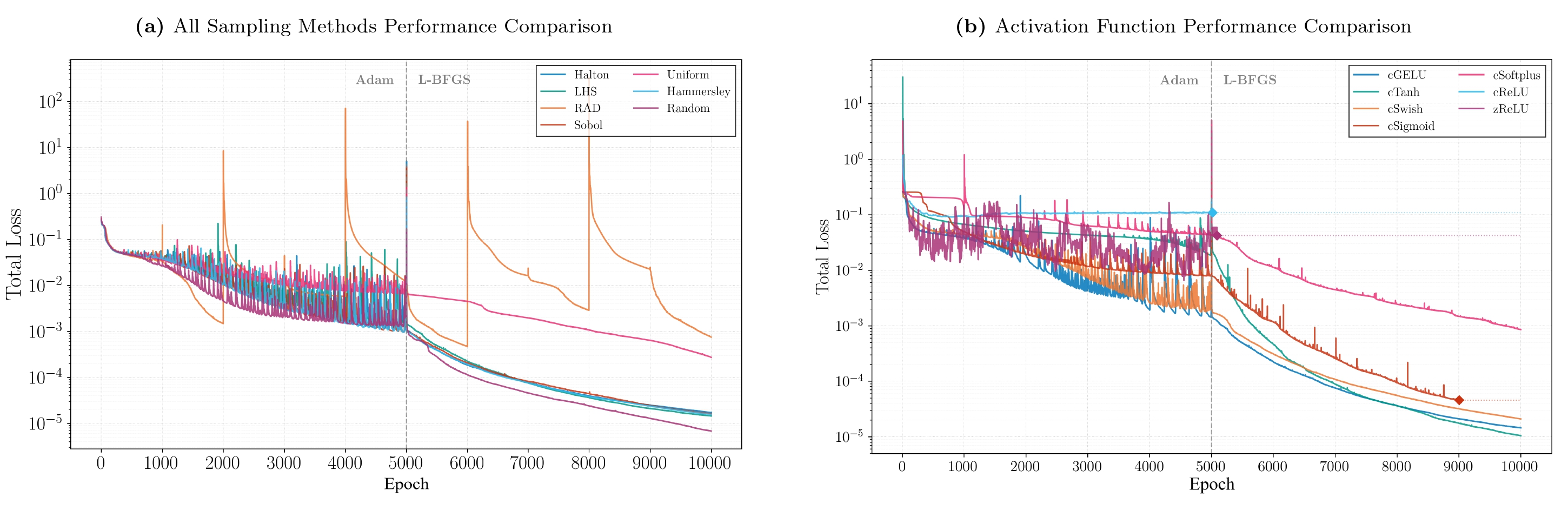}
\caption{
Comparison of different split complex-valued activation functions and collocation point sampling strategies employed in the proposed SCV-PINN framework. The figure illustrates the influence of activation selection and sampling distributions on the training convergence and solution accuracy.
}
\label{fig:fig30}
\end{figure}

\begin{figure}[pos=htbp]
\centering
\includegraphics[width=\textwidth]{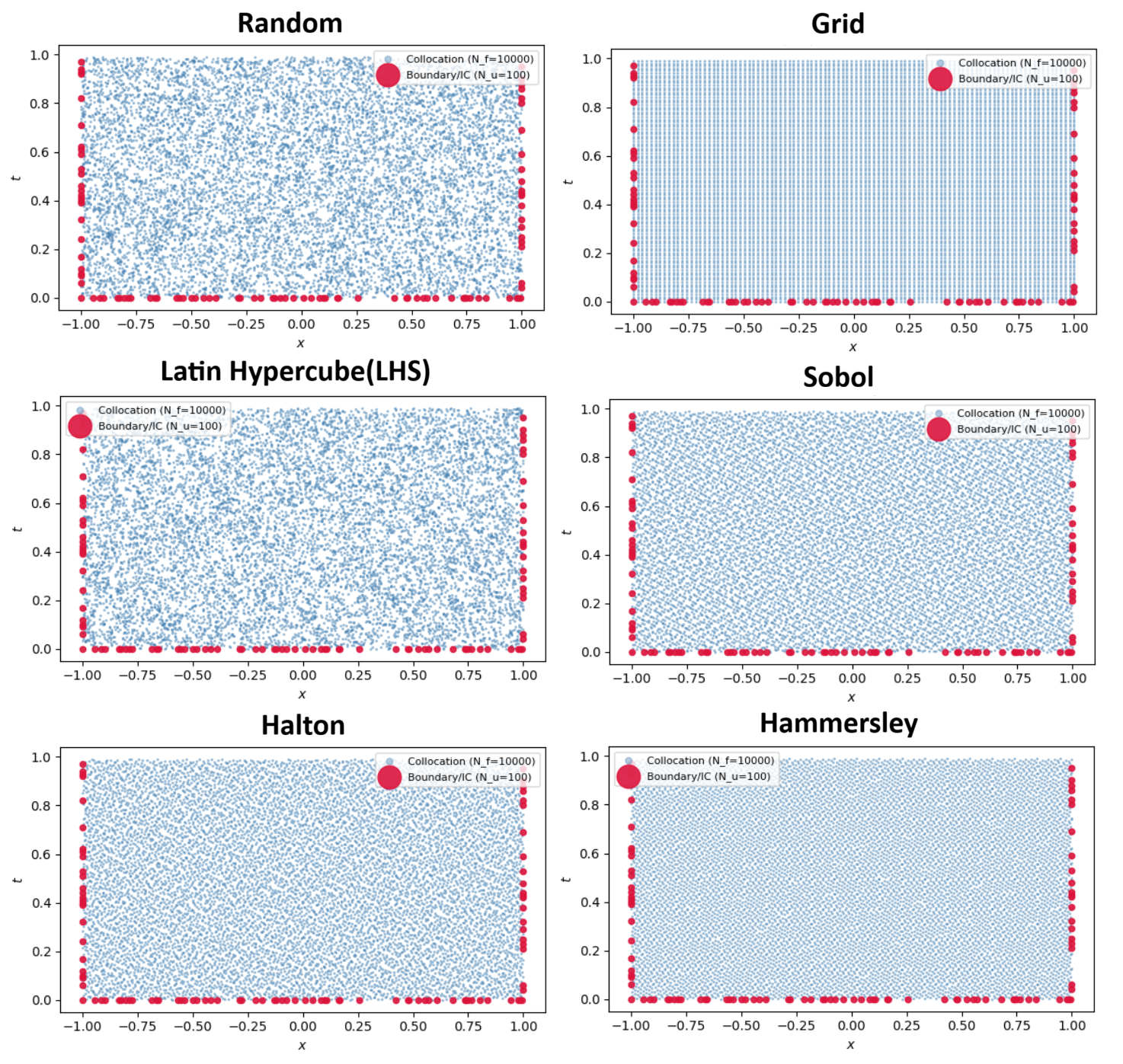}
\caption{
Distribution of collocation and initial--boundary training points for the Burgers' forward problem using different non-adaptive sampling strategies, including random, grid, Latin hypercube sampling (LHS), Sobol, Halton, and Hammersley sequences. The figure illustrates the spatial coverage and point distribution characteristics of each sampling method within the SCV-PINN framework.
}
\label{fig:fig31}
\end{figure}

\section*{Declaration of Generative AI and AI-assisted Technologies in the Writing Process}

During the preparation of this manuscript, the authors used ChatGPT solely for language polishing and grammatical correction. After using this tool, the authors carefully reviewed, revised, and validated the content. The authors take full responsibility for the scientific integrity and content of the published article.

\section*{Author Declaration}
The authors have no conflicts of interest to disclose.

\section*{Data Availability}
The data that support the findings of this study are available from the corresponding author upon reasonable request.

\section*{Code Availability}

The source code associated with this work is currently maintained in a private GitHub repository for peer-review purposes and will be made publicly available upon publication of the manuscript.


\begin{thebibliography}{00}
\bibitem{goodfellow2016deep}
Goodfellow, I., Bengio, Y., Courville, A. and Bengio, Y., 2016. Deep learning (Vol. 1, No. 2, pp. 1-800). Cambridge: MIT press.

\bibitem{bishop2006pattern}Bishop, C.M. and Nasrabadi, N.M., 2006. Pattern recognition and machine learning (Vol. 4, No. 4, p. 738). New York: springer.

\bibitem{bishop2023deep}Bishop, C.M. and Bishop, H., 2023. Deep learning: Foundations and concepts. Springer Nature.

\bibitem{aggarwal2018neural}Aggarwal, C.C., 2018. Neural networks and deep learning (Vol. 10, No. 978, p. 3). Cham: Springer.

\bibitem{strang2019linear}Strang, G., 2019. Linear algebra and learning from data (Vol. 4). Cambridge: Wellesley-Cambridge Press.

\bibitem{zhang2016understanding}Zhang, C., Bengio, S., Hardt, M., Recht, B. and Vinyals, O., 2016. Understanding deep learning requires rethinking generalization. arXiv preprint arXiv:1611.03530.

\bibitem{krantz1999handbook}Krantz, S.G., Kress, S. and Kress, R., 1999. Handbook of complex variables (p. 290). Boston: Birkhäuser.

\bibitem{ponnusamy2006complex}Ponnusamy, S. and Silverman, H., 2006. Complex variables with applications. Boston, MA: Birkhäuser Boston.

\bibitem{zill2009first}Zill, D. and Shanahan, P., 2009. A first course in complex analysis with applications. Jones \& Bartlett Learning.

\bibitem{suresh2013supervised}Suresh, S., Sundararajan, N. and Savitha, R., 2013. Supervised learning with complex-valued neural networks (Vol. 48). Berlin: Springer.

\bibitem{guberman2016complex}Guberman, N., 2016. On complex valued convolutional neural networks. arXiv preprint arXiv:1602.09046.

\bibitem{barrachina2023theory}Barrachina, J.A., Ren, C., Vieillard, G., Morisseau, C. and Ovarlez, J.P., 2023. Theory and implementation of complex-valued neural networks. arXiv preprint arXiv:2302.08286.

\bibitem{abdalla2023complex}Abdalla, R., 2023. Complex-valued Neural Networks--Theory and Analysis. arXiv preprint arXiv:2312.06087.

\bibitem{lee2022complex}Lee, C., Hasegawa, H. and Gao, S., 2022. Complex-valued neural networks: A comprehensive survey. IEEE/CAA Journal of Automatica Sinica, 9(8), pp.1406-1426.

\bibitem{bassey2021survey}Bassey, J., Qian, L. and Li, X., 2021. A survey of complex-valued neural networks. arXiv preprint arXiv:2101.12249.

\bibitem{fuchs2021complex}Fuchs, A., Rock, J., Toth, M., Meissner, P. and Pernkopf, F., 2021, May. Complex-valued convolutional neural networks for enhanced radar signal denoising and interference mitigation. In 2021 IEEE Radar Conference (RadarConf21) (pp. 1-6). IEEE.

\bibitem{scarnati2021complex}Scarnati, T. and Lewis, B., 2021, May. Complex-valued neural networks for synthetic aperture radar image classification. In 2021 IEEE Radar Conference (RadarConf21) (pp. 1-6). IEEE.

\bibitem{choi2018phase}Choi, H.S., Kim, J.H., Huh, J., Kim, A., Ha, J.W. and Lee, K., 2018, September. Phase-aware speech enhancement with deep complex u-net. In International Conference on Learning Representations.

\bibitem{fink2014predictin}Fink, O., Zio, E. and Weidmann, U., 2014. Predicting component reliability and level of degradation with complex-valued neural networks. Reliability Engineering \& System Safety, 121, pp.198-206.

\bibitem{trabelsi2017deep}Trabelsi, C., Bilaniuk, O., Zhang, Y., Serdyuk, D., Subramanian, S., Santos, J.F., Mehri, S., Rostamzadeh, N., Bengio, Y. and Pal, C.J., 2017. Deep complex networks. arXiv preprint arXiv:1705.09792.

\bibitem{hirose2006complex}Hirose, A., 2006. Complex-valued neural networks. Berlin, Heidelberg: Springer Berlin Heidelberg.

\bibitem{reichert2013neuronal}Reichert, D.P. and Serre, T., 2013. Neuronal synchrony in complex-valued deep networks. arXiv preprint arXiv:1312.6115.

\bibitem{nitta2003solving}Nitta, T., 2003. Solving the XOR problem and the detection of symmetry using a single complex-valued neuron. Neural Networks, 16(8), pp.1101-1105.

\bibitem{nitta2003inherent}Nitta, T., 2003. On the inherent property of the decision boundary in complex-valued neural networks. Neurocomputing, 50, pp.291-303.

\bibitem{nitta2000analysis}Nitta, T., 2000. An analysis of the fundamental structure of complex-valued neurons. Neural Processing Letters, 12(3), pp.239-246.

\bibitem{tanaka2013complex}Tanaka, G., 2013. Complex-valued neural networks: Advances and applications [book review]. ieee Computational intelligenCe magazine, 8(2), pp.77-79.

\bibitem{fischer2002appendix}Fischer, R.F., 2002. Appendix a: Wirtinger calculus.

\bibitem{wirtinger1927formalen}Wirtinger, W., 1927. Zur formalen theorie der funktionen von mehr komplexen veränderlichen. Mathematische Annalen, 97(1), pp.357-375.

\bibitem{amin2013learning}Amin, M.F., Murase, K. and Hirose, A., 2013. Learning algorithms in complex-valued neural networks using Wirtinger calculus. Complex-Valued Neural Networks: Advances and Applications, pp.75-102.

\bibitem{amin2011wirtinger}Amin, M.F., Amin, M.I., Al-Nuaimi, A.Y.H. and Murase, K., 2011, November. Wirtinger calculus based gradient descent and Levenberg-Marquardt learning algorithms in complex-valued neural networks. In International Conference on Neural Information Processing (pp. 550-559). Berlin, Heidelberg: Springer Berlin Heidelberg.

\bibitem{hammad2024comprehensive}Hammad, M.M., 2024. Comprehensive survey of complex-valued neural networks: Insights into backpropagation and activation functions. arXiv preprint arXiv:2407.19258.

\bibitem{raissi2019physics}Raissi, M., Perdikaris, P. and Karniadakis, G.E., 2019. Physics-informed neural networks: A deep learning framework for solving forward and inverse problems involving nonlinear partial differential equations. Journal of Computational physics, 378, pp.686-707.

\bibitem{karniadakis2021physics}Karniadakis, G.E., Kevrekidis, I.G., Lu, L., Perdikaris, P., Wang, S. and Yang, L., 2021. Physics-informed machine learning. Nature Reviews Physics, 3(6), pp.422-440.

\bibitem{cuomo2022scientific}Cuomo, S., Di Cola, V.S., Giampaolo, F., Rozza, G., Raissi, M. and Piccialli, F., 2022. Scientific machine learning through physics–informed neural networks: Where we are and what’s next. Journal of Scientific Computing, 92(3), p.88.

\bibitem{raissi2020hidden}Raissi, M., Yazdani, A. and Karniadakis, G.E., 2020. Hidden fluid mechanics: Learning velocity and pressure fields from flow visualizations. Science, 367(6481), pp.1026-1030.

\bibitem{sun2020surrogate}Sun, L., Gao, H., Pan, S. and Wang, J.X., 2020. Surrogate modeling for fluid flows based on physics-constrained deep learning without simulation data. Computer Methods in Applied Mechanics and Engineering, 361, p.112732.

\bibitem{raissi2019deep}Raissi, M., Wang, Z., Triantafyllou, M.S. and Karniadakis, G.E., 2019. Deep learning of vortex-induced vibrations. Journal of fluid mechanics, 861, pp.119-137.

\bibitem{jin2021nsfnets}Jin, X., Cai, S., Li, H. and Karniadakis, G.E., 2021. NSFnets (Navier-Stokes flow nets): Physics-informed neural networks for the incompressible Navier-Stokes equations. Journal of Computational Physics, 426, p.109951.

\bibitem{kissas2020machine}Kissas, G., Yang, Y., Hwuang, E., Witschey, W.R., Detre, J.A. and Perdikaris, P., 2020. Machine learning in cardiovascular flows modeling: Predicting arterial blood pressure from non-invasive 4D flow MRI data using physics-informed neural networks. Computer methods in applied mechanics and engineering, 358, p.112623.

\bibitem{liu2019multi}Liu, D. and Wang, Y., 2019. Multi-fidelity physics-constrained neural network and its application in materials modeling. Journal of Mechanical Design, 141(12), p.121403.

\bibitem{yang2019adversarial}Yang, Y. and Perdikaris, P., 2019. Adversarial uncertainty quantification in physics-informed neural networks. Journal of Computational Physics, 394, pp.136-152.

\bibitem{zhu2019physics}Zhu, Y., Zabaras, N., Koutsourelakis, P.S. and Perdikaris, P., 2019. Physics-constrained deep learning for high-dimensional surrogate modeling and uncertainty quantification without labeled data. Journal of computational physics, 394, pp.56-81.

\bibitem{yang2021b}Yang, L., Meng, X. and Karniadakis, G.E., 2021. B-PINNs: Bayesian physics-informed neural networks for forward and inverse PDE problems with noisy data. Journal of Computational Physics, 425, p.109913.

\bibitem{han2018solving}Han, J., Jentzen, A. and E, W., 2018. Solving high-dimensional partial differential equations using deep learning. Proceedings of the National Academy of Sciences, 115(34), pp.8505-8510.

\bibitem{pang2019fpinns}Pang, G., Lu, L. and Karniadakis, G.E., 2019. fPINNs: Fractional physics-informed neural networks. SIAM Journal on Scientific Computing, 41(4), pp.A2603-A2626.

\bibitem{fuks2020limitations}Fuks, O. and Tchelepi, H.A., 2020. Limitations of physics informed machine learning for nonlinear two-phase transport in porous media. Journal of Machine Learning for Modeling and Computing, 1(1).

\bibitem{krishnapriyan2021characterizing}Krishnapriyan, A., Gholami, A., Zhe, S., Kirby, R. and Mahoney, M.W., 2021. Characterizing possible failure modes in physics-informed neural networks. Advances in neural information processing systems, 34, pp.26548-26560.

\bibitem{wang2021understanding}Wang, S., Teng, Y. and Perdikaris, P., 2021. Understanding and mitigating gradient flow pathologies in physics-informed neural networks. SIAM Journal on Scientific Computing, 43(5), pp.A3055-A3081.

\bibitem{wang2022and}Wang, S., Yu, X. and Perdikaris, P., 2022. When and why PINNs fail to train: A neural tangent kernel perspective. Journal of Computational Physics, 449, p.110768.

\bibitem{rahaman2019spectral}Rahaman, N., Baratin, A., Arpit, D., Draxler, F., Lin, M., Hamprecht, F., Bengio, Y. and Courville, A., 2019, May. On the spectral bias of neural networks. In International conference on machine learning (pp. 5301-5310). PMLR.

\bibitem{cao2019towards}Cao, Y., Fang, Z., Wu, Y., Zhou, D.X. and Gu, Q., 2019. Towards understanding the spectral bias of deep learning. arXiv preprint arXiv:1912.01198.

\bibitem{tancik2020fourier}Tancik, M., Srinivasan, P., Mildenhall, B., Fridovich-Keil, S., Raghavan, N., Singhal, U., Ramamoorthi, R., Barron, J. and Ng, R., 2020. Fourier features let networks learn high frequency functions in low dimensional domains. Advances in neural information processing systems, 33, pp.7537-7547.

\bibitem{basri2020frequency}Basri, R., Galun, M., Geifman, A., Jacobs, D., Kasten, Y. and Kritchman, S., 2020, November. Frequency bias in neural networks for input of non-uniform density. In International conference on machine learning (pp. 685-694). PMLR.

\bibitem{jacot2018neural}Jacot, A., Gabriel, F. and Hongler, C., 2018. Neural tangent kernel: Convergence and generalization in neural networks. Advances in neural information processing systems, 31.

\bibitem{mcclenny2023self}McClenny, L.D. and Braga-Neto, U.M., 2023. Self-adaptive physics-informed neural networks. Journal of Computational Physics, 474, p.111722.

\bibitem{song2024loss}Song, Y., Wang, H., Yang, H., Taccari, M.L. and Chen, X., 2024. Loss-attentional physics-informed neural networks. Journal of Computational Physics, 501, p.112781.

\bibitem{mattey2022novel}Mattey, R. and Ghosh, S., 2022. A novel sequential method to train physics informed neural networks for Allen Cahn and Cahn Hilliard equations. Computer Methods in Applied Mechanics and Engineering, 390, p.114474.

\bibitem{sundar2025sequential}Sundar, R., Lucor, D. and Sarkar, S., 2025. Sequential learning based PINNs to overcome temporal domain complexities in unsteady flow past flapping wings. Journal of Fluids and Structures, 139, p.104421.

\bibitem{zhao2023pinnsformer}Zhao, Z., Ding, X. and Prakash, B.A., 2023. Pinnsformer: A transformer-based framework for physics-informed neural networks. arXiv preprint arXiv:2307.11833.

\bibitem{zhu2026physicssolver}Zhu, Z., Huang, Y. and Liu, L., 2026. Physicssolver: Transformer-enhanced physics-informed neural networks for forward and forecasting problems in partial differential equations. Journal of Computational and Applied Mathematics, 473, p.116900.

\bibitem{barman2026physicsformer}Barman, B., Chatterjee, D. and Ray, R.K., 2026. PhysicsFormer: An Efficient and Fast Attention-Based Physics-Informed Neural Network for Solving Incompressible Navier-Stokes Equations. arXiv preprint arXiv:2601.03613.

\bibitem{barman2026efficient}Barman, B. and Ray, R.K., 2026. An Efficient Wavelet-based Physics Informed Residual Neural Networks for Flow Field Reconstruction with Extremely Sparse Data. arXiv preprint arXiv:2601.18848.

\bibitem{wu2023comprehensive}Wu, C., Zhu, M., Tan, Q., Kartha, Y. and Lu, L., 2023. A comprehensive study of non-adaptive and residual-based adaptive sampling for physics-informed neural networks. Computer Methods in Applied Mechanics and Engineering, 403, p.115671.

\bibitem{kingma2014adam}Kingma, D.P. and Ba, J., 2014. Adam: A method for stochastic optimization. arXiv preprint arXiv:1412.6980.

\bibitem{liu1989limited}Liu, D.C. and Nocedal, J., 1989. On the limited memory BFGS method for large scale optimization. Mathematical programming, 45(1), pp.503-528.

\bibitem{si2026complex}Si, C., Yan, M., Li, X. and Xia, Z., 2026. Complex physics-informed neural network. Journal of Computational Physics, p.114713.

\bibitem{mohuț2026towards}Mohuț, A.I. and Popa, C.A., 2026. Towards Stable Training of Complex-Valued Physics-Informed Neural Networks: A Holomorphic Initialization Approach. Mathematics, 14(3), p.435.

\bibitem{zhang2025complex}Zhang, L., Du, M., Bai, X., Chen, Y. and Zhang, D., 2025. Complex-valued physics-informed machine learning for efficient solving of quintic nonlinear Schrödinger equations. Physical Review Research, 7(1), p.013164.

\bibitem{yang2007sensitivity}Yang, S.S., Ho, C.L. and Siu, S., 2007. Sensitivity analysis of the split-complex valued multilayer perceptron due to the errors of the iid inputs and weights. IEEE transactions on neural networks, 18(5), pp.1280-1293.

\bibitem{georgiou1992complex}Georgiou, G.M. and Koutsougeras, C., 1992. Complex domain backpropagation. IEEE transactions on Circuits and systems II: analog and digital signal processing, 39(5), pp.330-334.

\bibitem{hirose1992proposal}Hirose, A., 1992, June. Proposal of fully complex-valued neural networks. In [Proceedings 1992] IJCNN International Joint Conference on Neural Networks (Vol. 4, pp. 152-157). IEEE.

\bibitem{nocedal2006numerical}Nocedal, J. and Wright, S.J., 2006. Numerical optimization. New York, NY: Springer New York.

\bibitem{lu2021deepxde}Lu, L., Meng, X., Mao, Z. and Karniadakis, G.E., 2021. DeepXDE: A deep learning library for solving differential equations. SIAM review, 63(1), pp.208-228.

\bibitem{segur1973korteweg}Segur, H., 1973. The Korteweg-de Vries equation and water waves. Solutions of the equation. Part 1. Journal of Fluid Mechanics, 59(4), pp.721-736.

\bibitem{ethier1994exact}Ethier, C.R. and Steinman, D.A., 1994. Exact fully 3D Navier–Stokes solutions for benchmarking. International Journal for Numerical Methods in Fluids, 19(5), pp.369-375.

\bibitem{urban2025unveiling}Urbán, J.F., Stefanou, P. and Pons, J.A., 2025. Unveiling the optimization process of physics informed neural networks: How accurate and competitive can PINNs be?. Journal of Computational Physics, 523, p.113656.

\end{thebibliography}
\end{document}